\documentclass[11pt,a4paper,oneside]{extarticle}
\PassOptionsToPackage{unicode=true}{hyperref}
\PassOptionsToPackage{hyphens}{url}
\PassOptionsToPackage{dvipsnames,svgnames*,x11names*}{xcolor}
\usepackage{lmodern}
\usepackage{amssymb,amsmath}
\usepackage{ifpdf,ifxetex,ifluatex}
\usepackage{fixltx2e} 
\ifnum 0\ifxetex 1\fi\ifluatex 1\fi=0 
  \usepackage[T1]{fontenc}
  \usepackage[utf8]{inputenc}
  \usepackage{textcomp} 
\else 
  \usepackage{unicode-math}
  \defaultfontfeatures{Ligatures=TeX,Scale=MatchLowercase}
\fi
\IfFileExists{upquote.sty}{\usepackage{upquote}}{}
\IfFileExists{microtype.sty}{%
\usepackage[]{microtype}
\UseMicrotypeSet[protrusion]{basicmath} 
}{}
\usepackage{indentfirst}
\usepackage{xcolor}
\usepackage{hyperref}
\hypersetup{
            pdftitle={Exceptionally Monadic Error Handling},
            pdfauthor={Jan Malakhovski},
            breaklinks=true,
            }
\ifnum 0\ifpdf 1\fi\ifxetex 1\fi\ifluatex 1\fi=0 
\hypersetup{
            colorlinks=true,
            linkcolor=Maroon,
            citecolor=Blue,
            urlcolor=Blue,
            }
\else 
\hypersetup{
            }
\fi
\usepackage[margin=2cm]{geometry}
\usepackage{color}
\usepackage{fancyvrb}
\newcommand{\VerbBar}{|}
\newcommand{\VERB}{\Verb[commandchars=\\\{\}]}
\DefineVerbatimEnvironment{Highlighting}{Verbatim}{commandchars=\\\{\}}
\newenvironment{Shaded}{}{}

\newcommand{\CommentTok}[1]{\textcolor[rgb]{0.38,0.63,0.69}{\textit{#1}}}

\newcommand{\ControlFlowTok}[1]{\textcolor[rgb]{0.00,0.44,0.13}{\textbf{#1}}}
\newcommand{\DataTypeTok}[1]{\textcolor[rgb]{0.56,0.13,0.00}{#1}}
\newcommand{\DecValTok}[1]{\textcolor[rgb]{0.25,0.63,0.44}{#1}}

\newcommand{\FunctionTok}[1]{\textcolor[rgb]{0.02,0.16,0.49}{#1}}

\newcommand{\KeywordTok}[1]{\textcolor[rgb]{0.00,0.44,0.13}{\textbf{#1}}}
\newcommand{\NormalTok}[1]{#1}

\newcommand{\OtherTok}[1]{\textcolor[rgb]{0.00,0.44,0.13}{#1}}

\newcommand{\StringTok}[1]{\textcolor[rgb]{0.25,0.44,0.63}{#1}}

\providecommand{\tightlist}{%
  \setlength{\itemsep}{0pt}\setlength{\parskip}{0pt}}
\ifx\paragraph\undefined\else
\let\oldparagraph\paragraph
\renewcommand{\paragraph}[1]{\oldparagraph{#1}\mbox{}}
\fi
\ifx\subparagraph\undefined\else
\let\oldsubparagraph\subparagraph
\renewcommand{\subparagraph}[1]{\oldsubparagraph{#1}\mbox{}}
\fi

\makeatletter
\def\fps@figure{htbp}
\makeatother

\usepackage[style=numeric-comp,backend=biber,natbib=true,backref=true,sorting=none]{biblatex}
\usepackage{attachfile2}

\def\ignore#1{}

\usepackage{enumitem}

\usepackage{amsthm}
\newtheorem{definition}{Definition}
\newtheorem{lemma}{Lemma}
\newtheorem{theorem}{Theorem}
\newtheorem{quasitheorem}[theorem]{``Theorem''}
\newtheorem{remark}{Remark}

\usepackage{cleveref}
\usepackage{footnotebackref}
\crefformat{footnote}{#2\footnotemark[#1]#3}
\usepackage[affil-it]{authblk}
\author{Jan Malakhovski\thanks{\href{mailto:papers@oxij.org}{papers@oxij.org}; preferably with paper title in the subject line}}
\affil{IRIT, University of Toulouse-3 and ITMO University}

\title{Exceptionally Monadic Error Handling\\\small Looking at $bind$ and squinting really hard}

\date{February 2014 - October 2018}

\begin{document}
\maketitle

\begin{abstract}

We notice that the type of
\VERB|\FunctionTok{catch}\OtherTok{ ::}\NormalTok{ c a }\OtherTok{->}\NormalTok{ (e }\OtherTok{->}\NormalTok{ c a) }\OtherTok{->}\NormalTok{ c a}|
operator is a special case of monadic \VERB|\NormalTok{bind}| operator
\VERB|\OtherTok{(>>=) ::}\NormalTok{ m a }\OtherTok{->}\NormalTok{ (a }\OtherTok{->}\NormalTok{ m b) }\OtherTok{->}\NormalTok{ m b}|,
the semantics (surprisingly) matches, and this observation has many
interesting consequences.

For instance, the reader is probably aware that the monadic essence of
the \VERB|\NormalTok{(}\FunctionTok{>>=}\NormalTok{)}| operator of the
error monad \(\lambda A.E \lor A\) is to behave like identity monad for
"normal" values and to stop on "errors". The unappreciated fact is that
handling of said "errors" with a \VERB|\FunctionTok{catch}| operator of
the "flipped" "conjoined" error monad \(\lambda E.E \lor A\) is, too, a
monadic computation that treats still unhandled "errors" as "normal"
values and stops when an "error" is finally handled.

We show that for an appropriately indexed type of computations such a
"conjoined" structure naturally follows from the conventional
operational semantics of \VERB|\NormalTok{throw}| and
\VERB|\FunctionTok{catch}| operators. Consequently, we show that this
structure uniformly generalizes \emph{all} conventional monadic error
handling mechanisms we are aware of. We also demonstrate several more
interesting instances of this structure of which at least bi-indexed
monadic parser combinators and conventional exceptions implemented via
continuations have immediate practical applications. Finally, we notice
that these observations provide surprising perspectives on error
handling in general and point to a largely unexplored trail in
programming language design space.

\end{abstract}

\small\tableofcontents\normalsize\newpage

\hypertarget{extended-abstract}{%
\section{Extended Abstract}\label{extended-abstract}}

\label{sec:extabstract}

In this article we shall use Haskell programming language extensively
for the purposes of precise expression of thought (including Haskell
type class names for the names of the respective algebraic structures
where appropriate, e.g. "\VERB|\DataTypeTok{Monad}|" instead of
"monad").

\begin{itemize}
\item
  We note that the types of

\begin{Shaded}
\begin{Highlighting}[]
\OtherTok{throw ::}\NormalTok{ e }\OtherTok{->}\NormalTok{ c a}
\FunctionTok{catch}\OtherTok{ ::}\NormalTok{ c a }\OtherTok{->}\NormalTok{ (e }\OtherTok{->}\NormalTok{ c a) }\OtherTok{->}\NormalTok{ c a}
\end{Highlighting}
\end{Shaded}

  operators are special cases of \VERB|\DataTypeTok{Monad}|ic
  \VERB|\FunctionTok{return}| and
  \VERB|\NormalTok{(}\FunctionTok{>>=}\NormalTok{)}|
  (\VERB|\NormalTok{bind}|) operators

\begin{Shaded}
\begin{Highlighting}[]
\FunctionTok{return}\OtherTok{ ::}\NormalTok{ a }\OtherTok{->}\NormalTok{ m a}
\OtherTok{(>>=)  ::}\NormalTok{ m a }\OtherTok{->}\NormalTok{ (a }\OtherTok{->}\NormalTok{ m b) }\OtherTok{->}\NormalTok{ m b}
\end{Highlighting}
\end{Shaded}

  (substitute \([a \mapsto e, m \mapsto \lambda\_.c~a]\) into their
  types, see \cref{sec:init,sec:type-of-catch}).
\item
  Hence, a type of computations \VERB|\NormalTok{c e a}| with two
  indexes where \VERB|\NormalTok{e}| signifies a type of errors and
  \VERB|\NormalTok{a}| signifies a type of values can be made a
  \VERB|\DataTypeTok{Monad}| twice: once for \VERB|\NormalTok{e}| and
  once for \VERB|\NormalTok{a}|.

\begin{Shaded}
\begin{Highlighting}[]
\KeywordTok{class} \DataTypeTok{ConjoinedMonads}\NormalTok{ c }\KeywordTok{where}
\OtherTok{  return ::}\NormalTok{ a }\OtherTok{->}\NormalTok{ c e a}
\OtherTok{  (>>=)  ::}\NormalTok{ c e a }\OtherTok{->}\NormalTok{ (a }\OtherTok{->}\NormalTok{ c e b) }\OtherTok{->}\NormalTok{ c e b}

\OtherTok{  throw  ::}\NormalTok{ e }\OtherTok{->}\NormalTok{ c e a}
\OtherTok{  catch  ::}\NormalTok{ c e a }\OtherTok{->}\NormalTok{ (e }\OtherTok{->}\NormalTok{ c f a) }\OtherTok{->}\NormalTok{ c f a}
\end{Highlighting}
\end{Shaded}

  Moreover, for such a structure \VERB|\NormalTok{throw}| is a left zero
  for \VERB|\NormalTok{(}\FunctionTok{>>=}\NormalTok{)}| and
  \VERB|\FunctionTok{return}| is a left zero for
  \VERB|\FunctionTok{catch}| (see
  \cref{sec:conjoinedly-monadic,sec:logical}).
\item
  We prove that the type of the above \VERB|\FunctionTok{catch}| is most
  general type for any \VERB|\DataTypeTok{Monad}|ic structure
  \VERB|\NormalTok{\textbackslash{}a }\OtherTok{->}\NormalTok{ c e a}|
  with additional \VERB|\NormalTok{throw}| and
  \VERB|\FunctionTok{catch}| operators satisfying conventional
  operational semantics (via simple unification of types for several
  equations that follow from semantics of said operators, see
  \cref{sec:type-of-catch}). Or, dually, we prove that
  \VERB|\NormalTok{(}\FunctionTok{>>=}\NormalTok{)}| has the most
  general type for expressing sequential computations for
  \VERB|\DataTypeTok{Monad}|ic structure
  \VERB|\NormalTok{\textbackslash{}e }\OtherTok{->}\NormalTok{ c e a}|
  (with operators named \VERB|\NormalTok{throw}| and
  \VERB|\FunctionTok{catch}|) with additional
  \VERB|\FunctionTok{return}| and
  \VERB|\NormalTok{(}\FunctionTok{>>=}\NormalTok{)}| operators
  satisfying conventional operational semantics (see
  footnote~\ref{fn:its-dual}).
\item
  Substituting a \VERB|\DataTypeTok{Const}|ant
  \VERB|\DataTypeTok{Functor}| for \VERB|\NormalTok{c}| into
  \VERB|\DataTypeTok{ConjoinedMonads}| above (i.e., fixing the type of
  errors) produces the definition of \VERB|\DataTypeTok{MonadError}|,
  and, with some equivalent redefinitions,
  \VERB|\DataTypeTok{MonadCatch}| (see \cref{sec:instances:constant}).
  Similarly, \VERB|\DataTypeTok{IO}| with similar redefinitions and with
  the usual caveats of \cref{rem:io-caveats} is a
  \VERB|\DataTypeTok{ConjoinedMonads}| instance too (see
  \cref{sec:instances:io}).
\item
  \VERB|\DataTypeTok{ExceptT}| (\cref{sec:instances:either}) and some
  other lesser known and potentially novel concrete structures (see all
  sections with "Instance:" in the title, most interestingly,
  \cref{sec:instances:throw-catch-cc}) have operators of such types and
  their semantics matches (or they can be redefined in an equivalent way
  such that the core part of the resulting structure then matches) the
  semantics of \VERB|\DataTypeTok{Monad}| exactly.
\item
  \VERB|\DataTypeTok{Monad}| type class has a well-known "fish"
  representation where "\VERB|\NormalTok{bind}|"
  \VERB|\NormalTok{(}\FunctionTok{>>=}\NormalTok{)}| operator is
  replaced by "\VERB|\NormalTok{fish}|" operator

\begin{Shaded}
\begin{Highlighting}[]
\OtherTok{(>=>) ::}\NormalTok{ (a }\OtherTok{->}\NormalTok{ m b) }\OtherTok{->}\NormalTok{ (b }\OtherTok{->}\NormalTok{ m c) }\OtherTok{->}\NormalTok{ (a }\OtherTok{->}\NormalTok{ m c)}
\end{Highlighting}
\end{Shaded}

  and \VERB|\DataTypeTok{Monad}| laws are just monoidal laws.

  Hence, all those structures can be seen as a pairs of monoids over
  bi-indexed types with identity elements for respective
  \VERB|\NormalTok{bind}|s as left zeros for conjoined
  \VERB|\NormalTok{bind}|s (\cref{sec:conjoinedly-monadic}). We find
  this symmetry to be hypnotic and generalize it in
  \cref{sec:applicatives}.
\item
  The answer to "Why didn't anyone notice this already?" seems to be
  that this structure cannot be expressed well in Haskell (see
  \cref{sec:encodings}).
\item
  Meanwhile, it has at least several practically useful instances:

  \begin{itemize}
  \item
    Parser combinators that are precise about errors they produce and
    that reuse common \VERB|\DataTypeTok{Monad}|ic combinators for both
    parsing and handling of errors. For instance, the type of
    \VERB|\NormalTok{many}| for such a parser combinator guarantees that
    it cannot throw any errors

\begin{Shaded}
\begin{Highlighting}[]
\OtherTok{many ::}\NormalTok{ c e a }\OtherTok{->}\NormalTok{ c f [a]}
\end{Highlighting}
\end{Shaded}

    (since \VERB|\NormalTok{f}| can be anything, it cannot be anything
    in particular) and

\begin{Shaded}
\begin{Highlighting}[]
\OtherTok{choice ::}\NormalTok{ [c e a] }\OtherTok{->}\NormalTok{ c e a}
\end{Highlighting}
\end{Shaded}

    is an instance of \VERB|\NormalTok{foldM}| (see
    \cref{sec:instances:parser-combinators}).
  \item
    Conventional exceptions expressed using \VERB|\DataTypeTok{Reader}|
    \VERB|\DataTypeTok{Monad}| and second-rank \VERB|\NormalTok{callCC}|
    (the whole idea of which seems to be novel, see
    \cref{sec:instances:throw-catch-cc}).
  \item
    Error-explicit \VERB|\DataTypeTok{IO}| (\cref{sec:instances:eio}),
    the latter and similar structures with similar motivation were
    proposed before, but they did not use the fact that their "other
    half" is a \VERB|\DataTypeTok{Monad}| too.
  \end{itemize}
\end{itemize}

Every item on the above list, to our best knowledge, is a headline
contribution.

\hypertarget{preliminaries}{%
\section{Preliminaries}\label{preliminaries}}

\label{sec:preliminaries}

Most of the results of this paper are
\textbf{\textbf{language-agnostic}} and can be applied (if not straight
to practice, then at least to inform design choices) to any programming
language (that permits at least two computationally distinguishable
program states and some kind of dynamic control flow control) as our
definition of an "\emph{error}" in "error handling" is just "an abnormal
program state causing execution of an abnormal code path" and both
"abnormal"s can be arbitrarily defined (see footnote~\ref{fn:terms}).

While most of our results are applicable to any programming language, we
need \emph{some} language to express them in and Haskell seems to be the
most natural choice to host this discussion since

\begin{itemize}
\tightlist
\item
  most of the cited literature uses Haskell or some variant of ML;
\item
  it has the largest number of error handling mechanisms in active use
  of all the programming languages we are aware of;
\item
  as a consequence, most other programming languages implement a subset
  of Haskell's enormous library of error handling mechanisms;
\item
  while it is not ideal for our purposes (Haskell cannot properly
  express the main result and the improper encoding of the main result
  is not particularly convenient, see \cref{sec:encodings}), it is
  expressive enough to show how a convenient encoding could have been
  implemented in theory;
\item
  it is surprisingly popular for an "academic" language.
\end{itemize}

Using Haskell also allows this paper to be encoded as a set of Literate
Haskell programs in a single Emacs Org-Mode
tree~\cite{OrgMode,Schulte:2011:MLCELPRR}.\footnote{The source code is
  available at \url{https://oxij.org/paper/ExceptionallyMonadic/}.
  \ifnum 0\ifpdf 1\fi\ifxetex 1\fi\ifluatex 1\fi=0 It also gets embedded
  into the PDF version of this article when the source gets compiled
  straight to PDF (i.e. via \texttt{pdftex}, \texttt{xetex}, or
  \texttt{luatex}, but not via \texttt{dvipdf}). Unfortunately, the file
  you are looking at was compiled using \texttt{dvipdf}. This also means
  that in this file URLs with line breaks would not be clickable as
  \texttt{dvipdf} generates incorrect PDF link boxes for them. Properly
  compiled version is available via the above link. \else
  It is also embedded straight into the PDF version of this article
  (click here \attachfile{article.org} or look for "attachments" in
  your PDF viewer).\fi}

Our preferred compiler is The Glorious Glasgow Haskell Compiler
(GHC)~\cite{GHC} version 8.2 as we shall use a number of its extensions
over Haskell 2010~\cite{Haskell2010} specification.

Readers unfamiliar with Haskell are advised to read through any tutorial
introduction into Haskell at least until they start feeling like Haskell
is just a syntax for school-level arithmetic with user-definable
functions, lambdas, types, algebraic data types and type classes. After
that it is recommended to look over
Typeclassopedia~\cite{HaskellWiki:Typeclassopedia}, Diehl's
web-page~\cite{Diehl:2016:WIW}, the table of contents (just the list of
modules) of GHC's \texttt{base} package~\cite{Hackage:base4900}, and the
types and descriptions of functions from the
\VERB|\DataTypeTok{Prelude}| module of \texttt{base}.

The rest can be learned on-demand from
\cref{sec:tutorial:basic,sec:tutorial:non-basic} and cited
documentation.

\hypertarget{introduction}{%
\section{Introduction}\label{introduction}}

\begin{definition}

\label{dfn:error-handling}

Generally, when program encounters an "\emph{error}" all it can do is to
switch to an "\emph{exceptional}" execution path~\cite{Benton:2001:ES}.
The latter can then either encounter an "\emph{error}" itself or

\begin{enumerate}
\item
  \label{c:a} gracefully "\emph{terminate}" some part of the previous
  computation (including the whole program as a degenerate case) and
  continue (when there is something left to continue),
\item
  \label{c:b} "\emph{fix}" the "\emph{problem}" and resume the
  computation as if nothing has happened.
\end{enumerate}

\end{definition}

\emph{Error handling}\footnote{\label{fn:terms}Not a consensus term.
  Some people would disagree with this choice of a name as they would
  not consider some of our examples below to be about "errors". However,
  for the purposes of this article we opted into generalizing the term
  "\emph{error}" of "error handling" instead of inventing new
  terminology or appropriating terminology like "exceptions",
  "interrupts", "conditions" or "effects" that has other very specific
  uses. To see the problem with the conventional terminology consider
  how would you define "program encountered an error" formally and
  generally for \textbf{any} abstract interpreter (you can not). Now
  consider the case where an interpreter is a tower of interpreters
  interpreting one another. Clearly, what is an "\emph{error}" for one
  interpreter can be considered normal execution for the one below. A
  simple example of such a structure is the \VERB|\DataTypeTok{Maybe}|
  \VERB|\DataTypeTok{Monad}| discussed in \cref{sec:maybe} in which
  expressions using \VERB|\KeywordTok{do}|-syntax never consider
  \VERB|\DataTypeTok{Nothing}|s while handling of said
  \VERB|\DataTypeTok{Nothing}|s by the \VERB|\DataTypeTok{Monad}|ic
  \VERB|\NormalTok{(}\FunctionTok{>>=}\NormalTok{)}| operator is a
  completely ordinary \VERB|\KeywordTok{case}| for the underlying
  Haskell interpreter. Hence, in this article we consider anything that
  matches \cref{dfn:error-handling} to be about \emph{"error" handling}.
  If the reader still feels like disagreeing with our argument we advise
  mentally substituting every our use of "error" with something like "an
  abnormal program state causing execution of an abnormal code path"
  (where definitions of both "abnormal"s are interpreter-specific).} is
an algebraic subfield of the programming languages theory that studies
this sort of seemingly simple control structures.

Different substitutions for "\emph{error}", "\emph{exceptional}" and
"\emph{terminate}" into \cref{dfn:error-handling} variant~\ref{c:a} and
substitutions for "\emph{error}", "\emph{exceptional}", "\emph{fix}" and
"\emph{problem}" into \cref{dfn:error-handling} variant~\ref{c:b}
produce different error handling mechanisms. Some examples:

\begin{itemize}
\item
  Identity substitution for variant~\ref{c:a} gives programming with
  error codes, programming with algebraic data
  types~\cite{burstall-hope-80, Bailey:1985:HT} that encode errors,
  programming with algebraic data types with
  errors~\cite{ADJ76, Gogolla:1984:AOS} (not the same thing), exceptions
  in conventional programming
  languages~\cite{Goodenough:1975:EHI, Goodenough:1975:EHD,
   Goldberg:1983:SLI, Koenig:1990:EHC, Benton:2001:ES} (with so called
  "termination semantics"~\cite[16.6 Exception Handling: Resumption
   vs. Termination]{Stroustrup:DEC94}), error handling with
  monads~\cite{moggi-89, moggi-91, Wadler:1992:EFP,
   Swierstra:2008:DTL, Iborra:2010:ETE, Katsumata:2014:PEM}, monad
  transformers~\cite{Liang:1995:MTM, Benton:2002:ME,
   Hackage:transformers0520}, Scheme's and ML's
  \VERB|\KeywordTok{call/cc}|~\cite{Sperber:2010:RnRS}, and delimited
  \VERB|\NormalTok{callCC}|~\cite{Asai:2011:IPS, Kiselyov:2012:AAC,
   Hackage:transformers0520}.
\item
  Substituting "\emph{unparsable string}", "\emph{alternative}",
  "\emph{backtrack}" for variant~\ref{c:a} gives monadic parser
  combinators~\cite{Leijen:2001:PDS}.
\item
  Identity substitution for variant~\ref{c:b} gives error handling in
  languages with so called "resumption semantics"~\cite[16.6 Exception
   Handling: Resumption vs. Termination]{Stroustrup:DEC94} like, for
  instance, Common LISP~\cite{Pitman:2001:CHL} (\emph{condition
  handling}) and Smalltalk~\cite{Goldberg:1983:SLI}.
\item
  Substituting "\emph{effect}", "\emph{effect handler}", \emph{handle}",
  "\emph{it}" for variant~\ref{c:a} or~\ref{c:b} (depending on the
  details of the calculus) produces effect systems~\cite{Benton:2002:ME,
   Plotkin:2009:HAE, Brady:2013:PRA, Kammar:2013:HA, Kiselyov:2013:EEA,
   Kiselyov:2015:FMM} and effect systems based on modal logic with
  names~\cite{Nanevski:2004:FPNN, Nanevski:2005:MCEH}.
\item
  "\emph{System call}", "\emph{system call handler}", "\emph{handle}",
  "\emph{it}" for variant~\ref{c:b} produces conventional \emph{system
  calls}~\cite{IEEE:2001:ISR}.\footnote{Except in most UNIX-like
    operating systems system calls cannot call other system calls
    directly and have to use an equivalent kernel API instead.}\footnote{Indeed,
    algebraic effects from the point of view of an OS-developer are just
    properly typed system calls with nesting and modular handling.}
\item
  Substituting "\emph{signal}", "\emph{signal handler}",
  "\emph{handle}", "\emph{it}", "\emph{it}" for variant~\ref{c:b} gives
  hardware interrupts and POSIX signals~\cite{IEEE:2001:ISR}.\footnote{Indeed,
    POSIX signals and hardware interrupts are "system calls in reverse"
    (with some complications outside of the scope of this article):
    kernel and/or hardware raises and applications handle them.}
\end{itemize}

The first complication of the above scheme is the question of whenever
for a given error handling mechanism the "\emph{error}" raising operator

\begin{enumerate}
\item
  passes control to a statically selected (lexically closest or
  explicitly specified) enclosing error handling construct (e.g.
  \texttt{throw} and \texttt{catch} in Emacs
  LISP~\cite{Emacs:ELRM:Catch-and-Throw}, POSIX system calls and
  signals) or
\item
  the language does dynamic dispatch to select an appropriate error
  handler (like exceptions in most conventional languages like C++,
  Java, Python, etc do).
\end{enumerate}

Another complication is ordering:

\begin{enumerate}
\item
  Most conventional programming languages derive their error handling
  from SmallTalk~\cite{Goldberg:1983:SLI} and Common
  LISP~\cite{Pitman:2001:CHL} and the order in which the program handles
  "\emph{errors}" corresponds to the order in which execution encounters
  them.
\item
  Meanwhile, some CPU ISAs\footnote{Instruction Set Architecture (ISA)
    is a specification that describes a set of Operation Codes (OPcodes,
    which are a binary representation of an assembly language) with
    their operational semantics. "i386", "i686", "amd64" ("x86\_64"),
    "aarch64", "riscv64", etc are ISAs.} expose the internal
  non-determinism and allow different independent data-flows to produce
  hardware exceptions in non-deterministic manner (e.g. arithmetic
  instructions on DEC Alpha). So do Haskell~\cite{PeytonJones:1999:SIE}
  (see \cref{sec:imprecise}) and, to some extent,
  C++~\cite{CFAQ:SeqPoints} programming languages.
\end{enumerate}

Finally, another dimension of the problem is whenever the objects
signifying "\emph{errors}" (e.g. arguments of \VERB|\NormalTok{throw}|)
are

\begin{enumerate}
\item
  first-class values (error codes, algebraic data types) as in most
  conventional languages,
\item
  labels or tags as in modal logic with names and, to some degree, with
  \VERB|\KeywordTok{call/cc}| and \VERB|\NormalTok{callCC}|.
\end{enumerate}

In short, despite its seemingly simple operational semantics, error
handling is an algebraically rich field of programming languages theory.

Meanwhile, from the perspective of types there are several schools of
thought about effects.

\begin{itemize}
\item
  The first one, started by Gifford and
  Lucassen~\cite{Gifford:1986:IFI, Lucassen:1987:TE,
   lucassen-gifford-88} represents effects as type annotations. This
  works well in programming languages with eager evaluation, but becomes
  complicated in lazy languages (application in a lazy language delays
  effects until thunk's evaluation, hence type system has to either put
  nontrivial restrictions on the use of effects in expressions or
  annotate both arrows and values with effects, the latter, among other
  things, breaks type preservation of \(\eta\)-conversion since
  \(\lambda x . f x\) moves effect annotation from the arrow to the
  result type).
\item
  The second one, started by Moggi and
  Wadler~\cite{moggi-89,Wadler:1992:EFP} confines effects to monadic
  computations. The latter can then be annotated with effect annotations
  themselves~\cite{wadler-thiemann-03}. Monads work well for small
  programs with a small number of effects, but, it is commonly argued,
  they don't play as nice in larger programs because they lack in
  modularity~\cite{Brady:2013:PRA} (hence, the need for monad
  transformers, which are then critiqued as hard to
  tame~\cite{Kiselyov:2013:EEA}) and produce languages with non-uniform
  syntax (pure functions look very different from monadic ones and
  functions that are useful in both contexts have to be duplicated,
  think e.g. \VERB|\FunctionTok{map}| and \VERB|\FunctionTok{mapM}|).
\item
  The third one, started by Nanevski~\cite{Nanevski:2004:FPNN}
  represents effects using modal logic with names. Practical
  consequences of this way of doing things are unknown, as this
  construction didn't get much adoption yet.
\end{itemize}

In short, from type-theoretic point of view the progression of topics in
the cited literature can be seen as pursuing calculi that are, at the
same time, computationally efficient, algebraically simple (like
monads), but modular (like effect systems).

Note, however, that all of those schools of thought consider exceptions
to be effects, they only disagree about the way to represent the latter.
Meanwhile, from a perspective of a programming language implementer,
there are several problems with that world view:

\begin{itemize}
\item
  mechanisms that support resumption semantics are commonly disregarded
  as useless and computationally expensive error handling mechanisms
  (most notably~\cite[16.6 Exception Handling: Resumption
   vs. Termination, pp. 390–393]{Stroustrup:DEC94}),
\item
  in particular, all popular programming languages implement builtin
  exceptions even though they have more general error handling
  mechanisms like \emph{condition handling} in Common LISP and
  \VERB|\KeywordTok{call/cc}| in Scheme and ML because those are just
  too computationally expensive for emulation of conventional
  exceptions~\cite{Kiselyov:2012:AAC},
\item
  and even in languages with nothing but exceptions and termination
  semantics, high-performance libraries that do a lot of error handling
  frequently prefer not to use exceptions for performance reasons and to
  remove any non-local control-flow.
\end{itemize}

In short, from practical point of view \emph{most} of those
type-theoretic constructs are an overkill for \emph{most} programs.
Meanwhile, we are not aware of any non-ad-hoc language-agnostic
algebraic structure that captures all of the exception handling (both
\VERB|\NormalTok{throw}|ing, and \VERB|\FunctionTok{catch}|ing) without
introducing any other superfluous structure on top. In this article we
shall demonstrate a fairly straightforward but surprisingly useful
solution to this problem.

\hypertarget{not-a-tutorial-side-a}{%
\section{Not a Tutorial: Side A}\label{not-a-tutorial-side-a}}

\label{sec:tutorial:basic}

While algebraic structures used in this article are simple, there are a
lot of them. This section is intended as a reference point for all
algebraic structures relevant in the context of error handling that are
referenced in the rest of the paper (for reader's convenience and for
high self-sufficiency of the Literate Haskell version). Most of those
are usually assumed to be common knowledge among Haskell programmers.
Note however, that this section is not intended to be a tutorial on
either

\begin{itemize}
\tightlist
\item
  functional/declarative programming in general,
\item
  Haskell language in particular (see \cref{sec:preliminaries} for
  pointers),
\item
  error handling in Haskell in general,
\item
  practical usage of error handling structures discussed is this section
  in particular (we show only very primitive examples, if any; for the
  interesting ones the reader will have to look into citations and
  examples given in the original sources).
\end{itemize}

All structures of this section are ordered from semantically simple to
more complex (that is, we do not topologically sort them by their
dependencies in GHC sources). For the reasons of simplicity, uniformity,
self-containment, and novel perspective some of the given definitions
differ slightly from (but are isomorphic/equivalent to) the versions
provided by their original authors. The most notable difference is the
use of \VERB|\DataTypeTok{Pointed}| type class (see
\cref{sec:applicative-functor}) instead of conventional
\VERB|\DataTypeTok{Monad}|ic \VERB|\FunctionTok{return}| and
\VERB|\DataTypeTok{Applicative}| \VERB|\FunctionTok{pure}|. All
structures are listed alongside references to the corresponding papers,
documentation and original source code.

This section can be boring (although, we feel like most remarks and
footnotes are not). On the first reading we advise to skip straight to
\cref{sec:tutorial:non-basic} and refer back to this section on demand.

\hypertarget{before-monadic}{%
\subsection{Before-Monadic}\label{before-monadic}}

This subsection describes type classes that have less structure than
\VERB|\DataTypeTok{Monad}| but are useful for error handling
nevertheless.

\hypertarget{monoid}{%
\subsubsection{Monoid}\label{monoid}}

\VERB|\DataTypeTok{GHC.Base}| from \texttt{base}~\cite{Hackage:base4900}
package defines \VERB|\DataTypeTok{Monoid}| type class as
follows\footnote{\label{fn:monoid-split}Note that by following
  \VERB|\DataTypeTok{Pointed}| logic used below we should have split
  \VERB|\DataTypeTok{Monoid}| into two type classes, but since we will
  not use \VERB|\DataTypeTok{Monoid}|s that much in the rest of the
  article we shall use the original definition as is.}

\begin{Shaded}
\begin{Highlighting}[]
\KeywordTok{class} \DataTypeTok{Monoid}\NormalTok{ a }\KeywordTok{where}
\OtherTok{  mempty  ::}\NormalTok{ a}
\OtherTok{  mappend ::}\NormalTok{ a }\OtherTok{->}\NormalTok{ a }\OtherTok{->}\NormalTok{ a}

  \CommentTok{-- defined for performance reasons}
\OtherTok{  mconcat ::}\NormalTok{ [a] }\OtherTok{->}\NormalTok{ a}
  \FunctionTok{mconcat} \FunctionTok{=} \FunctionTok{foldr} \FunctionTok{mappend} \FunctionTok{mempty}
\end{Highlighting}
\end{Shaded}

\noindent and wants its instances to satisfy the following conventional
equations ("\VERB|\DataTypeTok{Monoid}| laws")

\begin{Shaded}
\begin{Highlighting}[]
\CommentTok{-- `mempty` is left identity for `mappend`,}
\FunctionTok{mempty} \OtherTok{`mappend`}\NormalTok{ x }\FunctionTok{==}\NormalTok{ x}

\CommentTok{-- `mempty` is right identity for `mappend`,}
\NormalTok{x }\OtherTok{`mappend`} \FunctionTok{mempty} \FunctionTok{==}\NormalTok{ x}

\CommentTok{-- `mappend` is associative,}
\NormalTok{x }\OtherTok{`mappend`}\NormalTok{ (y }\OtherTok{`mappend`}\NormalTok{ z)}
  \FunctionTok{==}\NormalTok{ (x }\OtherTok{`mappend`}\NormalTok{ y) }\OtherTok{`mappend`}\NormalTok{ z}
\end{Highlighting}
\end{Shaded}

\noindent and an additional constraint

\begin{Shaded}
\begin{Highlighting}[]
\CommentTok{-- and `mconcat` is extensionally}
\CommentTok{-- equal to its default implementation}
\FunctionTok{mconcat} \FunctionTok{==} \FunctionTok{foldr} \FunctionTok{mappend} \FunctionTok{mempty}
\end{Highlighting}
\end{Shaded}

Signature and default implementation for \VERB|\FunctionTok{mconcat}| is
defined in the type class because \VERB|\FunctionTok{mconcat}| is a
commonly used function that has different extensionally equal
intensionally non-equal definitions with varied performance trade-offs.
For instance,

\begin{Shaded}
\begin{Highlighting}[]
\OtherTok{mconcat' ::} \DataTypeTok{Monoid}\NormalTok{ a }\OtherTok{=>}\NormalTok{ [a] }\OtherTok{->}\NormalTok{ a}
\NormalTok{mconcat' }\FunctionTok{=}\NormalTok{ foldl' }\FunctionTok{mappend} \FunctionTok{mempty}
\end{Highlighting}
\end{Shaded}

\noindent (where \VERB|\NormalTok{foldl'}| is a strict left fold) is
another definition that satisfies the law given above (since
\VERB|\FunctionTok{mappend}| is associative), but this implementation
will not produce any superfluous thunks for strict
\VERB|\FunctionTok{mappend}|.

\VERB|\DataTypeTok{Monoid}|s are not designed for error handling per se
but programmers can use their neutral elements to represent an error and
associative composition to ignore them. Whenever "ignoring" is
"handling" is a matter of personal taste.

One of the simpler instances is, of course, a list

\begin{Shaded}
\begin{Highlighting}[]
\KeywordTok{instance} \DataTypeTok{Monoid}\NormalTok{ [a] }\KeywordTok{where}
  \FunctionTok{mempty} \FunctionTok{=}\NormalTok{ []}
  \FunctionTok{mappend} \FunctionTok{=}\NormalTok{ (}\FunctionTok{++}\NormalTok{)}
\end{Highlighting}
\end{Shaded}

\noindent and hence, for instance, functions generating errors can
produce empty lists on errors and singleton lists on successes.

\hypertarget{functor-pointed-applicative}{%
\subsubsection{Functor, Pointed,
Applicative}\label{functor-pointed-applicative}}

\label{sec:applicative-functor} \label{sec:identity}

Most of the error handling mechanisms that follow are
\VERB|\DataTypeTok{Applicative}| \VERB|\DataTypeTok{Functor}|s.
\VERB|\DataTypeTok{GHC.Base}| from \texttt{base}~\cite{Hackage:base4900}
package defines those two algebraic structures as follows

\begin{Shaded}
\begin{Highlighting}[]
\KeywordTok{class} \DataTypeTok{Functor}\NormalTok{ f }\KeywordTok{where}
\OtherTok{  fmap ::}\NormalTok{ (a }\OtherTok{->}\NormalTok{ b) }\OtherTok{->}\NormalTok{ f a }\OtherTok{->}\NormalTok{ f b}

\KeywordTok{infixl} \DecValTok{4} \FunctionTok{<*>}
\KeywordTok{class} \DataTypeTok{Functor}\NormalTok{ f }\OtherTok{=>} \DataTypeTok{Applicative}\NormalTok{ f }\KeywordTok{where}
\OtherTok{  pure ::}\NormalTok{ a }\OtherTok{->}\NormalTok{ f a}
\OtherTok{  (<*>) ::}\NormalTok{ f (a }\OtherTok{->}\NormalTok{ b) }\OtherTok{->}\NormalTok{ f a }\OtherTok{->}\NormalTok{ f b}
\end{Highlighting}
\end{Shaded}

\noindent and wants their instances to satisfy

\begin{Shaded}
\begin{Highlighting}[]
\CommentTok{-- `fmap` preserves identity}
\FunctionTok{fmap} \FunctionTok{id} \FunctionTok{==} \FunctionTok{id}

\CommentTok{-- `(<*>)` is `fmap` for pure functions}
\FunctionTok{pure}\NormalTok{ f }\FunctionTok{<*>}\NormalTok{ x }\FunctionTok{==} \FunctionTok{fmap}\NormalTok{ f x}
\end{Highlighting}
\end{Shaded}

\noindent and some more somewhat more complicated
equations~\cite{HaskellWiki:Typeclassopedia}. We shall ignore those for
the purposes of this article (we will never use them explicitly).
Meanwhile, for the purposes of this article we shall split the
\VERB|\FunctionTok{pure}| function out of
\VERB|\DataTypeTok{Applicative}| into its own
\VERB|\DataTypeTok{Pointed}| type class and redefine
\VERB|\DataTypeTok{Applicative}| using it as follows (this will simplify
some later definitions).

\begin{Shaded}
\begin{Highlighting}[]
\KeywordTok{class} \DataTypeTok{Pointed}\NormalTok{ f }\KeywordTok{where}
\OtherTok{  pure ::}\NormalTok{ a }\OtherTok{->}\NormalTok{ f a}

\KeywordTok{infixl} \DecValTok{4} \FunctionTok{<*>}
\KeywordTok{class}\NormalTok{ (}\DataTypeTok{Pointed}\NormalTok{ f, }\DataTypeTok{Functor}\NormalTok{ f) }\OtherTok{=>} \DataTypeTok{Applicative}\NormalTok{ f }\KeywordTok{where}
\OtherTok{  (<*>) ::}\NormalTok{ f (a }\OtherTok{->}\NormalTok{ b) }\OtherTok{->}\NormalTok{ f a }\OtherTok{->}\NormalTok{ f b}
\end{Highlighting}
\end{Shaded}

We shall give all definitions and laws using this hierarchy unless
explicitly stated otherwise.

The most trivial example of \VERB|\DataTypeTok{Applicative}| is the
\VERB|\DataTypeTok{Identity}| \VERB|\DataTypeTok{Functor}| defined in
\VERB|\DataTypeTok{Data.Functor.Identity}| of \texttt{base}

\begin{Shaded}
\begin{Highlighting}[]
\KeywordTok{newtype} \DataTypeTok{Identity}\NormalTok{ a }\FunctionTok{=} \DataTypeTok{Identity}
\NormalTok{  \{}\OtherTok{ runIdentity ::}\NormalTok{ a \}}

\KeywordTok{instance} \DataTypeTok{Pointed} \DataTypeTok{Identity} \KeywordTok{where}
  \FunctionTok{pure} \FunctionTok{=} \DataTypeTok{Identity}

\KeywordTok{instance} \DataTypeTok{Functor} \DataTypeTok{Identity} \KeywordTok{where}
  \FunctionTok{fmap}\NormalTok{ f (}\DataTypeTok{Identity}\NormalTok{ a) }\FunctionTok{=} \DataTypeTok{Identity}\NormalTok{ (f a)}

\KeywordTok{instance} \DataTypeTok{Applicative} \DataTypeTok{Identity} \KeywordTok{where}
\NormalTok{  (}\DataTypeTok{Identity}\NormalTok{ f) }\FunctionTok{<*>}\NormalTok{ (}\DataTypeTok{Identity}\NormalTok{ x) }\FunctionTok{=} \DataTypeTok{Identity}\NormalTok{ (f x)}
\end{Highlighting}
\end{Shaded}

The most trivial example of a \VERB|\DataTypeTok{Functor}| that is not
\VERB|\DataTypeTok{Applicative}| is \VERB|\DataTypeTok{Const}|ant
\VERB|\DataTypeTok{Functor}| defined in
\VERB|\DataTypeTok{Data.Functor.Const}| of \texttt{base} as

\begin{Shaded}
\begin{Highlighting}[]
\KeywordTok{newtype} \DataTypeTok{Const}\NormalTok{ a b }\FunctionTok{=} \DataTypeTok{Const}
\NormalTok{  \{}\OtherTok{ getConst ::}\NormalTok{ a \}}

\KeywordTok{instance} \DataTypeTok{Functor}\NormalTok{ (}\DataTypeTok{Const}\NormalTok{ a) }\KeywordTok{where}
  \CommentTok{-- note that it changes type here}
  \FunctionTok{fmap}\NormalTok{ f (}\DataTypeTok{Const}\NormalTok{ a) }\FunctionTok{=} \DataTypeTok{Const}\NormalTok{ a}
  \CommentTok{-- so the following would not work}
  \CommentTok{-- fmap f x = x}
\end{Highlighting}
\end{Shaded}

\noindent It is missing a \VERB|\DataTypeTok{Pointed}| instance.
However, if the argument of \VERB|\DataTypeTok{Const}| is a
\VERB|\DataTypeTok{Monoid}| we can define it as

\begin{Shaded}
\begin{Highlighting}[]
\KeywordTok{instance} \DataTypeTok{Monoid}\NormalTok{ a }\OtherTok{=>} \DataTypeTok{Pointed}\NormalTok{ (}\DataTypeTok{Const}\NormalTok{ a) }\KeywordTok{where}
  \FunctionTok{pure}\NormalTok{ a }\FunctionTok{=} \DataTypeTok{Const} \FunctionTok{mempty}

\KeywordTok{instance} \DataTypeTok{Monoid}\NormalTok{ a }\OtherTok{=>} \DataTypeTok{Applicative}\NormalTok{ (}\DataTypeTok{Const}\NormalTok{ a) }\KeywordTok{where}
  \DataTypeTok{Const}\NormalTok{ x }\FunctionTok{<*>} \DataTypeTok{Const}\NormalTok{ a }\FunctionTok{=} \DataTypeTok{Const}\NormalTok{ (}\FunctionTok{mappend}\NormalTok{ x a)}
\end{Highlighting}
\end{Shaded}

\begin{remark}

\label{rem:applicative-as-app}

One can think of \VERB|\DataTypeTok{Applicative}\NormalTok{ f}| as
representing \emph{generalized function application} on structure
\VERB|\NormalTok{f}|: \VERB|\FunctionTok{pure}| lifts pure values into
\VERB|\NormalTok{f}| while
\VERB|\NormalTok{(}\FunctionTok{<*>}\NormalTok{)}| provides a way to
apply functions to arguments over \VERB|\NormalTok{f}|. Note however,
that \VERB|\DataTypeTok{Applicative}| is not a structure for
representing \emph{generalized functions} (e.g.
\VERB|\DataTypeTok{Applicative}| gives no way to compose functions or to
introduce lambdas, unlike the \VERB|\DataTypeTok{Monad}|, see
\cref{rem:monad-as-app}).

\end{remark}

\hypertarget{alternative}{%
\subsubsection{Alternative}\label{alternative}}

\label{sec:alternative}

\VERB|\DataTypeTok{Control.Applicative}| module of
\texttt{base}~\cite{Hackage:base4900} defines
\VERB|\DataTypeTok{Alternative}| class as a monoid on
\VERB|\DataTypeTok{Applicative}|
\VERB|\DataTypeTok{Functor}|s.\cref{fn:monoid-split}

\begin{Shaded}
\begin{Highlighting}[]
\KeywordTok{class} \DataTypeTok{Applicative}\NormalTok{ f }\OtherTok{=>} \DataTypeTok{Alternative}\NormalTok{ f }\KeywordTok{where}
\OtherTok{  empty ::}\NormalTok{ f a}
\OtherTok{  (<|>) ::}\NormalTok{ f a }\OtherTok{->}\NormalTok{ f a }\OtherTok{->}\NormalTok{ f a}

  \CommentTok{-- defined for performance reasons}
\OtherTok{  some ::}\NormalTok{ f a }\OtherTok{->}\NormalTok{ f [a]}
\NormalTok{  some v }\FunctionTok{=} \FunctionTok{fmap}\NormalTok{ (}\FunctionTok{:}\NormalTok{) v }\FunctionTok{<*>}\NormalTok{ many v}

\OtherTok{  many ::}\NormalTok{ f a }\OtherTok{->}\NormalTok{ f [a]}
\NormalTok{  many v }\FunctionTok{=}\NormalTok{ some v }\FunctionTok{<|>} \FunctionTok{pure}\NormalTok{ []}
\end{Highlighting}
\end{Shaded}

\noindent requiring monoid laws to hold for \VERB|\NormalTok{empty}| and
\VERB|\NormalTok{(}\FunctionTok{<\VerbBar{}>}\NormalTok{)}|

\begin{Shaded}
\begin{Highlighting}[]
\CommentTok{-- `empty` is left identity for `(<|>)`,}
\NormalTok{empty }\FunctionTok{<|>}\NormalTok{ x }\FunctionTok{==}\NormalTok{ x}

\CommentTok{-- `empty` is right identity for `(<|>)`,}
\NormalTok{x }\FunctionTok{<|>}\NormalTok{ empty }\FunctionTok{==}\NormalTok{ x}

\CommentTok{-- `(<|>)` is associative,}
\NormalTok{x }\FunctionTok{<|>}\NormalTok{ (y }\FunctionTok{<|>}\NormalTok{ z)}
  \FunctionTok{==}\NormalTok{ (x }\FunctionTok{<|>}\NormalTok{ y) }\FunctionTok{<|>}\NormalTok{ z}

\CommentTok{-- and both `some` and `many` are}
\CommentTok{-- extensionally equal to their}
\CommentTok{-- default implementations}
\NormalTok{some v }\FunctionTok{==} \FunctionTok{fmap}\NormalTok{ (}\FunctionTok{:}\NormalTok{) v }\FunctionTok{<*>}\NormalTok{ many v}
\NormalTok{many v }\FunctionTok{==}\NormalTok{ some v }\FunctionTok{<|>} \FunctionTok{pure}\NormalTok{ []}
\end{Highlighting}
\end{Shaded}

Combinators \VERB|\NormalTok{some}| and \VERB|\NormalTok{many}|,
similarly to \VERB|\FunctionTok{mconcat}|, commonly occur in functions
handling \VERB|\DataTypeTok{Alternative}|s and can have different
definitions varying in performance for different types. The most common
use of \VERB|\DataTypeTok{Alternative}| type class is parser combinators
(\cref{sec:parser-combinators}) where \VERB|\NormalTok{some}| and
\VERB|\NormalTok{many}| coincide with \texttt{+} ("one or more") and
\texttt{*} ("zero or more", Kleene star) operators from regular
expressions/EBNF. Before the introduction of
\VERB|\DataTypeTok{Alternative}| that role was played by now deprecated
\VERB|\DataTypeTok{MonadPlus}| class, currently defined in
\VERB|\DataTypeTok{Control.Monad}| of \texttt{base} as follows

\begin{Shaded}
\begin{Highlighting}[]
\KeywordTok{class}\NormalTok{ (}\DataTypeTok{Alternative}\NormalTok{ m, }\DataTypeTok{Monad}\NormalTok{ m) }\OtherTok{=>} \DataTypeTok{MonadPlus}\NormalTok{ m }\KeywordTok{where}
\OtherTok{  mzero ::}\NormalTok{ m a}
\NormalTok{  mzero }\FunctionTok{=}\NormalTok{ empty}

\OtherTok{  mplus ::}\NormalTok{ m a }\OtherTok{->}\NormalTok{ m a }\OtherTok{->}\NormalTok{ m a}
\NormalTok{  mplus }\FunctionTok{=}\NormalTok{ (}\FunctionTok{<|>}\NormalTok{)}
\end{Highlighting}
\end{Shaded}

We shall give example instance and usage of
\VERB|\DataTypeTok{Alternative}| in \cref{sec:parser-combinators}.

\hypertarget{purely-monadic}{%
\subsection{Purely Monadic}\label{purely-monadic}}

This subsection describes algebraic structures that involve
\VERB|\DataTypeTok{Monad}| type class and its instances.

\hypertarget{monad-definition}{%
\subsubsection{Monad definition}\label{monad-definition}}

\label{sec:monad}

\VERB|\DataTypeTok{GHC.Base}| from \texttt{base}~\cite{Hackage:base4900}
defines \VERB|\DataTypeTok{Monad}| in the following way using the
original (i.e. not \VERB|\DataTypeTok{Pointed}|) hierarchy (also, at the
time of writing \texttt{base} uses a bit uglier definition which is
discussed in \cref{sec:monad-fail})

\begin{Shaded}
\begin{Highlighting}[]
\KeywordTok{infixl} \DecValTok{1} \FunctionTok{>>=}
\KeywordTok{class} \DataTypeTok{Applicative}\NormalTok{ m }\OtherTok{=>} \DataTypeTok{Monad}\NormalTok{ m }\KeywordTok{where}
\OtherTok{  return  ::}\NormalTok{ a }\OtherTok{->}\NormalTok{ m a}
\OtherTok{  (>>=)   ::}\NormalTok{ m a }\OtherTok{->}\NormalTok{ (a }\OtherTok{->}\NormalTok{ m b) }\OtherTok{->}\NormalTok{ m b}
\end{Highlighting}
\end{Shaded}

\noindent and wants its instances to satisfy the following equations
known as "\VERB|\DataTypeTok{Monad}| laws"

\begin{Shaded}
\begin{Highlighting}[]
\CommentTok{-- `return` is left identity for `(>>=)`}
\FunctionTok{return}\NormalTok{ a }\FunctionTok{>>=}\NormalTok{ f }\FunctionTok{==}\NormalTok{ f a}

\CommentTok{-- `return` is right identity for `(>>=)`}
\NormalTok{f }\FunctionTok{>>=} \FunctionTok{return} \FunctionTok{==}\NormalTok{ f}

\CommentTok{-- `(>>=)` is associative}
\NormalTok{(f }\FunctionTok{>>=}\NormalTok{ g) }\FunctionTok{>>=}\NormalTok{ h }\FunctionTok{==}\NormalTok{ f }\FunctionTok{>>=}\NormalTok{ (\textbackslash{}x }\OtherTok{->}\NormalTok{ g x }\FunctionTok{>>=}\NormalTok{ h)}
\end{Highlighting}
\end{Shaded}

Note that this definition also expects the following additional
"unspoken laws" from its parent structures (see \cref{sec:boilerplate}
for definitions of \VERB|\NormalTok{liftM}| and \VERB|\NormalTok{ap}|).

\begin{Shaded}
\begin{Highlighting}[]
\FunctionTok{fmap} \FunctionTok{==}\NormalTok{ liftM}
\FunctionTok{pure} \FunctionTok{==} \FunctionTok{return}
\NormalTok{(}\FunctionTok{<*>}\NormalTok{) }\FunctionTok{==}\NormalTok{ ap}
\end{Highlighting}
\end{Shaded}

Moreover, we feel that the name "return" itself is an unfortunate
accident since \VERB|\FunctionTok{return}| only injects pure values into
\VERB|\NormalTok{m}| and does not "return" anywhere. We shall avoid that
problem and simplify the above equations by redefining
\VERB|\DataTypeTok{Monad}| using \VERB|\DataTypeTok{Pointed}| hierarchy
instead

\begin{Shaded}
\begin{Highlighting}[]
\KeywordTok{infixl} \DecValTok{1} \FunctionTok{>>=}
\KeywordTok{class} \DataTypeTok{Applicative}\NormalTok{ m }\OtherTok{=>} \DataTypeTok{Monad}\NormalTok{ m }\KeywordTok{where}
\OtherTok{  (>>=)   ::}\NormalTok{ m a }\OtherTok{->}\NormalTok{ (a }\OtherTok{->}\NormalTok{ m b) }\OtherTok{->}\NormalTok{ m b}

\CommentTok{-- for backward-compatibility}
\FunctionTok{return}\OtherTok{ ::} \DataTypeTok{Monad}\NormalTok{ m }\OtherTok{=>}\NormalTok{ a }\OtherTok{->}\NormalTok{ m a}
\FunctionTok{return} \FunctionTok{=} \FunctionTok{pure}
\end{Highlighting}
\end{Shaded}

\begin{remark}

\label{rem:monad-as-app}

Note that while \VERB|\DataTypeTok{Applicative}| is too weak to express
\emph{generalized functions} (\cref{rem:applicative-as-app}),
\VERB|\DataTypeTok{Monad}|, in some sense, is too strong since
\VERB|\NormalTok{(}\FunctionTok{>>=}\NormalTok{)}| combines function
composition (the whole type) with lambda introduction (the type of the
second argument). This might be easier to see with the definition given
in \cref{sec:monad-fish}.

What is the "just right" structure for representing a \emph{generalized
function} is a matter of debate: some would state "an
\VERB|\DataTypeTok{Arrow}|!"~\cite{hughes-arrows-00}, others "a
(Cartesian Closed)
\VERB|\DataTypeTok{Category}|!"~\cite{Elliott:2017:CTC}, yet others
might disagree with both.

\end{remark}

A very common combinator used with \VERB|\DataTypeTok{Monad}|s bears a
name of \VERB|\NormalTok{(}\FunctionTok{>>}\NormalTok{)}| and can be
defined as

\begin{Shaded}
\begin{Highlighting}[]
\OtherTok{(>>) ::} \DataTypeTok{Monad}\NormalTok{ m }\OtherTok{=>}\NormalTok{ m a }\OtherTok{->}\NormalTok{ m b }\OtherTok{->}\NormalTok{ m b}
\NormalTok{a }\FunctionTok{>>}\NormalTok{ b }\FunctionTok{=}\NormalTok{ a }\FunctionTok{>>=} \FunctionTok{const}\NormalTok{ b}
      \CommentTok{-- a >>= \textbackslash{}_ -> b}
\end{Highlighting}
\end{Shaded}

The following subsections will provide many example instances.

\hypertarget{monadfish}{%
\subsubsection{MonadFish}\label{monadfish}}

\label{sec:monad-fish}

A somewhat lesser known but equivalent way to define
\VERB|\DataTypeTok{Monad}| is to define
\VERB|\NormalTok{(}\FunctionTok{>>=}\NormalTok{)}| in "fish" form as
follows

\begin{Shaded}
\begin{Highlighting}[]
\KeywordTok{infixl} \DecValTok{1} \FunctionTok{>=>}
\KeywordTok{class} \DataTypeTok{Applicative}\NormalTok{ m }\OtherTok{=>} \DataTypeTok{MonadFish}\NormalTok{ m }\KeywordTok{where}
\OtherTok{  (>=>)   ::}\NormalTok{ (a }\OtherTok{->}\NormalTok{ m b) }\OtherTok{->}\NormalTok{ (b }\OtherTok{->}\NormalTok{ m c) }\OtherTok{->}\NormalTok{ (a }\OtherTok{->}\NormalTok{ m c)}
\end{Highlighting}
\end{Shaded}

This way \VERB|\DataTypeTok{Monad}| laws become
\VERB|\DataTypeTok{Monoid}| laws

\begin{Shaded}
\begin{Highlighting}[]
\CommentTok{-- `pure` is left identity for `(>=>)`}
\FunctionTok{pure} \FunctionTok{>=>}\NormalTok{ f }\FunctionTok{==}\NormalTok{ f}

\CommentTok{-- `pure` is right identity for `(>=>)`}
\NormalTok{f }\FunctionTok{>=>} \FunctionTok{pure} \FunctionTok{==}\NormalTok{ f}

\CommentTok{-- `(>=>)` is associative}
\NormalTok{(f }\FunctionTok{>=>}\NormalTok{ g) }\FunctionTok{>=>}\NormalTok{ h }\FunctionTok{==}\NormalTok{ f }\FunctionTok{>=>}\NormalTok{ (g }\FunctionTok{>=>}\NormalTok{ h)}
\end{Highlighting}
\end{Shaded}

Both definitions of \VERB|\DataTypeTok{Monad}| are known to be
equivalent in the folklore, but we could not find a reference with a
simple proof of that fact, hence we shall give one ourselves.

\begin{lemma}

\VERB|\NormalTok{(f }\FunctionTok{>=>}\NormalTok{ g) }\FunctionTok{.}\NormalTok{ h }\FunctionTok{==}\NormalTok{ (f }\FunctionTok{.}\NormalTok{ h) }\FunctionTok{>=>}\NormalTok{ g}|

\end{lemma}

\begin{proof}

For pure values \VERB|\NormalTok{(}\FunctionTok{>=>}\NormalTok{)}| is a
composition with flipped order of arguments
\VERB|\NormalTok{(}\FunctionTok{.}\NormalTok{)}|

\begin{Shaded}
\begin{Highlighting}[]
\KeywordTok{instance} \DataTypeTok{MonadFish} \DataTypeTok{Identity} \KeywordTok{where}
\NormalTok{  f }\FunctionTok{>=>}\NormalTok{ g }\FunctionTok{=}\NormalTok{ g }\FunctionTok{.}\NormalTok{ runIdentity }\FunctionTok{.}\NormalTok{ f}
\end{Highlighting}
\end{Shaded}

In other words,
\VERB|\NormalTok{f }\FunctionTok{>=>}\NormalTok{ g }\FunctionTok{==}\NormalTok{ g }\FunctionTok{.}\NormalTok{ f}|,
which gives the following

\begin{Shaded}
\begin{Highlighting}[]
\NormalTok{(f }\FunctionTok{>=>}\NormalTok{ g) }\FunctionTok{.}\NormalTok{ h }\FunctionTok{==}\NormalTok{ h }\FunctionTok{>=>}\NormalTok{ (f }\FunctionTok{>=>}\NormalTok{ g)}
              \FunctionTok{==}\NormalTok{ (h }\FunctionTok{>=>} \FunctionTok{pure}\NormalTok{) }\FunctionTok{>=>}\NormalTok{ (f }\FunctionTok{>=>}\NormalTok{ g)}
              \FunctionTok{==}\NormalTok{ ((h }\FunctionTok{>=>} \FunctionTok{pure}\NormalTok{) }\FunctionTok{>=>}\NormalTok{ f) }\FunctionTok{>=>}\NormalTok{ g}
              \FunctionTok{==}\NormalTok{ (h }\FunctionTok{>=>}\NormalTok{ f) }\FunctionTok{>=>}\NormalTok{ g}
              \FunctionTok{==}\NormalTok{ (f }\FunctionTok{.}\NormalTok{ h) }\FunctionTok{>=>}\NormalTok{ g}
\end{Highlighting}
\end{Shaded}

\noindent which, with some abuse of notation
(\VERB|\NormalTok{(}\FunctionTok{>=>}\NormalTok{)}| is not
heterogeneous, the above lifts pure values into \VERB|\NormalTok{m}|
with \VERB|\FunctionTok{pure}|), can be written simply as

\begin{Shaded}
\begin{Highlighting}[]
\NormalTok{(f }\FunctionTok{>=>}\NormalTok{ g) }\FunctionTok{.}\NormalTok{ h }\FunctionTok{==}\NormalTok{ h }\FunctionTok{>=>}\NormalTok{ (f }\FunctionTok{>=>}\NormalTok{ g)}
              \FunctionTok{==}\NormalTok{ (h }\FunctionTok{>=>}\NormalTok{ f) }\FunctionTok{>=>}\NormalTok{ g}
              \FunctionTok{==}\NormalTok{ (f }\FunctionTok{.}\NormalTok{ h) }\FunctionTok{>=>}\NormalTok{ g}
\end{Highlighting}
\end{Shaded}

\end{proof}

\begin{lemma}

\VERB|\DataTypeTok{Monad}| and \VERB|\DataTypeTok{MonadFish}| define the
same structure.

\end{lemma}

\begin{proof}

The cross-definitions:

\begin{Shaded}
\begin{Highlighting}[]
\KeywordTok{instance}\NormalTok{ (}\DataTypeTok{Applicative}\NormalTok{ m, }\DataTypeTok{Monad}\NormalTok{ m) }\OtherTok{=>} \DataTypeTok{MonadFish}\NormalTok{ m }\KeywordTok{where}
\NormalTok{  f }\FunctionTok{>=>}\NormalTok{ g }\FunctionTok{=}\NormalTok{ \textbackslash{}a }\OtherTok{->}\NormalTok{ (f a) }\FunctionTok{>>=}\NormalTok{ g }\CommentTok{-- (1)}

\KeywordTok{instance} \OtherTok{\{-# OVERLAPPABLE #-\}}
\NormalTok{         (}\DataTypeTok{Applicative}\NormalTok{ m, }\DataTypeTok{MonadFish}\NormalTok{ m) }\OtherTok{=>} \DataTypeTok{Monad}\NormalTok{ m }\KeywordTok{where}
\NormalTok{  ma }\FunctionTok{>>=}\NormalTok{ f }\FunctionTok{=}\NormalTok{ (}\FunctionTok{id} \FunctionTok{>=>}\NormalTok{ f) ma }\CommentTok{-- (2)}
\end{Highlighting}
\end{Shaded}

\begin{itemize}
\item
  (1) implies (2):

\begin{Shaded}
\begin{Highlighting}[]
\NormalTok{ma }\FunctionTok{>>=}\NormalTok{ f }\FunctionTok{==}\NormalTok{ (}\FunctionTok{id} \FunctionTok{>=>}\NormalTok{ f) ma}
         \FunctionTok{==}\NormalTok{ (\textbackslash{}a }\OtherTok{->} \FunctionTok{id}\NormalTok{ a }\FunctionTok{>>=}\NormalTok{ f) ma}
         \FunctionTok{==}\NormalTok{ ma }\FunctionTok{>>=}\NormalTok{ f}
\end{Highlighting}
\end{Shaded}
\item
  (2) implies (1):

\begin{Shaded}
\begin{Highlighting}[]
\NormalTok{f }\FunctionTok{>=>}\NormalTok{ g }\FunctionTok{==}\NormalTok{ \textbackslash{}a }\OtherTok{->}\NormalTok{ (f a) }\FunctionTok{>>=}\NormalTok{ g}
        \FunctionTok{==}\NormalTok{ \textbackslash{}a }\OtherTok{->}\NormalTok{ (}\FunctionTok{id} \FunctionTok{>=>}\NormalTok{ g) (f a)}
        \FunctionTok{==}\NormalTok{ (}\FunctionTok{id} \FunctionTok{>=>}\NormalTok{ g) }\FunctionTok{.}\NormalTok{ f}
        \FunctionTok{==}\NormalTok{ (}\FunctionTok{id} \FunctionTok{.}\NormalTok{ f) }\FunctionTok{>=>}\NormalTok{ g}
        \FunctionTok{==}\NormalTok{ f }\FunctionTok{>=>}\NormalTok{ g}
\end{Highlighting}
\end{Shaded}
\end{itemize}

\end{proof}

\hypertarget{monads-fail-and-monadfail}{%
\subsubsection{Monad's fail and
MonadFail}\label{monads-fail-and-monadfail}}

\label{sec:monad-fail}

Section~\ref{sec:monad} did not give the complete definition of
\VERB|\DataTypeTok{Monad}| as is defined in the current version of
\texttt{base}~\cite{Hackage:base4900}. Current
\VERB|\DataTypeTok{GHC.Base}| module defines \VERB|\DataTypeTok{Monad}|
in the following way using the original (not
\VERB|\DataTypeTok{Pointed}|) hierarchy

\begin{Shaded}
\begin{Highlighting}[]
\KeywordTok{infixl} \DecValTok{1} \FunctionTok{>>=}
\KeywordTok{class} \DataTypeTok{Applicative}\NormalTok{ m }\OtherTok{=>} \DataTypeTok{Monad}\NormalTok{ m }\KeywordTok{where}
\OtherTok{  return  ::}\NormalTok{ a }\OtherTok{->}\NormalTok{ m a}
\OtherTok{  (>>=)   ::}\NormalTok{ m a }\OtherTok{->}\NormalTok{ (a }\OtherTok{->}\NormalTok{ m b) }\OtherTok{->}\NormalTok{ m b}

\OtherTok{  fail    ::} \DataTypeTok{String} \OtherTok{->}\NormalTok{ m a}
  \FunctionTok{fail}\NormalTok{ s  }\FunctionTok{=} \FunctionTok{error}\NormalTok{ s}
\end{Highlighting}
\end{Shaded}

Note the definition of the \VERB|\FunctionTok{fail}| operation. That
function is invoked by the compiler on pattern match failures in
\VERB|\KeywordTok{do}|-expressions (see \cref{sec:non-exhaustive} for
examples, see \cref{sec:error-undefined} for the definition of
\VERB|\FunctionTok{error}|), but it can also be called explicitly by the
programmer in any context where the type permits to do so.

The presence of \VERB|\FunctionTok{fail}| in \VERB|\DataTypeTok{Monad}|
class is, clearly\footnote{It involves an error handling mechanism that
  is more complicated than the thing itself. It creates semantic
  discrepancies (e.g. \VERB|\DataTypeTok{Maybe}| is not equivalent to
  \VERB|\DataTypeTok{Either}\NormalTok{ ()}|, see \cref{sec:either}).},
a hack. There is an ongoing effort (aka "\VERB|\DataTypeTok{MonadFail}|
proposal", "MFP") to move this function from \VERB|\DataTypeTok{Monad}|
to its own type class defined as follows (in both hierarchies)

\begin{Shaded}
\begin{Highlighting}[]
\KeywordTok{class} \DataTypeTok{Monad}\NormalTok{ m }\OtherTok{=>} \DataTypeTok{MonadFail}\NormalTok{ m }\KeywordTok{where}
\OtherTok{  fail ::} \DataTypeTok{String} \OtherTok{->}\NormalTok{ m a}
  \FunctionTok{fail}\NormalTok{ s  }\FunctionTok{=} \FunctionTok{error}\NormalTok{ s}
\end{Highlighting}
\end{Shaded}

As of writing of this article the new class is available from
\VERB|\DataTypeTok{Control.Monad.Fail}|, but \VERB|\FunctionTok{fail}|
from the original \VERB|\DataTypeTok{Monad}| is not even deprecated yet.
We shall use \VERB|\DataTypeTok{MonadFail}| instead of the original
\VERB|\FunctionTok{fail}| in our hierarchy for simplicity.

\hypertarget{identity-monad}{%
\subsubsection{Identity monad}\label{identity-monad}}

\label{sec:identity-monad}

We can define the following \VERB|\DataTypeTok{Monad}| and
\VERB|\DataTypeTok{MonadFail}| instances for the
\VERB|\DataTypeTok{Identity}| \VERB|\DataTypeTok{Functor}|

\begin{Shaded}
\begin{Highlighting}[]
\KeywordTok{instance} \DataTypeTok{Monad} \DataTypeTok{Identity} \KeywordTok{where}
\NormalTok{  (}\DataTypeTok{Identity}\NormalTok{ x) }\FunctionTok{>>=}\NormalTok{ f }\FunctionTok{=}\NormalTok{ f x}

\KeywordTok{instance} \DataTypeTok{MonadFail} \DataTypeTok{Identity} \KeywordTok{where}
  \CommentTok{-- default implementation}
\end{Highlighting}
\end{Shaded}

\noindent despite this instance it is still usually referenced as
"\VERB|\DataTypeTok{Identity}| \VERB|\DataTypeTok{Functor}|" even though
it is also an \VERB|\DataTypeTok{Applicative}| and a
\VERB|\DataTypeTok{Monad}|.

\hypertarget{maybe-monad}{%
\subsubsection{Maybe monad}\label{maybe-monad}}

\label{sec:maybe}

The simplest form of \VERB|\DataTypeTok{Monad}|ic error handling (that
is, not just "error ignoring") can be done with
\VERB|\DataTypeTok{Maybe}| data type and its \VERB|\DataTypeTok{Monad}|
instance defined in \VERB|\DataTypeTok{Data.Maybe}| of
\texttt{base}~\cite{Hackage:base4900} equivalently to

\begin{Shaded}
\begin{Highlighting}[]
\KeywordTok{data} \DataTypeTok{Maybe}\NormalTok{ a }\FunctionTok{=} \DataTypeTok{Nothing} \FunctionTok{|} \DataTypeTok{Just}\NormalTok{ a}

\KeywordTok{instance} \DataTypeTok{Pointed} \DataTypeTok{Maybe} \KeywordTok{where}
  \FunctionTok{pure} \FunctionTok{=} \DataTypeTok{Just}

\KeywordTok{instance} \DataTypeTok{Monad} \DataTypeTok{Maybe} \KeywordTok{where}
\NormalTok{  (}\DataTypeTok{Just}\NormalTok{ x) }\FunctionTok{>>=}\NormalTok{ k }\FunctionTok{=}\NormalTok{ k x}
  \DataTypeTok{Nothing}  \FunctionTok{>>=}\NormalTok{ _ }\FunctionTok{=} \DataTypeTok{Nothing}

\KeywordTok{instance} \DataTypeTok{MonadFail} \DataTypeTok{Maybe} \KeywordTok{where}
  \CommentTok{-- custom `fail`}
  \FunctionTok{fail}\NormalTok{ _         }\FunctionTok{=} \DataTypeTok{Nothing}
\end{Highlighting}
\end{Shaded}

The \VERB|\FunctionTok{pure}| operator simply injects a given value
under \VERB|\DataTypeTok{Just}| constructor, while the definition of
\VERB|\NormalTok{(}\FunctionTok{>>=}\NormalTok{)}| ensures that

\begin{itemize}
\item
  injected values are transparently propagated further down the
  computation path,
\item
  computation stops as soon as the first \VERB|\DataTypeTok{Nothing}|
  gets emitted.
\end{itemize}

In other words, \VERB|\DataTypeTok{Maybe}| \VERB|\DataTypeTok{Monad}| is
\VERB|\DataTypeTok{Identity}| \VERB|\DataTypeTok{Monad}| that can stop
its computation on request. A couple of examples follow

\begin{Shaded}
\begin{Highlighting}[]
\OtherTok{maybeTest1 ::} \DataTypeTok{Maybe} \DataTypeTok{Int}
\NormalTok{maybeTest1 }\FunctionTok{=} \KeywordTok{do}
\NormalTok{  x }\OtherTok{<-} \DataTypeTok{Just} \DecValTok{1}
  \FunctionTok{pure}\NormalTok{ x}

\OtherTok{maybeTest2 ::} \DataTypeTok{Maybe} \DataTypeTok{Int}
\NormalTok{maybeTest2 }\FunctionTok{=} \KeywordTok{do}
\NormalTok{  x }\OtherTok{<-} \DataTypeTok{Just} \DecValTok{1}
  \FunctionTok{pure}\NormalTok{ x}
  \DataTypeTok{Nothing}
  \DataTypeTok{Just} \DecValTok{2}

\NormalTok{maybeTest }\FunctionTok{=}\NormalTok{ maybeTest1 }\FunctionTok{==} \DataTypeTok{Just} \DecValTok{1}
         \FunctionTok{&&}\NormalTok{ maybeTest2 }\FunctionTok{==} \DataTypeTok{Nothing}
\end{Highlighting}
\end{Shaded}

\hypertarget{either-monad}{%
\subsubsection{Either monad}\label{either-monad}}

\label{sec:either}

\VERB|\DataTypeTok{Either}| data type is defined in
\VERB|\DataTypeTok{Data.Either}| of
\texttt{base}~\cite{Hackage:base4900} equivalently to

\begin{Shaded}
\begin{Highlighting}[]
\KeywordTok{data} \DataTypeTok{Either}\NormalTok{ a b }\FunctionTok{=} \DataTypeTok{Left}\NormalTok{ a }\FunctionTok{|} \DataTypeTok{Right}\NormalTok{ b}

\KeywordTok{instance} \DataTypeTok{Pointed}\NormalTok{ (}\DataTypeTok{Either}\NormalTok{ e) }\KeywordTok{where}
  \FunctionTok{pure} \FunctionTok{=} \DataTypeTok{Right}

\KeywordTok{instance} \DataTypeTok{Monad}\NormalTok{ (}\DataTypeTok{Either}\NormalTok{ e) }\KeywordTok{where}
  \DataTypeTok{Left}\NormalTok{  l }\FunctionTok{>>=}\NormalTok{ _ }\FunctionTok{=} \DataTypeTok{Left}\NormalTok{ l}
  \DataTypeTok{Right}\NormalTok{ r }\FunctionTok{>>=}\NormalTok{ k }\FunctionTok{=}\NormalTok{ k r}

\KeywordTok{instance} \DataTypeTok{MonadFail}\NormalTok{ (}\DataTypeTok{Either}\NormalTok{ e)}
  \CommentTok{-- default `fail`}
\end{Highlighting}
\end{Shaded}

\VERB|\DataTypeTok{Either}| is a computation that can stop and report a
given value (the argument of \VERB|\DataTypeTok{Left}|) when falling out
of \VERB|\DataTypeTok{Identity}| execution. The intended use is similar
to \VERB|\DataTypeTok{Maybe}|

\begin{Shaded}
\begin{Highlighting}[]
\OtherTok{eitherTest1 ::} \DataTypeTok{Either} \DataTypeTok{String} \DataTypeTok{Int}
\NormalTok{eitherTest1 }\FunctionTok{=} \KeywordTok{do}
\NormalTok{  x }\OtherTok{<-} \DataTypeTok{Right} \DecValTok{1}
  \FunctionTok{pure}\NormalTok{ x}

\OtherTok{eitherTest2 ::} \DataTypeTok{Either} \DataTypeTok{String} \DataTypeTok{Int}
\NormalTok{eitherTest2 }\FunctionTok{=} \KeywordTok{do}
\NormalTok{  x }\OtherTok{<-} \DataTypeTok{Right} \DecValTok{1}
  \FunctionTok{pure}\NormalTok{ x}
  \DataTypeTok{Left} \StringTok{"oops"}
  \DataTypeTok{Right} \DecValTok{2}

\NormalTok{eitherTest }\FunctionTok{=}\NormalTok{ eitherTest1 }\FunctionTok{==} \DataTypeTok{Right} \DecValTok{1}
          \FunctionTok{&&}\NormalTok{ eitherTest2 }\FunctionTok{==} \DataTypeTok{Left} \StringTok{"oops"}
\end{Highlighting}
\end{Shaded}

Purely by its data type definition
\VERB|\DataTypeTok{Maybe}\NormalTok{ a}| is isomorphic to
\VERB|\DataTypeTok{Either}\NormalTok{ () a}| (where
\VERB|\NormalTok{()}| is Haskell's name for the ML's
\VERB|\DataTypeTok{unit}| type and type-theoretic "top" type), but their
\VERB|\DataTypeTok{Monad}| instances (in the original hierarchy,
\VERB|\DataTypeTok{MonadFail}| in our hierarchy) differ:
\VERB|\DataTypeTok{Maybe}| has non-default \VERB|\FunctionTok{fail}|,
while \VERB|\DataTypeTok{Either}| does not. This produces some
observable differences discussed in \cref{sec:non-exhaustive}.

\hypertarget{an-intermission-on-monadic-boilerplate}{%
\subsection{An intermission on Monadic
boilerplate}\label{an-intermission-on-monadic-boilerplate}}

\label{sec:boilerplate}

Haskell does not support default definitions for functions in
superclasses that use definitions given in subclasses. That is, Haskell
has no syntax to define \VERB|\DataTypeTok{Functor}| and
\VERB|\DataTypeTok{Applicative}| defaults from
\VERB|\DataTypeTok{Monad}| instance of the same type.

Which is why to compile the code above we have to borrow a couple of
functions from \VERB|\DataTypeTok{Control.Monad}| of \texttt{base}

\begin{Shaded}
\begin{Highlighting}[]
\OtherTok{liftM ::}\NormalTok{ (}\DataTypeTok{Monad}\NormalTok{ m)}
      \OtherTok{=>}\NormalTok{ (a }\OtherTok{->}\NormalTok{ b) }\OtherTok{->}\NormalTok{ m a }\OtherTok{->}\NormalTok{ m b}
\NormalTok{liftM f ma }\FunctionTok{=}\NormalTok{ ma }\FunctionTok{>>=} \FunctionTok{pure} \FunctionTok{.}\NormalTok{ f}

\OtherTok{ap ::}\NormalTok{ (}\DataTypeTok{Monad}\NormalTok{ m)}
   \OtherTok{=>}\NormalTok{ m (a }\OtherTok{->}\NormalTok{ b) }\OtherTok{->}\NormalTok{ m a }\OtherTok{->}\NormalTok{ m b}
\NormalTok{ap mf ma }\FunctionTok{=}\NormalTok{ mf }\FunctionTok{>>=}\NormalTok{ \textbackslash{}f }\OtherTok{->}\NormalTok{ liftM f ma}
\end{Highlighting}
\end{Shaded}

\noindent and use them to define

\begin{Shaded}
\begin{Highlighting}[]
\KeywordTok{instance} \DataTypeTok{Functor} \DataTypeTok{Maybe} \KeywordTok{where}
  \FunctionTok{fmap} \FunctionTok{=}\NormalTok{ liftM}

\KeywordTok{instance} \DataTypeTok{Applicative} \DataTypeTok{Maybe} \KeywordTok{where}
\NormalTok{  (}\FunctionTok{<*>}\NormalTok{) }\FunctionTok{=}\NormalTok{ ap}
\end{Highlighting}
\end{Shaded}

\noindent and analogously for \VERB|\DataTypeTok{Either}|. For all the
listings that follow we shall silently hide this type of boiler-plate
code from the paper version where appropriate (it can still be observed
in the Literate Haskell version).

\hypertarget{monadtransformers}{%
\subsection{MonadTransformers}\label{monadtransformers}}

The problem with \VERB|\DataTypeTok{Monad}|s is that they, in general,
do not compose. \VERB|\DataTypeTok{Monad}|
transformers~\cite{Liang:1995:MTM} provide a systematic way to define
structures that represent "a \VERB|\DataTypeTok{Monad}| with a hole"
that allow computations from an inner \VERB|\DataTypeTok{Monad}|
\VERB|\NormalTok{m}| to be \VERB|\NormalTok{lift}|ed through a hole in
an outer \VERB|\DataTypeTok{Monad}\NormalTok{ (t m)}|
(\VERB|\NormalTok{t}| transforms monad \VERB|\NormalTok{m}|, hence
"monad transformer"). The main type class is defined in
\VERB|\DataTypeTok{Control.Monad.Trans.Class}| module of
\texttt{transformers}~\cite{Hackage:transformers0520} package as follows

\begin{Shaded}
\begin{Highlighting}[]
\KeywordTok{class} \DataTypeTok{MonadTrans}\NormalTok{ t }\KeywordTok{where}
\OtherTok{  lift ::}\NormalTok{ (}\DataTypeTok{Monad}\NormalTok{ m) }\OtherTok{=>}\NormalTok{ m a }\OtherTok{->}\NormalTok{ t m a}
\end{Highlighting}
\end{Shaded}

Haskell type class system is not flexible enough to encode the
requirement that \VERB|\NormalTok{t m}| needs to be a
\VERB|\DataTypeTok{Monad}| in a single definition, so it has to be
encoded in every instance by using the following instance schema

\begin{Shaded}
\begin{Highlighting}[]
\KeywordTok{instance} \DataTypeTok{Monad}\NormalTok{ m }\OtherTok{=>} \DataTypeTok{Monad}\NormalTok{ (t m) }\KeywordTok{where}
  \CommentTok{-- ...}
\end{Highlighting}
\end{Shaded}

Different \VERB|\DataTypeTok{MonadTrans}|formers (\VERB|\NormalTok{t1}|,
\VERB|\NormalTok{t2}| \ldots{} \VERB|\NormalTok{tn}|) can then be
composed with an arbitrary \VERB|\DataTypeTok{Monad}|
\VERB|\NormalTok{m}| (usually called "\emph{the} inner
\VERB|\DataTypeTok{Monad}|") using the following scheme

\begin{Shaded}
\begin{Highlighting}[]
\KeywordTok{newtype}\NormalTok{ comp m a }\FunctionTok{=}\NormalTok{ t1 (t2 (}\FunctionTok{..}\NormalTok{ (tn (m a))))}
\end{Highlighting}
\end{Shaded}

\noindent and the whole \VERB|\NormalTok{comp}|osed stack would get a
\VERB|\DataTypeTok{Monad}| instance inferred for it. Popular choices for
the inner \VERB|\DataTypeTok{Monad}| \VERB|\NormalTok{m}| include
\VERB|\DataTypeTok{Identity}| \VERB|\DataTypeTok{Functor}| and
\VERB|\DataTypeTok{IO}| \VERB|\DataTypeTok{Monad}| (see
\cref{sec:imprecise}).

In short, \VERB|\DataTypeTok{MonadTrans}|formers are, pretty much,
composable \VERB|\DataTypeTok{Monad}|ic structures. The following
subsections will provide many example instances. For an in-depth
tutorial readers are referred to \cite{Jones:1995:FPO}
and~\cite{Liang:1995:MTM}.

\hypertarget{identity}{%
\subsubsection{Identity}\label{identity}}

\label{sec:identity-monadtrans}

The simplest \VERB|\DataTypeTok{MonadTrans}|former is
\VERB|\DataTypeTok{IdentityT}| defined in
\VERB|\DataTypeTok{Control.Monad.Trans.Identity}| of
\texttt{transformers}~\cite{Hackage:transformers0520} package
equivalently to

\begin{Shaded}
\begin{Highlighting}[]
\KeywordTok{newtype} \DataTypeTok{IdentityT}\NormalTok{ m a }\FunctionTok{=} \DataTypeTok{IdentityT}
\NormalTok{  \{}\OtherTok{ runIdentityT ::}\NormalTok{ m a \}}

\KeywordTok{instance} \DataTypeTok{MonadTrans} \DataTypeTok{IdentityT} \KeywordTok{where}
\NormalTok{  lift }\FunctionTok{=} \DataTypeTok{IdentityT}

\KeywordTok{instance} \DataTypeTok{Monad}\NormalTok{ m}
      \OtherTok{=>} \DataTypeTok{Pointed}\NormalTok{ (}\DataTypeTok{IdentityT}\NormalTok{ m) }\KeywordTok{where}
  \FunctionTok{pure} \FunctionTok{=}\NormalTok{ lift }\FunctionTok{.} \FunctionTok{pure}

\KeywordTok{instance} \DataTypeTok{Monad}\NormalTok{ m}
      \OtherTok{=>} \DataTypeTok{Monad}\NormalTok{ (}\DataTypeTok{IdentityT}\NormalTok{ m) }\KeywordTok{where}
\NormalTok{  x }\FunctionTok{>>=}\NormalTok{ f }\FunctionTok{=} \DataTypeTok{IdentityT} \FunctionTok{$} \KeywordTok{do}
\NormalTok{    v }\OtherTok{<-}\NormalTok{ runIdentityT x}
\NormalTok{    runIdentityT (f v)}
\end{Highlighting}
\end{Shaded}

\begin{remark}

\label{rem:identity-transformer}

Note that \VERB|\DataTypeTok{IdentityT}|
\VERB|\DataTypeTok{MonadTrans}|former is different from
\VERB|\DataTypeTok{Identity}| \VERB|\DataTypeTok{Monad}| and cannot be
redefined as simply

\begin{Shaded}
\begin{Highlighting}[]
\KeywordTok{type} \DataTypeTok{IdentityT'}\NormalTok{ m a }\FunctionTok{=} \DataTypeTok{Identity}\NormalTok{ (m a)}
\end{Highlighting}
\end{Shaded}

\noindent (even though the data type definition matches exactly) because
\VERB|\DataTypeTok{IdentityT}| "inherits" \VERB|\DataTypeTok{Monad}|
implementation from its argument \VERB|\NormalTok{m}| while
\VERB|\DataTypeTok{Identity}| provides its own. I.e.
\VERB|\DataTypeTok{IdentityT}| is an identity on
\VERB|\DataTypeTok{MonadTrans}|formers while
\VERB|\DataTypeTok{Identity}| is an identity on types.

In particular, for
\VERB|\DataTypeTok{Identity}\NormalTok{ (}\DataTypeTok{Maybe}\NormalTok{ a)}|

\begin{Shaded}
\begin{Highlighting}[]
\FunctionTok{pure} \FunctionTok{==} \DataTypeTok{Identity}
\end{Highlighting}
\end{Shaded}

\noindent while for
\VERB|\DataTypeTok{IdentityT} \DataTypeTok{Maybe}\NormalTok{ a}|

\begin{Shaded}
\begin{Highlighting}[]
\FunctionTok{pure} \FunctionTok{==} \DataTypeTok{IdentityT} \FunctionTok{.} \FunctionTok{pure} \FunctionTok{==} \DataTypeTok{IdentityT} \FunctionTok{.} \DataTypeTok{Just}
\end{Highlighting}
\end{Shaded}

\end{remark}

\hypertarget{maybe}{%
\subsubsection{Maybe}\label{maybe}}

Transformer version of \VERB|\DataTypeTok{Maybe}| called
\VERB|\DataTypeTok{MaybeT}| is defined in
\VERB|\DataTypeTok{Control.Monad.Trans.Maybe}| from
\texttt{transformers}~\cite{Hackage:transformers0520} package
equivalently to

\begin{Shaded}
\begin{Highlighting}[]
\KeywordTok{newtype} \DataTypeTok{MaybeT}\NormalTok{ m a }\FunctionTok{=} \DataTypeTok{MaybeT}
\NormalTok{  \{}\OtherTok{ runMaybeT ::}\NormalTok{ m (}\DataTypeTok{Maybe}\NormalTok{ a) \}}

\KeywordTok{instance} \DataTypeTok{MonadTrans} \DataTypeTok{MaybeT} \KeywordTok{where}
\NormalTok{  lift }\FunctionTok{=} \DataTypeTok{MaybeT} \FunctionTok{.}\NormalTok{ liftM }\DataTypeTok{Just}

\KeywordTok{instance} \DataTypeTok{Monad}\NormalTok{ m}
      \OtherTok{=>} \DataTypeTok{Pointed}\NormalTok{ (}\DataTypeTok{MaybeT}\NormalTok{ m) }\KeywordTok{where}
  \FunctionTok{pure} \FunctionTok{=}\NormalTok{ lift }\FunctionTok{.} \FunctionTok{pure}

\KeywordTok{instance} \DataTypeTok{Monad}\NormalTok{ m}
      \OtherTok{=>} \DataTypeTok{Monad}\NormalTok{ (}\DataTypeTok{MaybeT}\NormalTok{ m) }\KeywordTok{where}
\NormalTok{  x }\FunctionTok{>>=}\NormalTok{ f }\FunctionTok{=} \DataTypeTok{MaybeT} \FunctionTok{$} \KeywordTok{do}
\NormalTok{    v }\OtherTok{<-}\NormalTok{ runMaybeT x}
    \KeywordTok{case}\NormalTok{ v }\KeywordTok{of}
      \DataTypeTok{Nothing} \OtherTok{->} \FunctionTok{pure} \DataTypeTok{Nothing}
      \DataTypeTok{Just}\NormalTok{ y  }\OtherTok{->}\NormalTok{ runMaybeT (f y)}

\KeywordTok{instance} \DataTypeTok{MonadFail}\NormalTok{ m}
      \OtherTok{=>} \DataTypeTok{MonadFail}\NormalTok{ (}\DataTypeTok{MaybeT}\NormalTok{ m) }\KeywordTok{where}
  \FunctionTok{fail}\NormalTok{ _ }\FunctionTok{=} \DataTypeTok{MaybeT}\NormalTok{ (}\FunctionTok{pure} \DataTypeTok{Nothing}\NormalTok{)}
\end{Highlighting}
\end{Shaded}

\hypertarget{except}{%
\subsubsection{Except}\label{except}}

\label{sec:either-monadtrans}

Transformer version of \VERB|\DataTypeTok{Either}| for historical
reasons bears a name of \VERB|\DataTypeTok{ExceptT}| and is defined in
\VERB|\DataTypeTok{Control.Monad.Trans.Except}| from
\texttt{transformers}~\cite{Hackage:transformers0520} package
equivalently to

\begin{Shaded}
\begin{Highlighting}[]
\KeywordTok{newtype} \DataTypeTok{ExceptT}\NormalTok{ e m a}
  \FunctionTok{=} \DataTypeTok{ExceptT}\NormalTok{ \{ runExceptT}
\OtherTok{           ::}\NormalTok{ m (}\DataTypeTok{Either}\NormalTok{ e a) \}}

\KeywordTok{instance} \DataTypeTok{MonadTrans}\NormalTok{ (}\DataTypeTok{ExceptT}\NormalTok{ e) }\KeywordTok{where}
\NormalTok{  lift }\FunctionTok{=} \DataTypeTok{ExceptT} \FunctionTok{.}\NormalTok{ liftM }\DataTypeTok{Right}

\KeywordTok{instance} \DataTypeTok{Pointed}\NormalTok{ m}
      \OtherTok{=>} \DataTypeTok{Pointed}\NormalTok{ (}\DataTypeTok{ExceptT}\NormalTok{ e m) }\KeywordTok{where}
  \FunctionTok{pure}\NormalTok{ a }\FunctionTok{=} \DataTypeTok{ExceptT} \FunctionTok{$} \FunctionTok{pure}\NormalTok{ (}\DataTypeTok{Right}\NormalTok{ a)}

\KeywordTok{instance} \DataTypeTok{Monad}\NormalTok{ m}
      \OtherTok{=>} \DataTypeTok{Monad}\NormalTok{ (}\DataTypeTok{ExceptT}\NormalTok{ e m) }\KeywordTok{where}
\NormalTok{   m }\FunctionTok{>>=}\NormalTok{ k }\FunctionTok{=} \DataTypeTok{ExceptT} \FunctionTok{$} \KeywordTok{do}
\NormalTok{     a }\OtherTok{<-}\NormalTok{ runExceptT m}
     \KeywordTok{case}\NormalTok{ a }\KeywordTok{of}
       \DataTypeTok{Left}\NormalTok{  e }\OtherTok{->} \FunctionTok{pure}\NormalTok{ (}\DataTypeTok{Left}\NormalTok{ e)}
       \DataTypeTok{Right}\NormalTok{ x }\OtherTok{->}\NormalTok{ runExceptT (k x)}

\KeywordTok{instance} \DataTypeTok{MonadFail}\NormalTok{ m}
      \OtherTok{=>} \DataTypeTok{MonadFail}\NormalTok{ (}\DataTypeTok{ExceptT}\NormalTok{ e m) }\KeywordTok{where}
   \FunctionTok{fail} \FunctionTok{=} \DataTypeTok{ExceptT} \FunctionTok{.} \FunctionTok{fail}
\end{Highlighting}
\end{Shaded}

The main attraction of \VERB|\DataTypeTok{ExceptT}| for the purposes of
this article is the fact that it provides its own non-imprecise
non-dynamic-dispatching \VERB|\NormalTok{throw}| and
\VERB|\FunctionTok{catch}| operators defined as

\begin{Shaded}
\begin{Highlighting}[]
\OtherTok{throwE ::}\NormalTok{ (}\DataTypeTok{Monad}\NormalTok{ m) }\OtherTok{=>}\NormalTok{ e }\OtherTok{->} \DataTypeTok{ExceptT}\NormalTok{ e m a}
\NormalTok{throwE }\FunctionTok{=} \DataTypeTok{ExceptT} \FunctionTok{.} \FunctionTok{pure} \FunctionTok{.} \DataTypeTok{Left}

\OtherTok{catchE ::}\NormalTok{ (}\DataTypeTok{Monad}\NormalTok{ m) }\OtherTok{=>}
    \DataTypeTok{ExceptT}\NormalTok{ e m a}
    \OtherTok{->}\NormalTok{ (e }\OtherTok{->} \DataTypeTok{ExceptT}\NormalTok{ f m a)}
    \OtherTok{->} \DataTypeTok{ExceptT}\NormalTok{ f m a}
\NormalTok{m }\OtherTok{`catchE`}\NormalTok{ h }\FunctionTok{=} \DataTypeTok{ExceptT} \FunctionTok{$} \KeywordTok{do}
\NormalTok{    a }\OtherTok{<-}\NormalTok{ runExceptT m}
    \KeywordTok{case}\NormalTok{ a }\KeywordTok{of}
        \DataTypeTok{Left}\NormalTok{  l }\OtherTok{->}\NormalTok{ runExceptT (h l)}
        \DataTypeTok{Right}\NormalTok{ r }\OtherTok{->} \FunctionTok{pure}\NormalTok{ (}\DataTypeTok{Right}\NormalTok{ r)}
\end{Highlighting}
\end{Shaded}

There also exists deprecated \VERB|\DataTypeTok{ErrorT}| (defined in
\VERB|\DataTypeTok{Control.Monad.Trans.Error}| from
\texttt{transformers} package) which at the time of writing has exactly
the same definition as \VERB|\DataTypeTok{ExceptT}|

\begin{Shaded}
\begin{Highlighting}[]
\KeywordTok{newtype} \DataTypeTok{ErrorT}\NormalTok{ e m a}
  \FunctionTok{=} \DataTypeTok{ErrorT}\NormalTok{ \{ runErrorT}
\OtherTok{          ::}\NormalTok{ m (}\DataTypeTok{Either}\NormalTok{ e a) \}}
\end{Highlighting}
\end{Shaded}

\noindent but its instances require type class
\VERB|\DataTypeTok{Exception}| (see \cref{sec:exception}) from its
argument \VERB|\NormalTok{e}|. Older versions of \texttt{transformers}
package made this requirement in the definition of
\VERB|\DataTypeTok{ErrorT}|

\newpage

\begin{Shaded}
\begin{Highlighting}[]
\KeywordTok{newtype} \DataTypeTok{ErrorT}\NormalTok{ e m a}
  \FunctionTok{=} \DataTypeTok{Exception}\NormalTok{ e }\OtherTok{=>}
    \DataTypeTok{ErrorT}\NormalTok{ \{ runErrorT}
\OtherTok{          ::}\NormalTok{ m (}\DataTypeTok{Either}\NormalTok{ e a) \}}
\end{Highlighting}
\end{Shaded}

\noindent but that mechanism itself was deprecated awhile ago.

\hypertarget{reader-and-state}{%
\subsubsection{Reader and State}\label{reader-and-state}}

\label{sec:reader-state-monadtrans}

While there seems to be no way to directly use
\VERB|\DataTypeTok{Reader}| and \VERB|\DataTypeTok{State}|
\VERB|\DataTypeTok{Monad}|s for error handling, these structures are
used in \VERB|\DataTypeTok{IO}| \VERB|\DataTypeTok{Monad}| of
\cref{sec:imprecise} and parser combinators of
\cref{sec:parser-combinators}. This seems to be as good place as any to
define them.

\VERB|\DataTypeTok{Reader}| \VERB|\DataTypeTok{Monad}| is defined in
\VERB|\DataTypeTok{Control.Monad.Trans.Reader}| module of
\texttt{transformers}~\cite{Hackage:transformers0520} package
equivalently to

\begin{Shaded}
\begin{Highlighting}[]
\KeywordTok{type} \DataTypeTok{Reader}\NormalTok{ s }\FunctionTok{=} \DataTypeTok{ReaderT}\NormalTok{ s }\DataTypeTok{Identity}

\KeywordTok{newtype} \DataTypeTok{ReaderT}\NormalTok{ s m a }\FunctionTok{=} \DataTypeTok{ReaderT}\NormalTok{ \{}\OtherTok{ runReaderT ::}\NormalTok{ s }\OtherTok{->}\NormalTok{ m a \}}

\KeywordTok{instance} \DataTypeTok{MonadTrans}\NormalTok{ (}\DataTypeTok{ReaderT}\NormalTok{ s) }\KeywordTok{where}
\NormalTok{  lift m }\FunctionTok{=} \DataTypeTok{ReaderT} \FunctionTok{$}\NormalTok{ \textbackslash{}_ }\OtherTok{->}\NormalTok{ m}

\KeywordTok{instance} \DataTypeTok{Pointed}\NormalTok{ m }\OtherTok{=>} \DataTypeTok{Pointed}\NormalTok{ (}\DataTypeTok{ReaderT}\NormalTok{ s m) }\KeywordTok{where}
  \FunctionTok{pure}\NormalTok{ a }\FunctionTok{=} \DataTypeTok{ReaderT} \FunctionTok{$}\NormalTok{ \textbackslash{}_ }\OtherTok{->} \FunctionTok{pure}\NormalTok{ a}

\KeywordTok{instance} \DataTypeTok{Monad}\NormalTok{ m }\OtherTok{=>} \DataTypeTok{Monad}\NormalTok{ (}\DataTypeTok{ReaderT}\NormalTok{ s m) }\KeywordTok{where}
\NormalTok{  m }\FunctionTok{>>=}\NormalTok{ k  }\FunctionTok{=} \DataTypeTok{ReaderT} \FunctionTok{$}\NormalTok{ \textbackslash{}s }\OtherTok{->} \KeywordTok{do}
\NormalTok{    a }\OtherTok{<-}\NormalTok{ runReaderT m s}
\NormalTok{    runReaderT (k a) s}

\KeywordTok{instance} \DataTypeTok{MonadFail}\NormalTok{ m }\OtherTok{=>} \DataTypeTok{MonadFail}\NormalTok{ (}\DataTypeTok{ReaderT}\NormalTok{ s m) }\KeywordTok{where}
  \FunctionTok{fail}\NormalTok{ str }\FunctionTok{=} \DataTypeTok{ReaderT} \FunctionTok{$}\NormalTok{ \textbackslash{}_ }\OtherTok{->} \FunctionTok{fail}\NormalTok{ str}
\end{Highlighting}
\end{Shaded}

Meanwhile, \VERB|\DataTypeTok{State}| \VERB|\DataTypeTok{Monad}| is
defined in \VERB|\DataTypeTok{Control.Monad.Trans.State.Lazy}| and
\VERB|\DataTypeTok{Control.Monad.Trans.State.Strict}| modules (the
difference between them does not matter for the purposes of this
article, so we shall ignore it) from
\texttt{transformers}~\cite{Hackage:transformers0520} package
equivalently to

\begin{Shaded}
\begin{Highlighting}[]
\KeywordTok{type} \DataTypeTok{State}\NormalTok{ s }\FunctionTok{=} \DataTypeTok{StateT}\NormalTok{ s }\DataTypeTok{Identity}

\KeywordTok{newtype} \DataTypeTok{StateT}\NormalTok{ s m a }\FunctionTok{=} \DataTypeTok{StateT}\NormalTok{ \{}\OtherTok{ runStateT ::}\NormalTok{ s }\OtherTok{->}\NormalTok{ m (a, s) \}}

\KeywordTok{instance} \DataTypeTok{MonadTrans}\NormalTok{ (}\DataTypeTok{StateT}\NormalTok{ s) }\KeywordTok{where}
\NormalTok{  lift m }\FunctionTok{=} \DataTypeTok{StateT} \FunctionTok{$}\NormalTok{ \textbackslash{}s }\OtherTok{->} \KeywordTok{do}
\NormalTok{    a }\OtherTok{<-}\NormalTok{ m}
    \FunctionTok{pure}\NormalTok{ (a, s)}

\KeywordTok{instance} \DataTypeTok{Pointed}\NormalTok{ m }\OtherTok{=>} \DataTypeTok{Pointed}\NormalTok{ (}\DataTypeTok{StateT}\NormalTok{ s m) }\KeywordTok{where}
  \FunctionTok{pure}\NormalTok{ a }\FunctionTok{=} \DataTypeTok{StateT} \FunctionTok{$}\NormalTok{ \textbackslash{}s }\OtherTok{->} \FunctionTok{pure}\NormalTok{ (a, s)}

\KeywordTok{instance} \DataTypeTok{Monad}\NormalTok{ m }\OtherTok{=>} \DataTypeTok{Monad}\NormalTok{ (}\DataTypeTok{StateT}\NormalTok{ s m) }\KeywordTok{where}
\NormalTok{  m }\FunctionTok{>>=}\NormalTok{ k  }\FunctionTok{=} \DataTypeTok{StateT} \FunctionTok{$}\NormalTok{ \textbackslash{}s }\OtherTok{->} \KeywordTok{do}
\NormalTok{    (a, s') }\OtherTok{<-}\NormalTok{ runStateT m s}
\NormalTok{    runStateT (k a) s'}

\KeywordTok{instance} \DataTypeTok{MonadFail}\NormalTok{ m }\OtherTok{=>} \DataTypeTok{MonadFail}\NormalTok{ (}\DataTypeTok{StateT}\NormalTok{ s m) }\KeywordTok{where}
  \FunctionTok{fail}\NormalTok{ str }\FunctionTok{=} \DataTypeTok{StateT} \FunctionTok{$}\NormalTok{ \textbackslash{}_ }\OtherTok{->} \FunctionTok{fail}\NormalTok{ str}
\end{Highlighting}
\end{Shaded}

Both structures provide \VERB|\DataTypeTok{Monad}|ic structures that
handle state. \VERB|\DataTypeTok{ReaderT}| simply applies variable
\VERB|\NormalTok{s}| throughout its whole computation via its
\VERB|\NormalTok{(}\FunctionTok{>>=}\NormalTok{)}| operator thus
supplying computations with a \emph{context} (i.e. read-only
\emph{state}). Meanwhile, \VERB|\DataTypeTok{StateT}| chains its
\VERB|\NormalTok{s}| between computations, thus providing computations
with a (read-write) \emph{state}.

\hypertarget{imprecise-exceptions}{%
\subsection{Imprecise exceptions}\label{imprecise-exceptions}}

\label{sec:imprecise}

As we mentioned in the introduction, GHC implements \emph{imprecise
exceptions} mechanism proposed in \cite{PeytonJones:1999:SIE}. Such
exceptions look superficially similar to those of C++/Java/Python/etc
but differ in two important aspects.

Firstly, GHC imprecise exceptions in pure computations are completely
imprecise. That is, evaluation of
\VERB|\NormalTok{(a }\OtherTok{`op`}\NormalTok{ b)}| with
\VERB|\NormalTok{a}| raising \VERB|\NormalTok{e}| and
\VERB|\NormalTok{b}| raising \VERB|\NormalTok{f}| (and assuming
\VERB|\NormalTok{op}| can evaluate either argument first) can raise
either or even both (on different evaluations) of \VERB|\NormalTok{e}|
and \VERB|\NormalTok{f}|. Haskell is not the only language that does
this, C++, for instance, defines \emph{sequence points} that serve the
same purpose~\cite{CFAQ:SeqPoints}. However, in GHC the order in which
exception are raised is limited only by data dependencies, while C++'s
sequence points add some more ordering on top.

Secondly, the C++/Java/Python exceptions have dynamic dispatch builtin,
while GHC's dynamically dispatched exceptions are implemented as a
library on top of statically dispatched exceptions. To be more specific

\begin{itemize}
\item
  on the base level GHC runtime defines
  \VERB|\NormalTok{raise}\FunctionTok{#}| and
  \VERB|\NormalTok{catch}\FunctionTok{#}| operations for which
  \VERB|\NormalTok{raise}\FunctionTok{#}| "simply"\footnote{\label{fn:simply}We
    put "simply" and "just" into quotes since unwinding of the stack
    must unwind into the lexically correct handler which is nontrivial
    in a lazy language like Haskell where thunks can be evaluated in an
    environment different from the one they were created in. In short,
    thunks must capture exception handlers as well as variables.}
  unwinds the stack to the closest
  \VERB|\NormalTok{catch}\FunctionTok{#}| (i.e.
  \VERB|\NormalTok{raise}\FunctionTok{#}| is "just"\cref{fn:simply} a
  \texttt{GOTO}; \VERB|\NormalTok{cast}|ing,
  re-\VERB|\NormalTok{raise}|ing, \VERB|\NormalTok{finally}|, etc are
  left for the libraries to implement and are not builtins),
\item
  on top of that GHC libraries then provide dynamically dispatched
  exceptions by \VERB|\NormalTok{cast}|ing elements of
  \VERB|\DataTypeTok{Typeable}| types from/to
  \VERB|\DataTypeTok{SomeException}| existential
  type~\cite{Marlow:2006:EDH}.
\end{itemize}

In the following subsections we shall discuss the details of the actual
implementation.

\hypertarget{io}{%
\subsubsection{IO}\label{io}}

\label{sec:io}

GHC defines the mystical \VERB|\DataTypeTok{IO}|
\VERB|\DataTypeTok{Monad}| in \VERB|\DataTypeTok{GHC.Types}| (the types)
and \VERB|\DataTypeTok{GHC.Base}| (the instances), pretty much, as a
\VERB|\DataTypeTok{State}| \VERB|\DataTypeTok{Monad}| (see
\cref{sec:reader-state-monadtrans}) on
\VERB|\DataTypeTok{State}\FunctionTok{#} \DataTypeTok{RealWorld}|
(definitions of both of which are beyond the scope of this article)

\begin{Shaded}
\begin{Highlighting}[]
\KeywordTok{type} \DataTypeTok{IO}\FunctionTok{#}\NormalTok{ a }\FunctionTok{=} \DataTypeTok{State}\FunctionTok{#} \DataTypeTok{RealWorld}
          \OtherTok{->}\NormalTok{ (}\FunctionTok{#} \DataTypeTok{State}\FunctionTok{#} \DataTypeTok{RealWorld}\NormalTok{, a }\FunctionTok{#}\NormalTok{)}

\KeywordTok{newtype} \DataTypeTok{IO}\NormalTok{ a }\FunctionTok{=} \DataTypeTok{IO}\NormalTok{ \{}\OtherTok{ runIO ::} \DataTypeTok{IO}\FunctionTok{#}\NormalTok{ a \}}

\KeywordTok{instance} \DataTypeTok{Pointed} \DataTypeTok{IO} \KeywordTok{where}
  \FunctionTok{pure}\NormalTok{ a }\FunctionTok{=} \DataTypeTok{IO} \FunctionTok{$}\NormalTok{ \textbackslash{}s }\OtherTok{->}\NormalTok{ (}\FunctionTok{#}\NormalTok{ s, a }\FunctionTok{#}\NormalTok{)}

\KeywordTok{instance} \DataTypeTok{Monad} \DataTypeTok{IO} \KeywordTok{where}
\NormalTok{  m }\FunctionTok{>>=}\NormalTok{ f }\FunctionTok{=} \DataTypeTok{IO} \FunctionTok{$}\NormalTok{ \textbackslash{}s }\OtherTok{->} \KeywordTok{case}\NormalTok{ runIO m s }\KeywordTok{of}
\NormalTok{    (}\FunctionTok{#}\NormalTok{ s', a }\FunctionTok{#}\NormalTok{) }\OtherTok{->}\NormalTok{ runIO (f a) s'}
\end{Highlighting}
\end{Shaded}

The \VERB|\DataTypeTok{IO}\FunctionTok{#}| definition given above is not
actually in GHC but without it all of the definitions below become
unreadable. We also renamed \VERB|\NormalTok{unIO}| to
\VERB|\NormalTok{runIO}| for uniformity with \VERB|\DataTypeTok{State}|.
Note however, that we did not swap the elements of the result tuple of
\VERB|\DataTypeTok{IO}\FunctionTok{#}| to match those of
\VERB|\DataTypeTok{State}| since that would make it incompatible with
GHC runtime we reuse in Literate Haskell version.

\begin{remark}

\label{rem:io-caveats}

Note that \VERB|\DataTypeTok{IO}| is not a proper
\VERB|\DataTypeTok{Monad}| since it cannot satisfy the laws simply for
the fact that \VERB|\DataTypeTok{RealWorld}| cannot have an
equality.\footnote{Although \VERB|\DataTypeTok{IO}| can be reformulated
  as a free \VERB|\DataTypeTok{Monad}| made of "requests to the
  interpreter" and continuations if one is willing to forget about the
  internal structure of the
  \VERB|\DataTypeTok{RealWorld}|~\cite{Kmett:2011:FMFL}.}

In this article, however, for the purposes of formal arguments involving
\VERB|\DataTypeTok{IO}| we shall treat \VERB|\DataTypeTok{IO}| as if it
was just a \VERB|\DataTypeTok{State}| over some state type with some
simple denotational semantics (although, possibly unknown value). This,
of course, immediately disqualifies our proofs for
\VERB|\DataTypeTok{IO}| from using non-determinism, hence, for instance,
we will not be able to prove things about imprecise exceptions or
threads.

The alternative would be to split every lemma and theorem mentioning
\VERB|\DataTypeTok{IO}| into two: one for a
\VERB|\DataTypeTok{RawMonad}| (\VERB|\DataTypeTok{Monad}| without laws)
for cases mentioning \VERB|\DataTypeTok{IO}|, and one for
\VERB|\DataTypeTok{Monad}| for all other cases. This would make a very
little practical sense for this article since we will not attempt proofs
involving non-determinism anyway.

\end{remark}

\hypertarget{raise-and-catch}{%
\subsubsection{raise\# and catch\#}\label{raise-and-catch}}

\label{sec:raise-catch}

Primitive \VERB|\NormalTok{raise}\FunctionTok{#}| and
\VERB|\NormalTok{catch}\FunctionTok{#}| operations are "defined" (those,
of course, are just stubs to be replaced by references to the actual
implementations in GHC runtime) in \VERB|\DataTypeTok{GHC.Prim}| module
like follows

\begin{Shaded}
\begin{Highlighting}[]
\NormalTok{raise}\FunctionTok{#}\OtherTok{ ::}\NormalTok{ a }\OtherTok{->}\NormalTok{ b}
\NormalTok{raise}\FunctionTok{#} \FunctionTok{=}\NormalTok{ raise}\FunctionTok{#}

\NormalTok{catch}\FunctionTok{#}\OtherTok{ ::} \DataTypeTok{IO}\FunctionTok{#}\NormalTok{ a }\OtherTok{->}\NormalTok{ (b }\OtherTok{->} \DataTypeTok{IO}\FunctionTok{#}\NormalTok{ a)}
       \OtherTok{->} \DataTypeTok{IO}\FunctionTok{#}\NormalTok{ a}
\NormalTok{catch}\FunctionTok{#} \FunctionTok{=}\NormalTok{ catch}\FunctionTok{#}
\end{Highlighting}
\end{Shaded}

Evaluating \VERB|\NormalTok{raise}\FunctionTok{#}|
"simply"\cref{fn:simply} unwinds computation stack to the point of the
closet \VERB|\NormalTok{catch}\FunctionTok{#}| with the appropriate type
and applies raised value to the second argument of the latter. Note,
however, that while the type of \VERB|\NormalTok{raise}\FunctionTok{#}|
permits its use anywhere in the program,
\VERB|\NormalTok{catch}\FunctionTok{#}| is sandboxed to
\VERB|\DataTypeTok{IO}\FunctionTok{#}| on the lowest observable level
and GHC provides no "\VERB|\NormalTok{unsafeCatch}|". This allows GHC to
perform many useful optimizations that influence evaluation order
without exposing pure computations to non-determinism.

\hypertarget{typeable}{%
\subsubsection{Typeable}\label{typeable}}

\label{sec:typeable}

GHC implements dynamic casting with \VERB|\DataTypeTok{Typeable}| type
class. The details of its actual implementation are beyond the scope of
this article. For our purposes it suffices to say that it is a type
class of types that have type representations that can be compared at
runtime

\begin{Shaded}
\begin{Highlighting}[]
\KeywordTok{class} \DataTypeTok{Typeable}\NormalTok{ a }\KeywordTok{where}
  \CommentTok{-- magic beyond the scope of this article}
\end{Highlighting}
\end{Shaded}

\noindent and it provides a \VERB|\NormalTok{cast}| operation with the
following type signature that shows that it compares said
representations of types of its argument and result and either returns
its argument value wrapped in \VERB|\DataTypeTok{Just}| constructor when
the types match or \VERB|\DataTypeTok{Nothing}| else

\begin{Shaded}
\begin{Highlighting}[]
\OtherTok{cast ::} \KeywordTok{forall}\NormalTok{ a b}
      \FunctionTok{.}\NormalTok{ (}\DataTypeTok{Typeable}\NormalTok{ a, }\DataTypeTok{Typeable}\NormalTok{ b)}
     \OtherTok{=>}\NormalTok{ a }\OtherTok{->} \DataTypeTok{Maybe}\NormalTok{ b}
\end{Highlighting}
\end{Shaded}

Interested readers should inspect the source code of
\VERB|\DataTypeTok{Data.Typeable}| module of
\texttt{base}~\cite{Hackage:base4900}.

\hypertarget{exception}{%
\subsubsection{Exception}\label{exception}}

\label{sec:exception}

On top of \VERB|\DataTypeTok{Typeable}| in
\VERB|\DataTypeTok{GHC.Exception}| module of
\texttt{base}~\cite{Hackage:base4900} GHC provides the
\VERB|\DataTypeTok{Exception}| type class that casts values to and from
\VERB|\DataTypeTok{SomeException}| existential type (the following
syntactic \VERB|\KeywordTok{forall}| is type-theoretic
\VERB|\NormalTok{exists}|, historic reasons)

\begin{Shaded}
\begin{Highlighting}[]
\KeywordTok{data} \DataTypeTok{SomeException} \FunctionTok{=} \KeywordTok{forall}\NormalTok{ e}\FunctionTok{.} \DataTypeTok{Exception}\NormalTok{ e}
                  \OtherTok{=>} \DataTypeTok{SomeException}\NormalTok{ e}

\KeywordTok{class}\NormalTok{ (}\DataTypeTok{Typeable}\NormalTok{ e, }\DataTypeTok{Show}\NormalTok{ e) }\OtherTok{=>} \DataTypeTok{Exception}\NormalTok{ e }\KeywordTok{where}
\OtherTok{  toException   ::}\NormalTok{ e }\OtherTok{->} \DataTypeTok{SomeException}
\OtherTok{  fromException ::} \DataTypeTok{SomeException}\OtherTok{->} \DataTypeTok{Maybe}\NormalTok{ e}

\NormalTok{  toException }\FunctionTok{=} \DataTypeTok{SomeException}
\NormalTok{  fromException (}\DataTypeTok{SomeException}\NormalTok{ e) }\FunctionTok{=}\NormalTok{ cast e}

\KeywordTok{instance} \DataTypeTok{Show} \DataTypeTok{SomeException} \KeywordTok{where}
  \FunctionTok{show}\NormalTok{ (}\DataTypeTok{SomeException}\NormalTok{ e) }\FunctionTok{=} \FunctionTok{show}\NormalTok{ e}

\KeywordTok{instance} \DataTypeTok{Exception} \DataTypeTok{SomeException} \KeywordTok{where}
\NormalTok{  toException }\FunctionTok{=} \FunctionTok{id}
\NormalTok{  fromException x }\FunctionTok{=} \DataTypeTok{Just}\NormalTok{ x}
\end{Highlighting}
\end{Shaded}

\hypertarget{throw-and-catch}{%
\subsubsection{throw and catch}\label{throw-and-catch}}

Finally, \VERB|\NormalTok{throw}| and \VERB|\FunctionTok{catch}|
operators defined in \VERB|\DataTypeTok{GHC.Exception}| module of
\texttt{base}~\cite{Hackage:base4900} use all of the above to implement
dynamic dispatch of exceptions.

The \VERB|\NormalTok{throw}| operator simply wraps given exception into
\VERB|\DataTypeTok{SomeException}| and
\VERB|\NormalTok{raise}\FunctionTok{#}|s

\begin{Shaded}
\begin{Highlighting}[]
\OtherTok{throw ::} \DataTypeTok{Exception}\NormalTok{ e }\OtherTok{=>}\NormalTok{ e }\OtherTok{->}\NormalTok{ a}
\NormalTok{throw e }\FunctionTok{=}\NormalTok{ raise}\FunctionTok{#}\NormalTok{ (toException e)}
\end{Highlighting}
\end{Shaded}

The \VERB|\NormalTok{catchException}| operator defined in
\VERB|\DataTypeTok{GHC.IO}| does the actual dynamic dispatch

\begin{itemize}
\item
  it \VERB|\NormalTok{catch}\FunctionTok{#}|es an exception produced by
  its first argument ("computation"),
\item
  tries to \VERB|\NormalTok{cast}| it to a type expected by its second
  argument ("handler") and either calls the latter on success, or
  \VERB|\NormalTok{raise}\FunctionTok{#}|s again on failure.
\end{itemize}

\begin{Shaded}
\begin{Highlighting}[]
\OtherTok{catchException ::} \DataTypeTok{Exception}\NormalTok{ e}
               \OtherTok{=>} \DataTypeTok{IO}\NormalTok{ a }\OtherTok{->}\NormalTok{ (e }\OtherTok{->} \DataTypeTok{IO}\NormalTok{ a)}
               \OtherTok{->} \DataTypeTok{IO}\NormalTok{ a}
\NormalTok{catchException (}\DataTypeTok{IO}\NormalTok{ io) handler}
  \FunctionTok{=} \DataTypeTok{IO} \FunctionTok{$}\NormalTok{ catch}\FunctionTok{#}\NormalTok{ io handler'}
  \KeywordTok{where}
\NormalTok{    handler' e }\FunctionTok{=} \KeywordTok{case}\NormalTok{ fromException e }\KeywordTok{of}
          \DataTypeTok{Just}\NormalTok{  f }\OtherTok{->}\NormalTok{ runIO (handler f)}
          \DataTypeTok{Nothing} \OtherTok{->}\NormalTok{ raiseIO}\FunctionTok{#}\NormalTok{ e}
\end{Highlighting}
\end{Shaded}

The \VERB|\FunctionTok{catch}| operator simply calls
\VERB|\NormalTok{catchException}| after forcing its first argument into
a thunk with \VERB|\NormalTok{lazy}| operator (this wrapping is
necessary to prevent GHC from performing strictness analysis on the
"computation" to prevent its evaluation before the exception is even
raised; this fact can be ignored for the purposes of this article) which
is yet another special GHC runtime function (this time, extentionally
equal to its definition, i.e. identity).

\begin{Shaded}
\begin{Highlighting}[]
\OtherTok{lazy ::}\NormalTok{ a }\OtherTok{->}\NormalTok{ a}
\NormalTok{lazy x }\FunctionTok{=}\NormalTok{ x}

\FunctionTok{catch}\OtherTok{   ::} \DataTypeTok{Exception}\NormalTok{ e}
        \OtherTok{=>} \DataTypeTok{IO}\NormalTok{ a }\OtherTok{->}\NormalTok{ (e }\OtherTok{->} \DataTypeTok{IO}\NormalTok{ a)}
        \OtherTok{->} \DataTypeTok{IO}\NormalTok{ a}
\FunctionTok{catch}\NormalTok{ act }\FunctionTok{=}\NormalTok{ catchException (lazy act)}
\end{Highlighting}
\end{Shaded}

That is, \VERB|\FunctionTok{catch}| is extentionally equal to
\VERB|\NormalTok{catchException}|.
\VERB|\DataTypeTok{Control.Exception}| module of \texttt{base} simply
reexports \VERB|\NormalTok{throw}|, \VERB|\FunctionTok{catch}|, and
\VERB|\DataTypeTok{Exception}| type class and implements a bunch of
practically convenient combinators using them.

We should also mention that older versions of \texttt{base} package had
another special \VERB|\FunctionTok{catch}| that handled only
\VERB|\DataTypeTok{IOError}|s defined in \VERB|\DataTypeTok{Prelude}|
and \VERB|\DataTypeTok{System.IO.Error}| respectively. Those were
deprecated in 2011 and as of writing of this article are completely gone
from current version of \texttt{base}. But they are are occasionally
mentioned in tutorials, usually in the context of "don't use
\VERB|\FunctionTok{catch}| from \VERB|\DataTypeTok{Prelude}|, use the
one from \VERB|\DataTypeTok{Control.Exception}|", nowadays the
\VERB|\FunctionTok{catch}| from \VERB|\DataTypeTok{Prelude}| \emph{is}
the \VERB|\FunctionTok{catch}| from
\VERB|\DataTypeTok{Control.Exception}|.

\hypertarget{error-and-undefined}{%
\subsubsection{error and undefined}\label{error-and-undefined}}

\label{sec:error-undefined}

\VERB|\FunctionTok{error}| and \VERB|\FunctionTok{undefined}| primitives
are defined in \VERB|\DataTypeTok{GHC.Err}| of \texttt{base} as follows

\begin{Shaded}
\begin{Highlighting}[]
\KeywordTok{newtype} \DataTypeTok{ErrorCall} \FunctionTok{=} \DataTypeTok{ErrorCall} \DataTypeTok{String}

\KeywordTok{instance} \DataTypeTok{Exception} \DataTypeTok{ErrorCall} \KeywordTok{where}

\FunctionTok{error}\OtherTok{ ::} \DataTypeTok{String} \OtherTok{->}\NormalTok{ a}
\FunctionTok{error}\NormalTok{ s }\FunctionTok{=}\NormalTok{ throw (}\DataTypeTok{ErrorCall}\NormalTok{ s)}

\FunctionTok{undefined}\OtherTok{ ::} \KeywordTok{forall}\NormalTok{ a }\FunctionTok{.}\NormalTok{ a}
\FunctionTok{undefined} \FunctionTok{=} \FunctionTok{error} \StringTok{"Prelude.undefined"}
\end{Highlighting}
\end{Shaded}

Actually, this implementation is taken from the older version of
\texttt{base}, modern version also implements call stack capture, which
is beyond the scope of this article. Interested readers are referred to
the source code of \VERB|\DataTypeTok{GHC.Err}|.

\hypertarget{precise-raiseio-and-throwio}{%
\subsection{Precise raiseIO\# and
throwIO}\label{precise-raiseio-and-throwio}}

Besides imprecise exceptions GHC's \VERB|\DataTypeTok{IO}| also has
operators for precise exceptions a-la \VERB|\DataTypeTok{ExceptT}|
defined in \VERB|\DataTypeTok{GHC.Prim}| and
\VERB|\DataTypeTok{GHC.Exception}| as follows

\begin{Shaded}
\begin{Highlighting}[]
\NormalTok{raiseIO}\FunctionTok{#}\OtherTok{ ::}\NormalTok{ a }\OtherTok{->} \DataTypeTok{IO}\FunctionTok{#}\NormalTok{ b}
\NormalTok{raiseIO}\FunctionTok{#} \FunctionTok{=}\NormalTok{ raiseIO}\FunctionTok{#}
\end{Highlighting}
\end{Shaded}

\begin{Shaded}
\begin{Highlighting}[]
\OtherTok{throwIO ::} \DataTypeTok{Exception}\NormalTok{ e }\OtherTok{=>}\NormalTok{ e }\OtherTok{->} \DataTypeTok{IO}\NormalTok{ a}
\NormalTok{throwIO e }\FunctionTok{=} \DataTypeTok{IO} \FunctionTok{$}\NormalTok{ raiseIO}\FunctionTok{#}\NormalTok{ (toException e)}
\end{Highlighting}
\end{Shaded}

While \VERB|\NormalTok{throwIO}| has a type that is an instance of
\VERB|\NormalTok{throw}|, their semantics differ:
\VERB|\NormalTok{throwIO}| produces \VERB|\DataTypeTok{Monad}|ic actions
while \VERB|\NormalTok{throw}| produces values. For example, both
functions in the following example will raise
\VERB|\DataTypeTok{SomethingElse}|, not \VERB|\DataTypeTok{ErrorCall}|.

\begin{Shaded}
\begin{Highlighting}[]
\KeywordTok{data} \DataTypeTok{SomethingElse} \FunctionTok{=} \DataTypeTok{SomethingElse}

\KeywordTok{instance} \DataTypeTok{Exception} \DataTypeTok{SomethingElse} \KeywordTok{where}
\end{Highlighting}
\end{Shaded}

\begin{Shaded}
\begin{Highlighting}[]
\OtherTok{throwTest ::} \DataTypeTok{IO}\NormalTok{ ()}
\NormalTok{throwTest }\FunctionTok{=} \KeywordTok{do}
  \KeywordTok{let}\NormalTok{ x }\FunctionTok{=}\NormalTok{ throw (}\DataTypeTok{ErrorCall} \StringTok{"lazy"}\NormalTok{)}
  \FunctionTok{pure}\NormalTok{ (}\DataTypeTok{Right}\NormalTok{ x)}
\NormalTok{  throwIO }\DataTypeTok{SomethingElse}

\OtherTok{throwTest' ::} \DataTypeTok{IO}\NormalTok{ ()}
\NormalTok{throwTest' }\FunctionTok{=} \KeywordTok{do}
  \KeywordTok{let}\NormalTok{ x }\FunctionTok{=}\NormalTok{ throw (}\DataTypeTok{ErrorCall} \StringTok{"lazy"}\NormalTok{)}
  \FunctionTok{pure}\NormalTok{ x}
\NormalTok{  throwIO }\DataTypeTok{SomethingElse}
\end{Highlighting}
\end{Shaded}

The \VERB|\FunctionTok{catch}| operator, however, can be reused for
handling both imprecise and precise exceptions.

\begin{remark}

\label{rem:io-two-kinds-of-exceptions}

In other words, we can say that \VERB|\DataTypeTok{IO}| has two
different exception mechanisms (precise and imprecise exceptions) with a
single exception handling mechanism (\VERB|\FunctionTok{catch}|). (And
this is pretty weird.)

\end{remark}

\hypertarget{non-exhaustive-patterns}{%
\subsection{Non-exhaustive patterns}\label{non-exhaustive-patterns}}

\label{sec:non-exhaustive}

As a side note, non-exhaustive pattern matches (and
\VERB|\KeywordTok{case}|s) \VERB|\NormalTok{throw}|
\VERB|\DataTypeTok{PatternMatchFail}| exception, while the default
\VERB|\FunctionTok{fail}| implementation calls
\VERB|\FunctionTok{error}| which \VERB|\NormalTok{throw}|s
\VERB|\DataTypeTok{ErrorCall}|.

\begin{Shaded}
\begin{Highlighting}[]
\OtherTok{\{-# LANGUAGE ScopedTypeVariables #-\}}

\KeywordTok{import} \DataTypeTok{Control.Exception}

\NormalTok{check t }\FunctionTok{=}
\NormalTok{  (evaluate t }\FunctionTok{>>} \FunctionTok{print} \StringTok{"ok"}\NormalTok{)}
  \OtherTok{`catch`}
\NormalTok{  (\textbackslash{}(}\OtherTok{e ::} \DataTypeTok{PatternMatchFail}\NormalTok{)}
    \OtherTok{->} \FunctionTok{print} \StringTok{"throws PatternMatchFail"}\NormalTok{)}
  \OtherTok{`catch`}
\NormalTok{  (\textbackslash{}(}\OtherTok{e ::} \DataTypeTok{ErrorCall}\NormalTok{)}
    \OtherTok{->} \FunctionTok{print} \StringTok{"throws ErrorCall"}\NormalTok{)}

\NormalTok{patFail }\DecValTok{1}\NormalTok{ x }\FunctionTok{=} \KeywordTok{case}\NormalTok{ x }\KeywordTok{of} \DecValTok{0} \OtherTok{->} \DecValTok{1}
\NormalTok{fail1 }\FunctionTok{=}\NormalTok{ patFail }\DecValTok{1} \DecValTok{1}
\NormalTok{fail2 }\FunctionTok{=}\NormalTok{ patFail }\DecValTok{2} \DecValTok{2}
\NormalTok{maybeDont }\FunctionTok{=} \KeywordTok{do}\NormalTok{ \{ }\DecValTok{1} \OtherTok{<-} \DataTypeTok{Just} \DecValTok{1}\NormalTok{ ; }\FunctionTok{return} \DecValTok{2}\NormalTok{ \}}
\NormalTok{maybeFail }\FunctionTok{=} \KeywordTok{do}\NormalTok{ \{ }\DecValTok{0} \OtherTok{<-} \DataTypeTok{Just} \DecValTok{1}\NormalTok{ ; }\FunctionTok{return} \DecValTok{2}\NormalTok{ \}}
\NormalTok{eithrDont }\FunctionTok{=} \KeywordTok{do}\NormalTok{ \{ }\DecValTok{1} \OtherTok{<-} \DataTypeTok{Right} \DecValTok{1}\NormalTok{ ; }\FunctionTok{return} \DecValTok{2}\NormalTok{ \}}
\NormalTok{eithrFail }\FunctionTok{=} \KeywordTok{do}\NormalTok{ \{ }\DecValTok{0} \OtherTok{<-} \DataTypeTok{Right} \DecValTok{1}\NormalTok{ ; }\FunctionTok{return} \DecValTok{2}\NormalTok{ \}}

\NormalTok{testPatterns }\FunctionTok{=} \KeywordTok{do}
\NormalTok{  check fail1     }\CommentTok{-- throws PatternMatchFail}
\NormalTok{  check fail2     }\CommentTok{-- throws PatternMatchFail}
\NormalTok{  check maybeDont }\CommentTok{-- ok}
\NormalTok{  check maybeFail }\CommentTok{-- ok (`Nothing`)}
\NormalTok{  check eithrDont }\CommentTok{-- ok}
\NormalTok{  check eithrFail }\CommentTok{-- throws ErrorCall}
\end{Highlighting}
\end{Shaded}

\hypertarget{monadic-generalizations}{%
\subsection{Monadic generalizations}\label{monadic-generalizations}}

\label{sec:monadic-generalizations}

In previous subsections we have seen a plethora of slightly different
error handling structures with different \VERB|\NormalTok{throw}| and
\VERB|\FunctionTok{catch}| operators. In this subsection we shall
describe several Hackage packages that provide structures that try to
unify this algebraic zoo.

\hypertarget{monaderror}{%
\subsubsection{MonadError}\label{monaderror}}

\label{sec:monad-error}

\VERB|\DataTypeTok{MonadError}| class
(\VERB|\DataTypeTok{Control.Monad.Error.Class}| from
\texttt{mtl}~\cite{Hackage:mtl221} package) is defined as

\begin{Shaded}
\begin{Highlighting}[]
\KeywordTok{class}\NormalTok{ (}\DataTypeTok{Monad}\NormalTok{ m) }\OtherTok{=>} \DataTypeTok{MonadError}\NormalTok{ e m}
                 \FunctionTok{|}\NormalTok{ m }\OtherTok{->}\NormalTok{ e }\KeywordTok{where}
\OtherTok{  throwError ::}\NormalTok{ e }\OtherTok{->}\NormalTok{ m a}
\OtherTok{  catchError ::}\NormalTok{ m a}
             \OtherTok{->}\NormalTok{ (e }\OtherTok{->}\NormalTok{ m a) }\OtherTok{->}\NormalTok{ m a}
\end{Highlighting}
\end{Shaded}

This structure simply generalizes \VERB|\DataTypeTok{ExceptT}|

\begin{Shaded}
\begin{Highlighting}[]
\KeywordTok{instance} \DataTypeTok{Monad}\NormalTok{ m }\OtherTok{=>} \DataTypeTok{MonadError}\NormalTok{ e (}\DataTypeTok{ExceptT}\NormalTok{ e m) }\KeywordTok{where}
\NormalTok{  throwError }\FunctionTok{=}\NormalTok{ throwE}
\NormalTok{  catchError }\FunctionTok{=}\NormalTok{ catchE}
\end{Highlighting}
\end{Shaded}

\noindent in a way that is transitive over many other
\VERB|\DataTypeTok{MonadTrans}|formers, for instance

\begin{Shaded}
\begin{Highlighting}[]
\CommentTok{-- (these require UndecidableInstances GHC extension, however)}

\KeywordTok{instance} \DataTypeTok{MonadError}\NormalTok{ e m }\OtherTok{=>} \DataTypeTok{MonadError}\NormalTok{ e (}\DataTypeTok{IdentityT}\NormalTok{ m) }\KeywordTok{where}
\NormalTok{  throwError }\FunctionTok{=}\NormalTok{ lift }\FunctionTok{.}\NormalTok{ throwError}
\NormalTok{  catchError a h }\FunctionTok{=} \DataTypeTok{IdentityT} \FunctionTok{$}\NormalTok{ catchError (runIdentityT a) (runIdentityT }\FunctionTok{.}\NormalTok{ h)}

\KeywordTok{instance} \DataTypeTok{MonadError}\NormalTok{ e m }\OtherTok{=>} \DataTypeTok{MonadError}\NormalTok{ e (}\DataTypeTok{MaybeT}\NormalTok{ m) }\KeywordTok{where}
\NormalTok{  throwError }\FunctionTok{=}\NormalTok{ lift }\FunctionTok{.}\NormalTok{ throwError}
\NormalTok{  catchError a h }\FunctionTok{=} \DataTypeTok{MaybeT} \FunctionTok{$}\NormalTok{ catchError (runMaybeT a) (runMaybeT }\FunctionTok{.}\NormalTok{ h)}
\end{Highlighting}
\end{Shaded}

\hypertarget{monadthrow-and-monadcatch}{%
\subsubsection{MonadThrow and
MonadCatch}\label{monadthrow-and-monadcatch}}

\label{sec:monad-catch}

\VERB|\DataTypeTok{MonadThrow}| and \VERB|\DataTypeTok{MonadCatch}|
classes (\VERB|\DataTypeTok{Control.Monad.Catch}| from
\texttt{exceptions}~\cite{Hackage:exceptions083}) are defined
as\footnote{Except for the fact that \VERB|\DataTypeTok{MonadCatch}|
  from \texttt{exceptions} names its operator
  \VERB|\FunctionTok{catch}|, not \VERB|\NormalTok{catchM}|, we renamed
  it for uniformity and so that it would not be confused with the
  operator from \VERB|\DataTypeTok{Control.Exception}|.}

\begin{Shaded}
\begin{Highlighting}[]
\KeywordTok{class} \DataTypeTok{Monad}\NormalTok{ m }\OtherTok{=>} \DataTypeTok{MonadThrow}\NormalTok{ m }\KeywordTok{where}
\OtherTok{  throwM ::} \DataTypeTok{Exception}\NormalTok{ e }\OtherTok{=>}\NormalTok{ e }\OtherTok{->}\NormalTok{ m a}

\KeywordTok{class} \DataTypeTok{MonadThrow}\NormalTok{ m }\OtherTok{=>} \DataTypeTok{MonadCatch}\NormalTok{ m }\KeywordTok{where}
\OtherTok{  catchM ::} \DataTypeTok{Exception}\NormalTok{ e}
         \OtherTok{=>}\NormalTok{ m a }\OtherTok{->}\NormalTok{ (e }\OtherTok{->}\NormalTok{ m a) }\OtherTok{->}\NormalTok{ m a}
\end{Highlighting}
\end{Shaded}

These two structures, too, generalizes \VERB|\DataTypeTok{ExceptT}|

\begin{Shaded}
\begin{Highlighting}[]
\KeywordTok{instance} \DataTypeTok{MonadThrow}\NormalTok{ m }\OtherTok{=>} \DataTypeTok{MonadThrow}\NormalTok{ (}\DataTypeTok{ExceptT}\NormalTok{ e m) }\KeywordTok{where}
\NormalTok{  throwM }\FunctionTok{=}\NormalTok{ lift }\FunctionTok{.}\NormalTok{ throwM}

\KeywordTok{instance} \DataTypeTok{MonadCatch}\NormalTok{ m }\OtherTok{=>} \DataTypeTok{MonadCatch}\NormalTok{ (}\DataTypeTok{ExceptT}\NormalTok{ e m) }\KeywordTok{where}
\NormalTok{  catchM x f }\FunctionTok{=} \DataTypeTok{ExceptT} \FunctionTok{$}\NormalTok{ catchM (runExceptT x) (runExceptT }\FunctionTok{.}\NormalTok{ f)}
\end{Highlighting}
\end{Shaded}

\noindent and they, too, are transitive over common
\VERB|\DataTypeTok{MonadTrans}|formers

\begin{Shaded}
\begin{Highlighting}[]
\CommentTok{-- (this time without UndecidableInstances)}

\KeywordTok{instance} \DataTypeTok{MonadThrow}\NormalTok{ m }\OtherTok{=>} \DataTypeTok{MonadThrow}\NormalTok{ (}\DataTypeTok{IdentityT}\NormalTok{ m) }\KeywordTok{where}
\NormalTok{  throwM }\FunctionTok{=}\NormalTok{ lift }\FunctionTok{.}\NormalTok{ throwM}

\KeywordTok{instance} \DataTypeTok{MonadCatch}\NormalTok{ m }\OtherTok{=>} \DataTypeTok{MonadCatch}\NormalTok{ (}\DataTypeTok{IdentityT}\NormalTok{ m) }\KeywordTok{where}
\NormalTok{  catchM x f }\FunctionTok{=} \DataTypeTok{IdentityT} \FunctionTok{$}\NormalTok{ catchM (runIdentityT x) (runIdentityT }\FunctionTok{.}\NormalTok{ f)}

\KeywordTok{instance} \DataTypeTok{MonadThrow}\NormalTok{ m }\OtherTok{=>} \DataTypeTok{MonadThrow}\NormalTok{ (}\DataTypeTok{MaybeT}\NormalTok{ m) }\KeywordTok{where}
\NormalTok{  throwM }\FunctionTok{=}\NormalTok{ lift }\FunctionTok{.}\NormalTok{ throwM}

\KeywordTok{instance} \DataTypeTok{MonadCatch}\NormalTok{ m }\OtherTok{=>} \DataTypeTok{MonadCatch}\NormalTok{ (}\DataTypeTok{MaybeT}\NormalTok{ m) }\KeywordTok{where}
\NormalTok{  catchM x f }\FunctionTok{=} \DataTypeTok{MaybeT} \FunctionTok{$}\NormalTok{ catchM (runMaybeT x) (runMaybeT }\FunctionTok{.}\NormalTok{ f)}
\end{Highlighting}
\end{Shaded}

\noindent but they constrain their argument \VERB|\NormalTok{e}| to the
\VERB|\DataTypeTok{Exception}| type class, and they also generalize the
imprecise exceptions

\begin{Shaded}
\begin{Highlighting}[]
\KeywordTok{instance} \DataTypeTok{MonadThrow} \DataTypeTok{IO} \KeywordTok{where}
\NormalTok{  throwM }\FunctionTok{=}\NormalTok{ throw}

\KeywordTok{instance} \DataTypeTok{MonadCatch} \DataTypeTok{IO} \KeywordTok{where}
\NormalTok{  catchM }\FunctionTok{=} \FunctionTok{catch}
\end{Highlighting}
\end{Shaded}

The latter fact complicates their use somewhat since one can not be sure
about the dynamic-dispatch part of the semantics without actually
looking at the definitions for a particular instance.

\hypertarget{not-a-tutorial-side-b}{%
\section{Not a Tutorial: Side B}\label{not-a-tutorial-side-b}}

\label{sec:tutorial:non-basic}

This section, logically, is a continuation of \cref{sec:tutorial:basic}.
However, in contrast to that section this section discusses non-basic
structures that are of particular importance to the rest of the article.
While this section does not introduce any non-trivial novel ideas, some
perspectives on well-known ideas seem to be novel.

\hypertarget{continuations}{%
\subsection{Continuations}\label{continuations}}

\label{sec:continuations}

When speaking of "continuations" people usually mean one or more of the
three related aspects explained in this subsection.

\hypertarget{continuation-passing-style}{%
\subsubsection{Continuation-Passing
Style}\label{continuation-passing-style}}

Any (sub-)program can be rewritten into Continuation-Passing Style
(CPS)~\cite{Reynolds:1993:DC,appel-92} by adding a number of additional
\emph{continuation} arguments to every function and tail-calling into
those arguments with the results-to-be at every return point instead of
just returning said results.

For instance, the following pseudo-Haskell program

\begin{Shaded}
\begin{Highlighting}[]
\NormalTok{foo }\FunctionTok{=}
  \KeywordTok{if}\NormalTok{ something}
     \KeywordTok{then} \DataTypeTok{Result1}\NormalTok{ result1}
     \KeywordTok{else} \DataTypeTok{Result2}\NormalTok{ result2}

\NormalTok{bar }\FunctionTok{=} \KeywordTok{case}\NormalTok{ foo }\KeywordTok{of}
  \DataTypeTok{Result1}\NormalTok{ a }\OtherTok{->}\NormalTok{ bar1 a}
  \DataTypeTok{Result2}\NormalTok{ b }\OtherTok{->}\NormalTok{ bar2 b}
\end{Highlighting}
\end{Shaded}

\noindent can be transformed into (here we CPS-ignore
\VERB|\NormalTok{something}| and the \VERB|\KeywordTok{if}| for
illustrative purposes)

\newpage

\begin{Shaded}
\begin{Highlighting}[]
\NormalTok{fooCPS cont1 cont2 }\FunctionTok{=}
  \KeywordTok{if}\NormalTok{ something}
    \KeywordTok{then}\NormalTok{ cont1 result1}
    \KeywordTok{else}\NormalTok{ cont2 result2}

\NormalTok{barCPS }\FunctionTok{=}\NormalTok{ fooCPS bar1 bar2}
\end{Highlighting}
\end{Shaded}

In conventional modern low-level imperative terms this transformation
requires all functions to receive their return addresses as explicit
parameters instead of \texttt{pop}ing them from the bottom of their
stack frame.

The latter, of course, means that we can treat "\emph{normal}" programs
(in which all functions have a single return address) as a degenerate
case of programs written in "\emph{implicit-CPS}" (in fact,
\VERB|\DataTypeTok{Cont}| \VERB|\DataTypeTok{Monad}| of \cref{sec:cont}
is exactly such an "\emph{implicit-CPS}") --- a syntactic variant of CPS
in which

\begin{itemize}
\item
  every function has an implicit argument that specifies a default
  return address (which is set to the next instruction following a
  corresponding function call by default)
\item
  that can be reached from the body of the function by tail-calling a
  special symbol that \texttt{jmp}s to the implicitly given address.
\end{itemize}

Finally, one can even imagine a computer with a "\emph{CPS-ISA}" (i.e.
an ISA where each instruction explicitly specifies its own return
address) in which case all programs for such a computer would have to be
translated into an explicit CPS form to be executed. In fact, drum
memory-based computers like IBM 650 had exactly such an ISA. From the
point of view of an IBM 650 programmer modern conventional CPUs simply
convert their non-CPS OPcodes into their CPS forms on the fly, thus
applying CPS-transform to any given program on the fly.

Returning to the pseudo-Haskell listing above, note that programs
written in CPS

\begin{itemize}
\tightlist
\item
  introduce a linear order on their computations, hence they are not
  particularly good for parallel execution,
\item
  consume somewhat more memory in comparison to their "\emph{normal}"
  representations (as they have to handle more explicit addresses),
\item
  can have poorer performance on modern conventional CPUs (since said
  CPUs split their branch predictors into "jump" and "call" units and
  the latter unit rests completely unused by CPS programs),
\item
  are harder to understand.
\end{itemize}

However, the advantage of the CPS form is that it allows elimination of
duplicate computations. For instance, in the example above
\VERB|\NormalTok{foo}| produces different results depending on the value
of \VERB|\NormalTok{something}| and \VERB|\NormalTok{bar}| has to
duplicate that choice (but not the computation of
\VERB|\NormalTok{something}|) again by switching
\VERB|\KeywordTok{case}|s on the result of \VERB|\NormalTok{foo}|.
Meanwhile, \VERB|\NormalTok{barCPS}| is free from such an inefficiency.
Applying this transformation recursively to a whole (sub-)program allows
one to transform the (sub-)program into a series of tail calls whilst
replacing all constructors and eliminators in the (sub-)program with
tail calls to newly introduced continuation arguments and
\VERB|\KeywordTok{case}| bodies respectively.

The logical mechanic behind this transformation is a technique we call
\emph{generalized Kolmogorov's translation} (since it is a trivial
extension of Kolmogorov's
translation~\cite{Kolmogorov:1925:OPT:reprint}) of types of functions'
results. That is, double negation followed by rewriting by well-known
isomorphisms until formula contains only arrows, bottoms and variables
followed by generalizing bottoms by a bound variable.

For instance, the result of a function of type

\[i \to j \to b\]

\noindent is \(b\), which can be doubly negated as

\newpage \[\lnot \lnot b\] \[(b \to \bot) \to \bot\]

\noindent and generalized to either of

\[\forall c . (b \to c) \to c\] \[\lambda c . (b \to c) \to c\]

\noindent which allows us to generalize the whole function to either of

\[former = \forall c . i \to j \to (b \to c) \to c\]
\[latter = \lambda c . i \to j \to (b \to c) \to c\]

\noindent depending on the desired properties:

\begin{itemize}
\item
  the former term requires a rank-2 type system but it does not add any
  new type lambdas or free type variables, thus keeping the
  transformation closed,
\item
  the latter term does not need rank-2 types, but it requires tracking
  of these new type variables,
\item
  the latter term also retains full control over \(c\) variable, (for
  instance, it can produce the former term in rank-2 type system on
  demand with \(\forall c . latter~c\)).
\end{itemize}

Similarly, \VERB|\DataTypeTok{Either}\NormalTok{ a b}| may be seen as
logical \(a \lor b\) which can be rewritten as

\[\lnot \lnot (a \lor b)\] \[\lnot (\lnot a \land \lnot b)\]
\[(a \to \bot \land b \to \bot) \to \bot\]
\[(a \to \bot) \to (b \to \bot) \to \bot\]

\noindent and a pair of \VERB|\NormalTok{(a, b)}| is logical
\(a \land b\) and can be rewritten as

\[\lnot \lnot (a \land b)\] \[\lnot (a \land b) \to \bot\]
\[(a \land b \to \bot) \to \bot\] \[(a \to b \to \bot) \to \bot\]

Hence, \(i \to j \to (a \lor b)\) can be rewritten into either of

\[\forall c . i \to j \to (a \to c) \to (b \to c) \to c\]
\[\lambda c . i \to j \to (a \to c) \to (b \to c) \to c\]

\noindent and \(i \to j \to (a \land b)\) into either of

\[\forall c . i \to j \to (a \to b \to c) \to c\]
\[\lambda c . i \to j \to (a \to b \to c) \to c\]

\hypertarget{scott-encoding}{%
\subsubsection{Scott-encoding}\label{scott-encoding}}

\label{sec:scott-encoding}

A technique of applying generalized Kolmogorov's translation to data
types and their constructors and eliminators instead of normal functions
in a (sub-)program is called Scott-encoding (apparently, Dana Scott did
not publish, to our best knowledge the first mention in print is
\cite[p.~219]{Curry:1972:CL2} and first generic description of the
technique for arbitrary data types is
\cite{Steensgaard-Madsen:1989:TRO}).

As before, \VERB|\DataTypeTok{Either}| can be replaced with either of

\[\forall c . (a \to c) \to (b \to c) \to c\]
\[\lambda c . (a \to c) \to (b \to c) \to c\]

\noindent which can be encoded in Haskell as either of

\begin{Shaded}
\begin{Highlighting}[]
\KeywordTok{newtype} \DataTypeTok{EitherS}\NormalTok{ a b }\FunctionTok{=} \DataTypeTok{EitherS}
\NormalTok{  \{ runEitherS}
\OtherTok{      ::} \KeywordTok{forall}\NormalTok{ c}
       \FunctionTok{.}\NormalTok{ (a }\OtherTok{->}\NormalTok{ c) }\OtherTok{->}\NormalTok{ (b }\OtherTok{->}\NormalTok{ c) }\OtherTok{->}\NormalTok{ c \}}

\OtherTok{left ::}\NormalTok{ a }\OtherTok{->} \DataTypeTok{EitherS}\NormalTok{ a b}
\NormalTok{left a }\FunctionTok{=} \DataTypeTok{EitherS}\NormalTok{ (\textbackslash{}ac bc }\OtherTok{->}\NormalTok{ ac a)}

\OtherTok{right ::}\NormalTok{ b }\OtherTok{->} \DataTypeTok{EitherS}\NormalTok{ a b}
\NormalTok{right b }\FunctionTok{=} \DataTypeTok{EitherS}\NormalTok{ (\textbackslash{}ac bc }\OtherTok{->}\NormalTok{ bc b)}

\KeywordTok{newtype} \DataTypeTok{EitherS'}\NormalTok{ c a b }\FunctionTok{=} \DataTypeTok{EitherS'}
\NormalTok{  \{ runEitherS'}
\OtherTok{      ::}\NormalTok{ (a }\OtherTok{->}\NormalTok{ c) }\OtherTok{->}\NormalTok{ (b }\OtherTok{->}\NormalTok{ c) }\OtherTok{->}\NormalTok{ c \}}

\OtherTok{left' ::}\NormalTok{ a }\OtherTok{->} \DataTypeTok{EitherS'}\NormalTok{ c a b}
\NormalTok{left' a }\FunctionTok{=} \DataTypeTok{EitherS'}\NormalTok{ (\textbackslash{}ac bc }\OtherTok{->}\NormalTok{ ac a)}

\OtherTok{right' ::}\NormalTok{ b }\OtherTok{->} \DataTypeTok{EitherS'}\NormalTok{ c a b}
\NormalTok{right' b }\FunctionTok{=} \DataTypeTok{EitherS'}\NormalTok{ (\textbackslash{}ac bc }\OtherTok{->}\NormalTok{ bc b)}
\end{Highlighting}
\end{Shaded}

\noindent with \VERB|\NormalTok{runEitherS}|
(\VERB|\NormalTok{runEitherS'}|) taking the role of an eliminator
(\VERB|\KeywordTok{case}| operator) and \VERB|\NormalTok{left}| and
\VERB|\NormalTok{right}| (\VERB|\NormalTok{left'}| and
\VERB|\NormalTok{right'}|) taking the roles of \VERB|\DataTypeTok{Left}|
and \VERB|\DataTypeTok{Right}| constructors respectively.

Similarly, \VERB|\NormalTok{(a, b)}| can then be generalized to either
of

\[\forall c . (a \to b \to c) \to c\]
\[\lambda c . (a \to b \to c) \to c\]

\noindent and encoded in Haskell as either of

\begin{Shaded}
\begin{Highlighting}[]
\KeywordTok{newtype} \DataTypeTok{PairS}\NormalTok{ a b }\FunctionTok{=} \DataTypeTok{PairS}
\NormalTok{  \{ runPairS}
\OtherTok{      ::} \KeywordTok{forall}\NormalTok{ c}
       \FunctionTok{.}\NormalTok{ (a }\OtherTok{->}\NormalTok{ b }\OtherTok{->}\NormalTok{ c) }\OtherTok{->}\NormalTok{ c \}}

\OtherTok{pair ::}\NormalTok{ a }\OtherTok{->}\NormalTok{ b }\OtherTok{->} \DataTypeTok{PairS}\NormalTok{ a b}
\NormalTok{pair a b }\FunctionTok{=} \DataTypeTok{PairS}\NormalTok{ (\textbackslash{}f }\OtherTok{->}\NormalTok{ f a b)}

\KeywordTok{newtype} \DataTypeTok{PairS'}\NormalTok{ c a b }\FunctionTok{=} \DataTypeTok{PairS'}
\NormalTok{  \{ runPairS'}
\OtherTok{      ::}\NormalTok{ (a }\OtherTok{->}\NormalTok{ b }\OtherTok{->}\NormalTok{ c) }\OtherTok{->}\NormalTok{ c \}}

\OtherTok{pair' ::}\NormalTok{ a }\OtherTok{->}\NormalTok{ b }\OtherTok{->} \DataTypeTok{PairS'}\NormalTok{ c a b}
\NormalTok{pair' a b }\FunctionTok{=} \DataTypeTok{PairS'}\NormalTok{ (\textbackslash{}f }\OtherTok{->}\NormalTok{ f a b)}
\end{Highlighting}
\end{Shaded}

Substituting all \VERB|\DataTypeTok{Left}|s with
\VERB|\NormalTok{left}|, \VERB|\DataTypeTok{Right}|s with
\VERB|\NormalTok{right}|, \VERB|\KeywordTok{case}|s on
\VERB|\DataTypeTok{Either}|s with \VERB|\NormalTok{runEitherS}|, pair
constructions with \VERB|\NormalTok{pair}|, and
\VERB|\KeywordTok{case}|s on pairs with \VERB|\NormalTok{runPairS}| (and
similarly for primed versions) does not change computational properties
of the transformed program in the sense that Scott-transformation of the
original program's normal form coincides with the normal form of the
Scott-transformed program.

Replacing a single data type in a program with its Scott-encoding can be
viewed as a kind of selective CPS-transform on those subterms of the
program that use the data type. The type of transformed functions
changes the same way in both transformations, but Scott-encoding groups
all continuation arguments, hides them behind a type alias and
introduces a bunch of redundant beta reductions in constructors and
eliminators.

The upside of CPS-transforming with Scott-encoding is that it supports
partial applications, requires absolutely no thought to perform and no
substantial changes to the bodies of the functions that are being
transformed. It is also very useful for designing new languages and
emulating data types in languages that do not support them\footnote{For
  example, most instances of the \emph{visitor} object-oriented (OOP)
  design pattern that are not simply emulating
  \VERB|\DataTypeTok{Functor}| instances usually emulate pattern
  matching with Scott-encoding.} as it allows to use data types when
none are supported by the core language.

The most immediate downside of this transformation is very poor
performance on modern conventional CPUs. For instance, pattern matching
on \VERB|\DataTypeTok{Either}| produces a simple short conditional
\texttt{jmp} while for \VERB|\NormalTok{runEitherS}| the compiler, in
general, cannot be sure about value of the arguments (it can be anything
of the required type, not only \VERB|\NormalTok{left}| or
\VERB|\NormalTok{right}|) and has to produce an indirect \texttt{jmp}
(or \texttt{call} if it is not a tail call) and both
\VERB|\NormalTok{left}| and \VERB|\NormalTok{right}| require another
indirect \texttt{jmp}. This wastes address cache of CPU's branch
predictor and confuses it\footnote{Note that this does not happen for
  the full CPS-transform of the previous subsection since that
  translation does no \texttt{call}s.} when instruction pointer jumps
out of the stack frame.

For some classes of programs, however, it can increase performance
significantly. For instance, in a
"\emph{\VERB|\KeywordTok{case}|-tower}" like

\begin{Shaded}
\begin{Highlighting}[]
\NormalTok{doSomethingOn s }\FunctionTok{=} \KeywordTok{case}\NormalTok{ internally s }\KeywordTok{of}
  \DataTypeTok{Right}\NormalTok{ a }\OtherTok{->}\NormalTok{ returnResult a}
  \DataTypeTok{Left}\NormalTok{ b }\OtherTok{->}\NormalTok{ handeError b}

\NormalTok{internally s }\FunctionTok{=}
  \KeywordTok{case}\NormalTok{ evenMoreInternally s }\KeywordTok{of}
    \DataTypeTok{Right}\NormalTok{ (a,s) }\OtherTok{->}\NormalTok{ doSomethingElse a s}
    \DataTypeTok{Left}\NormalTok{ b }\OtherTok{->} \DataTypeTok{Left}\NormalTok{ b}

\NormalTok{doSomethingElse a s }\FunctionTok{=}
  \KeywordTok{case}\NormalTok{ evenMoreInternally s }\KeywordTok{of}
    \DataTypeTok{Right}\NormalTok{ (a,s) }\OtherTok{->} \DataTypeTok{Right}\NormalTok{ a}
    \DataTypeTok{Left}\NormalTok{ b }\OtherTok{->} \DataTypeTok{Left}\NormalTok{ b}
\end{Highlighting}
\end{Shaded}

\noindent (which is commonly produced by parser combinators) performing
this selective CPS-transform followed by inlining and partial evaluation
of the affected functions will replace all construction sites of
\VERB|\DataTypeTok{Left}|s with direct calls to
\VERB|\NormalTok{handeError}|, and \VERB|\DataTypeTok{Right}|s in
\VERB|\NormalTok{doSomethingElse}| (and, possibly, the ones residing in
\VERB|\NormalTok{evenMoreInternally}|) with
\VERB|\NormalTok{returnResult}|.

In other words, rewriting this type of code using Scott-encoded data
types is a way to apply deforestation~\cite{Wadler:1990:DTP} to it, but
semi-manually as opposed to automatically, and with high degree of
control. This fact gets used a lot in Hackage libraries, where, for
example, most parser combinators (\cref{sec:parser-combinators}) use
Scott-encoded forms internally.

\hypertarget{cont}{%
\subsubsection{Cont}\label{cont}}

\label{sec:cont}

One of the roundabout ways to express pure values in Haskell is to wrap
them with the \VERB|\DataTypeTok{Identity}| \VERB|\DataTypeTok{Functor}|
(\cref{sec:identity}) for which
\VERB|\DataTypeTok{Identity}\NormalTok{ a}|, logically, is just a pure
type variable \(a\). Applying generalized Kolmogorov's translation to
this variable gives either of

\[\forall c . (a \to c) \to c\] \[\lambda c . (a \to c) \to c\]

In Haskell the latter type is called \VERB|\DataTypeTok{Cont}|. It is
defined in \VERB|\DataTypeTok{Control.Monad.Cont}| of
\texttt{mtl}~\cite{Hackage:mtl221} as

\begin{Shaded}
\begin{Highlighting}[]
\KeywordTok{newtype} \DataTypeTok{Cont}\NormalTok{ r a }\FunctionTok{=} \DataTypeTok{Cont}
\NormalTok{  \{}\OtherTok{ runCont ::}\NormalTok{ (a }\OtherTok{->}\NormalTok{ r) }\OtherTok{->}\NormalTok{ r \}}
\end{Highlighting}
\end{Shaded}

\noindent with the following \VERB|\DataTypeTok{Monad}| instance

\begin{Shaded}
\begin{Highlighting}[]
\KeywordTok{instance} \DataTypeTok{Pointed}\NormalTok{ (}\DataTypeTok{Cont}\NormalTok{ r) }\KeywordTok{where}
  \FunctionTok{pure}\NormalTok{ a }\FunctionTok{=} \DataTypeTok{Cont} \FunctionTok{$}\NormalTok{ \textbackslash{}c }\OtherTok{->}\NormalTok{ c a}

\KeywordTok{instance} \DataTypeTok{Monad}\NormalTok{ (}\DataTypeTok{Cont}\NormalTok{ r) }\KeywordTok{where}
\NormalTok{  m }\FunctionTok{>>=}\NormalTok{ f  }\FunctionTok{=} \DataTypeTok{Cont} \FunctionTok{$}\NormalTok{ \textbackslash{}c }\OtherTok{->}\NormalTok{ runCont m}
                  \FunctionTok{$}\NormalTok{ \textbackslash{}a }\OtherTok{->}\NormalTok{ runCont (f a) c}
\end{Highlighting}
\end{Shaded}

\VERB|\DataTypeTok{Cont}| has a transformer version defined in
\VERB|\DataTypeTok{Control.Monad.Trans.Cont}| module of
\texttt{transformers}~\cite{Hackage:transformers0520} package as follows

\begin{Shaded}
\begin{Highlighting}[]
\KeywordTok{newtype} \DataTypeTok{ContT}\NormalTok{ r m a }\FunctionTok{=} \DataTypeTok{ContT}\NormalTok{ \{}\OtherTok{ runContT ::}\NormalTok{ (a }\OtherTok{->}\NormalTok{ m r) }\OtherTok{->}\NormalTok{ m r \}}

\KeywordTok{instance} \DataTypeTok{MonadTrans}\NormalTok{ (}\DataTypeTok{ContT'}\NormalTok{ r) }\KeywordTok{where}
\NormalTok{  lift m }\FunctionTok{=} \DataTypeTok{ContT}\NormalTok{ (m }\FunctionTok{>>=}\NormalTok{)}
\end{Highlighting}
\end{Shaded}

Interestingly, however, unlike \VERB|\DataTypeTok{Identity}| and
\VERB|\DataTypeTok{IdentityT}| which have different
\VERB|\DataTypeTok{Monad}| instances (see
\cref{sec:identity-monadtrans}), \VERB|\DataTypeTok{Cont}| and
\VERB|\DataTypeTok{ContT}| have identical ones (equivalent to the one
given above). Of particular note is the fact that the definition of
\VERB|\NormalTok{(}\FunctionTok{>>=}\NormalTok{)}| for
\VERB|\DataTypeTok{ContT}| does not refer to the
\VERB|\DataTypeTok{Monad}| operators of its argument
\VERB|\NormalTok{m}|. This means that in cases when we do not need the
\VERB|\DataTypeTok{MonadTrans}| instance (for which we have to have a
\VERB|\KeywordTok{newtype}| wrapper) we can redefine
\VERB|\DataTypeTok{ContT}| as simply

\begin{Shaded}
\begin{Highlighting}[]
\KeywordTok{type} \DataTypeTok{ContT}\NormalTok{ r m a }\FunctionTok{=} \DataTypeTok{Cont}\NormalTok{ (m r) a}
\end{Highlighting}
\end{Shaded}

The latter fact means that \VERB|\DataTypeTok{ContT}|, unlike other
\VERB|\DataTypeTok{MonadTrans}|formers we saw before, is not a
"\VERB|\DataTypeTok{Monad}| transformer" as it is not a functor on
category of monads (it is always a \VERB|\DataTypeTok{Monad}|
irrespective of the argument \VERB|\NormalTok{m}|). This property can be
explained by the fact that, as we noted at the top of this section,
\VERB|\DataTypeTok{Cont}| \VERB|\DataTypeTok{Monad}| is a kind of
\emph{"implicit-CPS"} form of computations. Since all it does is chain
return addresses it does not care about types of computations those
addresses point to.

\hypertarget{delimited-callcc}{%
\subsubsection{Delimited callCC}\label{delimited-callcc}}

\label{sec:callcc}

Peirce's law states that

\[((a \to b) \to a) \to a\]

\noindent by applying generalized Kolmogorov's translation we get

\[\lnot \lnot (((a \to b) \to a) \to a)\]
\[\lnot (\lnot a \to \lnot ((a \to b) \to a))\]
\[\lnot \lnot ((a \to b) \to a) \to \lnot \lnot a\]
\[(\lnot \lnot (a \to b) \to \lnot \lnot a) \to \lnot \lnot a\]
\[((\lnot \lnot a \to \lnot \lnot b) \to \lnot \lnot a) \to \lnot \lnot a\]

\noindent which can be encoded in Haskell as (note that this time we use
\(\forall\) variant of the translation)

\begin{Shaded}
\begin{Highlighting}[]
\OtherTok{peirceCC ::}\NormalTok{ ((}\DataTypeTok{Cont}\NormalTok{ r a }\OtherTok{->} \DataTypeTok{Cont}\NormalTok{ r b) }\OtherTok{->} \DataTypeTok{Cont}\NormalTok{ r a)}
         \OtherTok{->} \DataTypeTok{Cont}\NormalTok{ r a}
\NormalTok{peirceCC f }\FunctionTok{=} \DataTypeTok{Cont} \FunctionTok{$}\NormalTok{ \textbackslash{}c }\OtherTok{->}
\NormalTok{  runCont (f (\textbackslash{}ac }\OtherTok{->} \DataTypeTok{Cont} \FunctionTok{$}\NormalTok{ \textbackslash{}_ }\OtherTok{->}\NormalTok{ runCont ac c)) c}
\end{Highlighting}
\end{Shaded}

This operator takes a function \VERB|\NormalTok{f}|, applies some
magical subterm to it and then gives it its own return address. That is,
for a function \VERB|\NormalTok{f}| that ignores its argument
\VERB|\NormalTok{peirceCC}| is completely transparent. The magical
argument \VERB|\NormalTok{peirceCC}| applies to \VERB|\NormalTok{f}| is
itself a function that takes a computation producing value of the same
type \VERB|\NormalTok{f}| returns as a result. The subterm then computes
the value of the argument but ignores its own return address and
continues to the return address given to \VERB|\NormalTok{peirceCC}|
instead. In other words, \VERB|\NormalTok{peirceCC}| applies
\VERB|\NormalTok{f}| with an \emph{escape continuation} which works
exactly like a \VERB|\ControlFlowTok{return}| statement of conventional
imperative languages (as opposed to \VERB|\DataTypeTok{Monad}|'s
\VERB|\FunctionTok{pure}| which should not be called
"\VERB|\FunctionTok{return}|", see \cref{sec:monad}).

Note that \VERB|\NormalTok{ac}| argument to the magical subterm is
pretty boring: it is a computation that gets computed immediately.
Hence, unless we require every subterm of our program to be written in
\emph{implicit-CPS} form we can simplify \VERB|\NormalTok{peirceCC}| a
bit as follows

\begin{Shaded}
\begin{Highlighting}[]
\OtherTok{callCC ::}\NormalTok{ ((a }\OtherTok{->} \DataTypeTok{Cont}\NormalTok{ r b) }\OtherTok{->} \DataTypeTok{Cont}\NormalTok{ r a) }\OtherTok{->} \DataTypeTok{Cont}\NormalTok{ r a}
\NormalTok{callCC f }\FunctionTok{=} \DataTypeTok{Cont} \FunctionTok{$}\NormalTok{ \textbackslash{}c }\OtherTok{->}
\NormalTok{  runCont (f (\textbackslash{}a }\OtherTok{->} \DataTypeTok{Cont} \FunctionTok{$}\NormalTok{ \textbackslash{}_ }\OtherTok{->}\NormalTok{ c a)) c}
\end{Highlighting}
\end{Shaded}

This operator bears a name of "delimited \VERB|\KeywordTok{call/cc}|
(\VERB|\NormalTok{callCC}|)"~\cite{Asai:2011:IPS} and the escape
continuation it supplies to \VERB|\NormalTok{f}| not only works but also
looks exactly like an imperative \VERB|\ControlFlowTok{return}| (in that
it takes a pure value instead of a computation producing it).

\hypertarget{schemes-callcc-and-mls-callcc}{%
\subsubsection{Scheme's call/cc and ML's
callcc}\label{schemes-callcc-and-mls-callcc}}

Note that delimited \VERB|\NormalTok{callCC}| is semantically different
from similarly named operators of SML~\cite{SML:Cont} and
Scheme~\cite{Sperber:2010:RnRS}. SML defines its operator as

\begin{Shaded}
\begin{Highlighting}[]
\KeywordTok{type}\NormalTok{ 'a cont}
\KeywordTok{val}\NormalTok{ callcc : ('a cont -> 'a) -> 'a}
\end{Highlighting}
\end{Shaded}

\noindent where \VERB|\NormalTok{'a cont}| type is the type of the
\emph{current global continuation} which is the computation till the end
of the whole program, this type is a kind of technical alias for what,
logically, should be \(a \to b\), i.e. \VERB|\NormalTok{callcc}|'s type,
logically, is non-Kolmogorov-translated Peirce's law.

The difference is that by applying Kolmogorov's translation to Peirce's
law \VERB|\NormalTok{callCC}| gains intuitionistic witnesses (and,
hence, purely functional implementations) and becomes \emph{delimited}
by the current \VERB|\DataTypeTok{Cont}| context instead of the whole
program. Meanwhile, implementations of non-delimited
\VERB|\NormalTok{callcc}| and \VERB|\KeywordTok{call/cc}| require
special support from the compiler/interpreter and
Kiselyov~\cite{Kiselyov:2012:AAC} eloquently advocates that they simply
should not exist as they are \emph{less} useful than their delimited
versions and their implementations introduce nontrivial trade-offs to
the languages in question.

\hypertarget{monadic-parser-combinators}{%
\subsection{Monadic Parser
Combinators}\label{monadic-parser-combinators}}

\label{sec:parser-combinators}

\VERB|\DataTypeTok{Monad}|ic parser combinators are not by themselves an
error handling mechanism, but they have to handle failed parsing
attempts and such computations can be seen as a kind of error handling.

Parser combinators can possess a wide variety of semantics and
implementations, to mention just a few possible dimensions of the space:

\begin{itemize}
\tightlist
\item
  they can either automatically backtrack on errors or keep the state as
  is,
\item
  they can distinguish not only successful and failed parsing attempts
  but also attempts that consumed none of the input and those that
  consumed at least one element of the input~\cite{Leijen:2001:PDS},
\item
  they can support an impure state (e.g., make it a
  \VERB|\DataTypeTok{Monad}|),
\item
  track position in the input stream,
\item
  allow programmer-provided types in errors,
\item
  provide \VERB|\DataTypeTok{MonadTrans}|former versions,
\item
  encode their internals with Scott-encoding (\cref{sec:scott-encoding})
  for efficiency.
\end{itemize}

Discussing most of those features and their combinations is beyond the
scope of this article. In the following subsections we shall only
mention "backtrack vs. not" problem, in \cref{sec:instances:eio} we
shall also apply Scott-encoding to an almost identical structure.
Detailed implementations of other features can be studied by following
respective references.

The most popular parser combinator libraries for Haskell are
Parsec~\cite{Hackage:parsec3111},
Attoparsec~\cite{Hackage:attoparsec01310}, and
Megaparsec~\cite{Hackage:megaparsec630}.

\hypertarget{simple-stateful-parser-combinator}{%
\subsubsection{Simple stateful parser
combinator}\label{simple-stateful-parser-combinator}}

\label{sec:parser-combinators:without-access}

The simplest \VERB|\DataTypeTok{Monad}|ic parser combinator is just a
composition of \VERB|\DataTypeTok{StateT}|
(\cref{sec:reader-state-monadtrans}) and \VERB|\DataTypeTok{ExceptT}|
(\cref{sec:either-monadtrans}) \VERB|\DataTypeTok{MonadTrans}|formers
with inner \VERB|\DataTypeTok{Identity}| (\cref{sec:identity})

\begin{Shaded}
\begin{Highlighting}[]
\KeywordTok{type} \DataTypeTok{SParser}\NormalTok{ s e }\FunctionTok{=} \DataTypeTok{StateT}\NormalTok{ s (}\DataTypeTok{ExceptT}\NormalTok{ e }\DataTypeTok{Identity}\NormalTok{)}
\end{Highlighting}
\end{Shaded}

\noindent which can be \(\beta\)-reduced into

\begin{Shaded}
\begin{Highlighting}[]
\KeywordTok{newtype} \DataTypeTok{SParser}\NormalTok{ s e a }\FunctionTok{=} \DataTypeTok{SParser}
\NormalTok{  \{}\OtherTok{ runSParser ::}\NormalTok{ s }\OtherTok{->} \DataTypeTok{Either}\NormalTok{ e (a, s) \}}
\end{Highlighting}
\end{Shaded}

\noindent with the following \VERB|\DataTypeTok{Monad}| instance

\begin{Shaded}
\begin{Highlighting}[]
\KeywordTok{instance} \DataTypeTok{Pointed}\NormalTok{ (}\DataTypeTok{SParser}\NormalTok{ s e) }\KeywordTok{where}
  \FunctionTok{pure}\NormalTok{ a }\FunctionTok{=} \DataTypeTok{SParser} \FunctionTok{$}\NormalTok{ \textbackslash{}s }\OtherTok{->} \DataTypeTok{Right}\NormalTok{ (a, s)}

\KeywordTok{instance} \DataTypeTok{Monad}\NormalTok{ (}\DataTypeTok{SParser}\NormalTok{ s e) }\KeywordTok{where}
\NormalTok{  p }\FunctionTok{>>=}\NormalTok{ f }\FunctionTok{=} \DataTypeTok{SParser} \FunctionTok{$}\NormalTok{ \textbackslash{}s }\OtherTok{->}
    \KeywordTok{case}\NormalTok{ runSParser p s }\KeywordTok{of}
      \DataTypeTok{Left}\NormalTok{ x }\OtherTok{->} \DataTypeTok{Left}\NormalTok{ x}
      \DataTypeTok{Right}\NormalTok{ (a, s') }\OtherTok{->}\NormalTok{ runSParser (f a) s'}
\end{Highlighting}
\end{Shaded}

\hypertarget{with-full-access-to-the-state}{%
\subsubsection{\ldots{} with full access to the
state}\label{with-full-access-to-the-state}}

\label{sec:parser-combinators:with-access}

While the definition above is, in fact, exactly the definition used in
\texttt{ponder}~\cite{Hackage:ponder001} parser combinator library, it
provides no way to access the state of the parser on error, which makes
it very inconvenient in practice. However, a simple modification of the
type that moves \VERB|\DataTypeTok{Either}| into the tuple

\begin{Shaded}
\begin{Highlighting}[]
\KeywordTok{newtype} \DataTypeTok{Parser}\NormalTok{ s e a }\FunctionTok{=} \DataTypeTok{Parser}
\NormalTok{  \{}\OtherTok{ runParser ::}\NormalTok{ s }\OtherTok{->}\NormalTok{ (}\DataTypeTok{Either}\NormalTok{ e a, s) \}}
\end{Highlighting}
\end{Shaded}

\noindent which, of course, in isomorphic to\footnote{The corresponding
  \VERB|\DataTypeTok{MonadTrans}|former stack is better left unspoken.}

\begin{Shaded}
\begin{Highlighting}[]
\KeywordTok{newtype} \DataTypeTok{Parser}\NormalTok{ s e a }\FunctionTok{=} \DataTypeTok{Parser}
\NormalTok{  \{}\OtherTok{ runParser ::}\NormalTok{ s }\OtherTok{->} \DataTypeTok{Either}\NormalTok{ (e, s) (a, s) \}}
\end{Highlighting}
\end{Shaded}

\noindent solves this problem of access to state while keeping the
definition of \VERB|\DataTypeTok{Monad}| identical to the above.

\newpage

\begin{theorem}

\VERB|\DataTypeTok{Parser}| satisfies \VERB|\DataTypeTok{Monad}| laws.

\end{theorem}

\begin{proof}

By case analysis. Also see the next proof.

\end{proof}

The binary \emph{choice} operator can be implemented by one of the two
possible instances of \VERB|\DataTypeTok{Alternative}|. The first one
rolls-back the state on error

\begin{Shaded}
\begin{Highlighting}[]
\KeywordTok{instance} \DataTypeTok{Monoid}\NormalTok{ e }\OtherTok{=>} \DataTypeTok{Alternative}\NormalTok{ (}\DataTypeTok{Parser}\NormalTok{ s e) }\KeywordTok{where}
\NormalTok{  empty }\FunctionTok{=} \DataTypeTok{Parser} \FunctionTok{$}\NormalTok{ \textbackslash{}s }\OtherTok{->} \DataTypeTok{Left}\NormalTok{ (}\FunctionTok{mempty}\NormalTok{, s)}
\NormalTok{  f }\FunctionTok{<|>}\NormalTok{ g }\FunctionTok{=} \DataTypeTok{Parser} \FunctionTok{$}\NormalTok{ \textbackslash{}s }\OtherTok{->} \KeywordTok{case}\NormalTok{ runParser f s }\KeywordTok{of}
    \DataTypeTok{Right}\NormalTok{ x }\OtherTok{->} \DataTypeTok{Right}\NormalTok{ x}
    \DataTypeTok{Left}\NormalTok{  (e, _) }\OtherTok{->} \KeywordTok{case}\NormalTok{ runParser g s }\KeywordTok{of}
      \DataTypeTok{Right}\NormalTok{ x }\OtherTok{->} \DataTypeTok{Right}\NormalTok{ x}
      \DataTypeTok{Left}\NormalTok{  (f, _) }\OtherTok{->} \DataTypeTok{Left}\NormalTok{ (e }\OtherTok{`mappend`}\NormalTok{ f, s)}
\end{Highlighting}
\end{Shaded}

\noindent while the second does not

\begin{Shaded}
\begin{Highlighting}[]
\KeywordTok{instance} \DataTypeTok{Monoid}\NormalTok{ e }\OtherTok{=>} \DataTypeTok{Alternative}\NormalTok{ (}\DataTypeTok{Parser}\NormalTok{ s e) }\KeywordTok{where}
\NormalTok{  empty }\FunctionTok{=} \DataTypeTok{Parser} \FunctionTok{$}\NormalTok{ \textbackslash{}s }\OtherTok{->} \DataTypeTok{Left}\NormalTok{ (}\FunctionTok{mempty}\NormalTok{, s)}
\NormalTok{  f }\FunctionTok{<|>}\NormalTok{ g }\FunctionTok{=} \DataTypeTok{Parser} \FunctionTok{$}\NormalTok{ \textbackslash{}s }\OtherTok{->} \KeywordTok{case}\NormalTok{ runParser f s }\KeywordTok{of}
    \DataTypeTok{Right}\NormalTok{ x }\OtherTok{->} \DataTypeTok{Right}\NormalTok{ x}
    \DataTypeTok{Left}\NormalTok{  (e, s') }\OtherTok{->} \KeywordTok{case}\NormalTok{ runParser g s' }\KeywordTok{of}
      \DataTypeTok{Right}\NormalTok{ x }\OtherTok{->} \DataTypeTok{Right}\NormalTok{ x}
      \DataTypeTok{Left}\NormalTok{  (f, s'') }\OtherTok{->} \DataTypeTok{Left}\NormalTok{ (e }\OtherTok{`mappend`}\NormalTok{ f, s'')}
\end{Highlighting}
\end{Shaded}

\begin{theorem}

\label{thm:with-heuristic}

Both versions satisfy the laws of
\VERB|\DataTypeTok{Alternative}|.\footnote{Note, however, that
  \VERB|\DataTypeTok{SParser}| from
  \cref{sec:parser-combinators:without-access} can only do backtracking
  because, unlike \VERB|\DataTypeTok{Parser}|, it is asymmetric in its
  use of \VERB|\DataTypeTok{Either}|.}

\end{theorem}

\begin{proof}

By case analysis.

Note that to convince yourself that
\VERB|\NormalTok{(}\FunctionTok{<\VerbBar{}>}\NormalTok{)}| is
associative it is enough to observe that in
\VERB|\NormalTok{a }\FunctionTok{<\VerbBar{}>}\NormalTok{ b }\FunctionTok{<\VerbBar{}>}\NormalTok{ c}|
for the above definitions

\begin{itemize}
\tightlist
\item
  \VERB|\DataTypeTok{Right}| is a zero,
\item
  values of \VERB|\NormalTok{e}| always propagate to the right,
\item
  while \VERB|\NormalTok{s}| is stays constant in the roll-back version,
  or always propagates in the no-roll-back version, but never both.
\end{itemize}

Which means that parentheses can't influence anything in either case.

The same idea can be used in similar proofs involving similar operators
of \VERB|\DataTypeTok{State}| and \VERB|\DataTypeTok{Parser}|.

\end{proof}

From the popular Haskell parser combinator libraries mentioned above
Attoparsec rolls-back while Parsec and Megaparsec do not, instead they
implement backtracking with a separate combinator for which we could
give the following type signature

\begin{Shaded}
\begin{Highlighting}[]
\OtherTok{try ::} \DataTypeTok{Parser}\NormalTok{ s e a }\OtherTok{->} \DataTypeTok{Parser}\NormalTok{ s e a}
\end{Highlighting}
\end{Shaded}

\hypertarget{examples}{%
\subsubsection{Examples}\label{examples}}

\label{sec:parser-combinators:examples}

The already given definitions allow us enough headroom to define some
primitive parsers and a couple of examples. For instance, assuming
\VERB|\DataTypeTok{Alternative}| rolls-back we can write

\begin{Shaded}
\begin{Highlighting}[]
\KeywordTok{type} \DataTypeTok{Failures} \FunctionTok{=}\NormalTok{ [}\DataTypeTok{String}\NormalTok{]}

\OtherTok{eof ::} \DataTypeTok{Parser} \DataTypeTok{String} \DataTypeTok{Failures}\NormalTok{ ()}
\NormalTok{eof }\FunctionTok{=} \DataTypeTok{Parser} \FunctionTok{$}\NormalTok{ \textbackslash{}s }\OtherTok{->} \KeywordTok{case}\NormalTok{ s }\KeywordTok{of}
\NormalTok{  [] }\OtherTok{->} \DataTypeTok{Right}\NormalTok{ ((), s)}
\NormalTok{  _  }\OtherTok{->} \DataTypeTok{Left}\NormalTok{  ([}\StringTok{"expected eof"}\NormalTok{], s)}

\OtherTok{char ::} \DataTypeTok{Char} \OtherTok{->} \DataTypeTok{Parser} \DataTypeTok{String} \DataTypeTok{Failures}\NormalTok{ ()}
\NormalTok{char x }\FunctionTok{=} \DataTypeTok{Parser} \FunctionTok{$}\NormalTok{ \textbackslash{}s }\OtherTok{->} \KeywordTok{case}\NormalTok{ s }\KeywordTok{of}
\NormalTok{  []     }\OtherTok{->} \DataTypeTok{Left}\NormalTok{  ([}\StringTok{"unexpected eof"}\NormalTok{], s)}
\NormalTok{  (c}\FunctionTok{:}\NormalTok{cs) }\OtherTok{->} \KeywordTok{if}\NormalTok{ (c }\FunctionTok{==}\NormalTok{ x)}
    \KeywordTok{then} \DataTypeTok{Right}\NormalTok{ ((), cs)}
    \KeywordTok{else} \DataTypeTok{Left}\NormalTok{ ([}\StringTok{"expected `"} \FunctionTok{++}\NormalTok{ [x] }\FunctionTok{++} \StringTok{"' got `"} \FunctionTok{++}\NormalTok{ [c] }\FunctionTok{++} \StringTok{"'"}\NormalTok{], s)}

\OtherTok{string ::} \DataTypeTok{String} \OtherTok{->} \DataTypeTok{Parser} \DataTypeTok{String} \DataTypeTok{Failures}\NormalTok{ ()}
\NormalTok{string [] }\FunctionTok{=} \FunctionTok{pure}\NormalTok{ ()}
\NormalTok{string (c}\FunctionTok{:}\NormalTok{cs) }\FunctionTok{=}\NormalTok{ char c }\FunctionTok{>>}\NormalTok{ string cs}

\NormalTok{parseTest }\FunctionTok{=}\NormalTok{ runParser (string }\StringTok{"foo"}\NormalTok{) }\StringTok{"foo bar"}
            \FunctionTok{==} \DataTypeTok{Right}\NormalTok{((), }\StringTok{" bar"}\NormalTok{)}
         \FunctionTok{&&}\NormalTok{ runParser (string }\StringTok{"abb"} \FunctionTok{<|>}\NormalTok{ string }\StringTok{"abc"}\NormalTok{) }\StringTok{"aba"}
            \FunctionTok{==} \DataTypeTok{Left}\NormalTok{ ([}\StringTok{"expected `b' got `a'"}
\NormalTok{                     ,}\StringTok{"expected `c' got `a'"}\NormalTok{], }\StringTok{"aba"}\NormalTok{)}
\end{Highlighting}
\end{Shaded}

To use the other implementation of \VERB|\DataTypeTok{Alternative}| we
would need to wrap all calls to \VERB|\NormalTok{string}| on the left
hand sides of
\VERB|\NormalTok{(}\FunctionTok{<\VerbBar{}>}\NormalTok{)}| with
\VERB|\NormalTok{try}|s.

Semantics-wise our \VERB|\DataTypeTok{Parser}| combines features of
Attoparsec (backtracking) and Megaparsec (custom error types). Of
course, it fits on a single page only because it has a minuscule number
of features in comparison to either of the two. To make it practical we
would need, at the very least, to implement tracking of the position in
the input stream and a bunch of primitive parsers, which we leave as an
exercise to the interested reader.

Interestingly, this exact implementation of handling of errors by
accumulation via \VERB|\DataTypeTok{Alternative}| over a
\VERB|\DataTypeTok{Monoid}| seems to be novel (although, pretty
trivial). Megaparsec, however, does something very similar by
accumulating errors in \VERB|\DataTypeTok{Set}|s instead of
\VERB|\DataTypeTok{Monoid}|s.

\VERB|\DataTypeTok{MonadTrans}|former versions of these structures can
be trivially obtained by adding \VERB|\DataTypeTok{Monad}|ic index
\VERB|\NormalTok{m}| after the arrow in definition of
\VERB|\DataTypeTok{Parser}| (i.e. by exposing the internal
\VERB|\DataTypeTok{Monad}| of the original
\VERB|\DataTypeTok{MonadTrans}| stack) and correspondingly tweaking
base-level definitions and all type signatures.

\hypertarget{other-variants-of-monadcatch}{%
\subsection{Other variants of
MonadCatch}\label{other-variants-of-monadcatch}}

\label{sec:other-monadic-generalizations}

Finally, worth mentioning are two lesser-known variants of structures
similar to structures of \cref{sec:monadic-generalizations}. The first
one is defined in \VERB|\DataTypeTok{Control.Monad.Exception.Catch}|
module of
\texttt{control-monad-exception}~\cite{Hackage:control-monad-exception0112}
package as

\begin{Shaded}
\begin{Highlighting}[]
\KeywordTok{class}\NormalTok{ (}\DataTypeTok{Monad}\NormalTok{ m, }\DataTypeTok{Monad}\NormalTok{ n) }\OtherTok{=>} \DataTypeTok{MonadCatch}\NormalTok{ e m n }\FunctionTok{|}\NormalTok{ e m }\OtherTok{->}\NormalTok{ n, e n }\OtherTok{->}\NormalTok{ m }\KeywordTok{where}
\OtherTok{   catch ::}\NormalTok{ m a }\OtherTok{->}\NormalTok{ (e }\OtherTok{->}\NormalTok{ n a) }\OtherTok{->}\NormalTok{ n a}
\end{Highlighting}
\end{Shaded}

\noindent and the second one in
\VERB|\DataTypeTok{Control.Monad.Catch.Class}| module of
\texttt{catch-fd}~\cite{Hackage:catch-fd0202} package

\begin{Shaded}
\begin{Highlighting}[]
\KeywordTok{class} \DataTypeTok{Monad}\NormalTok{ m }\OtherTok{=>} \DataTypeTok{MonadThrow}\NormalTok{ e m }\FunctionTok{|}\NormalTok{ m }\OtherTok{->}\NormalTok{ e }\KeywordTok{where}
\OtherTok{  throw ::}\NormalTok{ e }\OtherTok{->}\NormalTok{ m a}

\KeywordTok{class}\NormalTok{ (}\DataTypeTok{MonadThrow}\NormalTok{ e m, }\DataTypeTok{Monad}\NormalTok{ n) }\OtherTok{=>} \DataTypeTok{MonadCatch}\NormalTok{ e m n }\FunctionTok{|}\NormalTok{ n e }\OtherTok{->}\NormalTok{ m }\KeywordTok{where}
\OtherTok{  catch ::}\NormalTok{ m a }\OtherTok{->}\NormalTok{ (e }\OtherTok{->}\NormalTok{ n a) }\OtherTok{->}\NormalTok{ n a}
\end{Highlighting}
\end{Shaded}

Note that \texttt{control-monad-exception} does not define a type class
with a \VERB|\NormalTok{throw}| operator, that library provides a
universal computation type \VERB|\DataTypeTok{EM}| (similar to
\VERB|\DataTypeTok{EIO}| of \cref{sec:instances:eio}) with such an
operator instead. Also note that the common point of those two
definitions is that both \VERB|\FunctionTok{catch}| operators change the
type of computations from \VERB|\NormalTok{m}| to \VERB|\NormalTok{n}|.

\hypertarget{the-nature-of-an-error}{%
\section{The nature of an error}\label{the-nature-of-an-error}}

\label{sec:init}

Lets forget for a minute about every concrete algebraic error-handling
structure mentioned before and try to invent our own algebra of
computations by reasoning like a purely pragmatic programmer who likes
to make everything typed as precisely as possible.

We start, of course, by pragmatically naming our type of computations to
be \VERB|\DataTypeTok{C}|. Then, we reason, it should be indexed by both
the type of the result, which we shall pragmatically call
\VERB|\NormalTok{a}|, and the type of exceptions \VERB|\NormalTok{e}|.
We are not sure about the body of that definition, so we just leave it
undefined

\begin{Shaded}
\begin{Highlighting}[]
\KeywordTok{data} \DataTypeTok{C}\NormalTok{ e a}
\end{Highlighting}
\end{Shaded}

Now, we know that \VERB|\DataTypeTok{Monad}|s usually work pretty well
for the computation part (since we can as well just lift everything into
\VERB|\DataTypeTok{IO}| which is a \VERB|\DataTypeTok{Monad}|), so we
write

\begin{Shaded}
\begin{Highlighting}[]
\FunctionTok{return}\OtherTok{ ::}\NormalTok{ a }\OtherTok{->} \DataTypeTok{C}\NormalTok{ e a}

\OtherTok{(>>=) ::} \DataTypeTok{C}\NormalTok{ e a }\OtherTok{->}\NormalTok{ (a }\OtherTok{->} \DataTypeTok{C}\NormalTok{ e b) }\OtherTok{->} \DataTypeTok{C}\NormalTok{ e b}
\end{Highlighting}
\end{Shaded}

\noindent and expect these operators to satisfy
\VERB|\DataTypeTok{Monad}| laws (\cref{sec:monad}).

Meanwhile, pragmatically, an "exceptional" execution path requires two
conventional operators:

\begin{itemize}
\item
  a method of raising an exception; the type of this operator seems to
  be pretty straightforward

\begin{Shaded}
\begin{Highlighting}[]
\OtherTok{throw ::}\NormalTok{ e }\OtherTok{->} \DataTypeTok{C}\NormalTok{ e a}
\end{Highlighting}
\end{Shaded}

  as it simply injects the error into \VERB|\DataTypeTok{C}|,
\item
  and a method to catch exceptions; the overly-general type for this
  operator is, again, pretty straightforward

\begin{Shaded}
\begin{Highlighting}[]
\FunctionTok{catch}\OtherTok{ ::} \DataTypeTok{C}\NormalTok{ e a }\OtherTok{->}\NormalTok{ (e }\OtherTok{->} \DataTypeTok{C}\NormalTok{ f b) }\OtherTok{->} \DataTypeTok{C}\NormalTok{ g c}
\end{Highlighting}
\end{Shaded}

  The only obvious requirement here is that the type the "handler"
  function (the second argument of \VERB|\FunctionTok{catch}|) can
  handle should coincide with the type of errors the "computation" (the
  first argument) can \VERB|\NormalTok{throw}|.
\end{itemize}

Finally, we pragmatically expect the above to obey the conventional
operational semantics of error handling operators, giving us the
following definition.

\begin{definition}

\label{dfn:structure}

\textbf{\textbf{Pragmatic error handling structure.}} Structure
\VERB|\OtherTok{m ::} \FunctionTok{*} \OtherTok{=>} \FunctionTok{*} \OtherTok{=>} \FunctionTok{*}|
with \VERB|\FunctionTok{return}|,
\VERB|\NormalTok{(}\FunctionTok{>>=}\NormalTok{)}|,
\VERB|\NormalTok{throw}|, and \VERB|\FunctionTok{catch}| operators
satisfying

\begin{enumerate}
\item
  \label{dfn:structure:monad} \VERB|\FunctionTok{return}| and
  \VERB|\NormalTok{(}\FunctionTok{>>=}\NormalTok{)}| obey
  \VERB|\DataTypeTok{Monad}| laws (\cref{sec:monad}),
\item
  \label{dfn:structure:throw-bind}
  \VERB|\NormalTok{throw e }\FunctionTok{>>=}\NormalTok{ f }\FunctionTok{==}\NormalTok{ throw e}|
  ("\VERB|\NormalTok{throw}|ing of an error stops the computation"),
\item
  \label{dfn:structure:throw-catch}
  \VERB|\NormalTok{throw e }\OtherTok{`catch`}\NormalTok{ f }\FunctionTok{==}\NormalTok{ f e}|
  ("\VERB|\NormalTok{throw}|ing of an error invokes the most recent
  error handler"),\footnote{Similarly to GHC's imprecise exceptions of
    \cref{sec:imprecise} dynamic dispatch can be implemented on top of
    such a structure. We shall do this in
    \cref{sec:instances:constant:monadcatch}.}
\item
  \label{dfn:structure:return-catch}
  \VERB|\FunctionTok{return}\NormalTok{ a }\OtherTok{`catch`}\NormalTok{ f }\FunctionTok{==} \FunctionTok{return}\NormalTok{ a}|
  ("\VERB|\FunctionTok{return}| is not an error").
\end{enumerate}

\end{definition}

\hypertarget{the-type-of-error-handling-operator}{%
\section{The type of error handling
operator}\label{the-type-of-error-handling-operator}}

\label{sec:type-of-catch}

The first question to the structure of \VERB|\DataTypeTok{C}| is, of
course, what is the precise type of \VERB|\FunctionTok{catch}| operator.

\begin{Shaded}
\begin{Highlighting}[]
\FunctionTok{catch}\OtherTok{ ::} \DataTypeTok{C}\NormalTok{ e a }\OtherTok{->}\NormalTok{ (e }\OtherTok{->} \DataTypeTok{C}\NormalTok{ f b) }\OtherTok{->} \DataTypeTok{C}\NormalTok{ g c}
\end{Highlighting}
\end{Shaded}

\noindent In other words, we would like to know which of the variables
\VERB|\NormalTok{f}|, \VERB|\NormalTok{g}|, \VERB|\NormalTok{b}|, and
\VERB|\NormalTok{c}| in this signature should have their own universal
quantifier and which should be substituted with others. The answer comes
by considering several cases.

\begin{itemize}
\item
  Firstly, let us consider the following expression.

\begin{Shaded}
\begin{Highlighting}[]
\FunctionTok{return}\NormalTok{ a }\OtherTok{`catch`}\NormalTok{ f}
\end{Highlighting}
\end{Shaded}

  The expected semantics of \VERB|\FunctionTok{catch}| requires (by
  \cref{dfn:structure:return-catch} of \cref{dfn:structure})

\begin{Shaded}
\begin{Highlighting}[]
\FunctionTok{return}\NormalTok{ a }\OtherTok{`catch`}\NormalTok{ f }\FunctionTok{==} \FunctionTok{return}\NormalTok{ a}
\end{Highlighting}
\end{Shaded}

  Note that the most general type for
  \VERB|\FunctionTok{return}\NormalTok{ a}| expression is
  \VERB|\KeywordTok{forall}\NormalTok{ e }\FunctionTok{.} \DataTypeTok{C}\NormalTok{ e a}|
  for \VERB|\NormalTok{a }\FunctionTok{:}\NormalTok{ a}|\footnote{The
    reader might have noticed already that we abuse notation somewhat by
    assuming type variables and term variables use distinct namespaces.
    This expression happens to be the first and the only one that uses
    both at the same time, hence it looks like an exiting "type-in-type"
    kind of thing, but it is not, it is ordinarily boring.}. Moreover,
  we can assign the same type to any expression that does not
  \VERB|\NormalTok{throw}| since

  \begin{itemize}
  \item
    both \VERB|\NormalTok{a}| and \VERB|\NormalTok{e}| in the type
    signify the potential to \VERB|\FunctionTok{return}| and
    \VERB|\NormalTok{throw}| values of the corresponding types,
  \item
    and an expression that does not \VERB|\NormalTok{throw}| any errors
    can be said to not-\VERB|\NormalTok{throw}| an error of any
    particular type, similarly to how bottom elimination rule works. Or,
    equivalently, any such computation can be said to
    \VERB|\NormalTok{throw}| values of an empty type and an empty type
    can always be replaced with any other type by bottom
    elimination.\footnote{Implicitly or with
      \VERB|\NormalTok{f }\OtherTok{`catch`}\NormalTok{ bot}\FunctionTok{-}\NormalTok{elim}|
      which is extentionally equal to \VERB|\NormalTok{f}|.}
  \end{itemize}
\item
  Now let us consider the following expression, assuming
  \VERB|\NormalTok{e}| and \VERB|\NormalTok{f}| are of different types
  (i.e. both the computation and the handler throw different
  exceptions).

\begin{Shaded}
\begin{Highlighting}[]
\NormalTok{throw e }\OtherTok{`catch`}\NormalTok{ (\textbackslash{}_ }\OtherTok{->}\NormalTok{ throw f)}
\end{Highlighting}
\end{Shaded}

  The expected semantics of \VERB|\FunctionTok{catch}| requires (by
  \cref{dfn:structure:throw-catch} of \cref{dfn:structure})

\begin{Shaded}
\begin{Highlighting}[]
\NormalTok{throw e }\OtherTok{`catch`}\NormalTok{ (\textbackslash{}_ }\OtherTok{->}\NormalTok{ throw f) }\FunctionTok{==}\NormalTok{ throw f}
\end{Highlighting}
\end{Shaded}
\end{itemize}

These two cases show that \VERB|\NormalTok{g}| should be substituted
with \VERB|\NormalTok{f}| and \VERB|\NormalTok{e}| should be kept
separate from \VERB|\NormalTok{f}| because

\begin{itemize}
\item
  if computation \VERB|\NormalTok{throw}|s then the type
  \VERB|\NormalTok{f}| in the handler "wins",
\item
  but if it does not \VERB|\NormalTok{throw}| then \VERB|\NormalTok{e}|
  is an empty type and it can be substituted for any other type,
  including \VERB|\NormalTok{f}| (similarly to the type of
  \VERB|\FunctionTok{return}| above)\footnote{The only nontrivial
    observation in this section.}
\item
  these two cases are mutually exclusive.
\end{itemize}

That is, the type for \VERB|\FunctionTok{catch}| is at most as general
as

\begin{Shaded}
\begin{Highlighting}[]
\FunctionTok{catch}\OtherTok{ ::} \KeywordTok{forall}\NormalTok{ e f }\FunctionTok{.} \DataTypeTok{C}\NormalTok{ e a }\OtherTok{->}\NormalTok{ (e }\OtherTok{->} \DataTypeTok{C}\NormalTok{ f b) }\OtherTok{->} \DataTypeTok{C}\NormalTok{ f c}
\end{Highlighting}
\end{Shaded}

\begin{itemize}
\item
  Continuing, \cref{dfn:structure:return-catch} of \cref{dfn:structure}
  shows that \VERB|\NormalTok{c}| has to coincide with
  \VERB|\NormalTok{a}|.
\item
  Similarly, \cref{dfn:structure:throw-catch} requires

\begin{Shaded}
\begin{Highlighting}[]
\NormalTok{throw e }\OtherTok{`catch`}\NormalTok{ (\textbackslash{}_ }\OtherTok{->} \FunctionTok{return}\NormalTok{ a) }\FunctionTok{==} \FunctionTok{return}\NormalTok{ a}
\end{Highlighting}
\end{Shaded}

  which shows that \VERB|\NormalTok{c}| has to coincide with
  \VERB|\NormalTok{b}|.
\end{itemize}

All these observations combine into the following.\footnote{\label{fn:its-dual}Spoilers!
  The reader is only supposed to notice the following after reading
  \cref{sec:logical}.\\
  ~\\
  Note that we could have written an equivalent up to names of operators
  sections~\ref{sec:init} and~\ref{sec:type-of-catch} that explained why
  the type of \VERB|\NormalTok{(}\FunctionTok{>>=}\NormalTok{)}| is the
  correct type for sequencing computations in \VERB|\DataTypeTok{C}|
  given that error handling should be done
  \VERB|\DataTypeTok{Monad}|ically. In particular, the fact that the
  dual of \cref{dfn:structure} lists valid operational equations is a
  rather curious observation by itself. Which is another reason why we
  disagree with the conventional wisdom in footnote~\ref{fn:terms}.}

\begin{theorem}

\label{thm:catch-type}

For any type
\VERB|\DataTypeTok{C}\OtherTok{ ::} \FunctionTok{*} \OtherTok{=>} \FunctionTok{*} \OtherTok{=>} \FunctionTok{*}|
obeying \cref{dfn:structure} the most general type for the
\VERB|\FunctionTok{catch}| operator is

\begin{Shaded}
\begin{Highlighting}[]
\FunctionTok{catch}\OtherTok{ ::} \KeywordTok{forall}\NormalTok{ a e f }\FunctionTok{.} \DataTypeTok{C}\NormalTok{ e a }\OtherTok{->}\NormalTok{ (e }\OtherTok{->} \DataTypeTok{C}\NormalTok{ f a) }\OtherTok{->} \DataTypeTok{C}\NormalTok{ f a}
\end{Highlighting}
\end{Shaded}

\end{theorem}

\begin{proof}

By the above reasoning. That is, by simple unification of types of
\VERB|\FunctionTok{return}|, \VERB|\NormalTok{throw}|,
\VERB|\NormalTok{(}\FunctionTok{>>=}\NormalTok{)}| operators of
\cref{dfn:structure} and the following equations that are consequences
of equations of \cref{dfn:structure}

\begin{Shaded}
\begin{Highlighting}[]
\FunctionTok{return}\NormalTok{ a }\OtherTok{`catch`}\NormalTok{ f }\FunctionTok{==} \FunctionTok{return}\NormalTok{ a}
\NormalTok{throw e }\OtherTok{`catch`}\NormalTok{ (\textbackslash{}_ }\OtherTok{->} \FunctionTok{return}\NormalTok{ a) }\FunctionTok{==} \FunctionTok{return}\NormalTok{ a}
\NormalTok{throw e }\OtherTok{`catch`}\NormalTok{ (\textbackslash{}_ }\OtherTok{->}\NormalTok{ throw f) }\FunctionTok{==}\NormalTok{ throw f}
\end{Highlighting}
\end{Shaded}

\end{proof}

\hypertarget{conjoinedly-monadic-algebra}{%
\section{Conjoinedly Monadic
algebra}\label{conjoinedly-monadic-algebra}}

\label{sec:conjoinedly-monadic}

After \cref{thm:catch-type} it becomes hard to ignore the fact that
\VERB|\NormalTok{throw}| has the type of \VERB|\FunctionTok{return}| and
\VERB|\FunctionTok{catch}| has the type of
\VERB|\NormalTok{(}\FunctionTok{>>=}\NormalTok{)}| in the "wrong" index
for \VERB|\DataTypeTok{C}|. Moreover, \cref{dfn:structure:throw-catch}
of \cref{dfn:structure} looks exactly like a left identity law for
\VERB|\DataTypeTok{Monad}| (\cref{sec:monad}). While it is not as
immediately clear that \VERB|\FunctionTok{catch}| should be associative,
it seems only natural to ask whenever the following conjoinedly
\VERB|\DataTypeTok{Monad}|ic restriction of \cref{dfn:structure} has any
instances.

\begin{definition}

\label{dfn:proper}

\textbf{\textbf{Conjoinedly monadic error algebra}}. A type
\VERB|\OtherTok{m ::} \FunctionTok{*} \OtherTok{=>} \FunctionTok{*} \OtherTok{=>} \FunctionTok{*}|
for which

\begin{itemize}
\item
  \label{dfn:proper:bind-monad} \VERB|\NormalTok{m}| is a
  \VERB|\DataTypeTok{Monad}| in its second index (that is,
  \VERB|\NormalTok{m e}| is a \VERB|\DataTypeTok{Monad}| for all
  \VERB|\NormalTok{e}|),
\item
  \label{dfn:proper:catch-monad} \VERB|\NormalTok{m}| is a
  \VERB|\DataTypeTok{Monad}| in its first index (that is,
  \VERB|\NormalTok{\textbackslash{}e }\FunctionTok{.}\NormalTok{ m e a}|
  is a \VERB|\DataTypeTok{Monad}| for all \VERB|\NormalTok{a}|),
\end{itemize}

and assuming

\begin{itemize}
\item
  the names of \VERB|\DataTypeTok{Monad}| operators in the second index
  of \VERB|\NormalTok{m}| are \VERB|\FunctionTok{return}| and
  \VERB|\NormalTok{(}\FunctionTok{>>=}\NormalTok{)}|,
\item
  the names of \VERB|\DataTypeTok{Monad}| operators in the first index
  are \VERB|\NormalTok{throw}| and \VERB|\FunctionTok{catch}|,
\end{itemize}

the following equations hold

\begin{enumerate}
\item
  \label{dfn:proper:return-catch}
  \VERB|\FunctionTok{return}\NormalTok{ x }\OtherTok{`catch`}\NormalTok{ f }\FunctionTok{==} \FunctionTok{return}\NormalTok{ x}|,
\item
  \label{dfn:proper:throw-bind}
  \VERB|\NormalTok{throw e }\FunctionTok{>>=}\NormalTok{ f }\FunctionTok{==}\NormalTok{ throw e}|.
\end{enumerate}

\end{definition}

If we replace \VERB|\DataTypeTok{Monad}| in \cref{dfn:proper} with
\VERB|\DataTypeTok{MonadFish}| (\cref{sec:monad-fish}), as usual, the
latter two equations become a bit clearer.

\newpage

\begin{definition}

\label{dfn:fishy}

\textbf{\textbf{Fishy conjoinedly monadic error algebra}}. A type
\VERB|\OtherTok{m ::} \FunctionTok{*} \OtherTok{=>} \FunctionTok{*} \OtherTok{=>} \FunctionTok{*}|
for which

\begin{itemize}
\item
  \label{dfn:fishy:bind-monad} \VERB|\NormalTok{m}| is a
  \VERB|\DataTypeTok{MonadFish}| in its second index,
\item
  \label{dfn:fishy:catch-monad} \VERB|\NormalTok{m}| is a
  \VERB|\DataTypeTok{MonadFish}| in its first index,
\end{itemize}

and assuming

\begin{itemize}
\item
  the names of \VERB|\DataTypeTok{MonadFish}| operators in the second
  index are \VERB|\FunctionTok{return}| and
  \VERB|\NormalTok{(}\FunctionTok{>=>}\NormalTok{)}|,
\item
  the names of \VERB|\DataTypeTok{MonadFish}| operators in the first
  index are \VERB|\NormalTok{throw}| and \VERB|\NormalTok{handle}|,
\end{itemize}

the following equations hold

\begin{enumerate}
\item
  \label{dfn:fishy:return-catch}
  \VERB|\FunctionTok{return} \OtherTok{`handle`}\NormalTok{ f }\FunctionTok{==} \FunctionTok{return}|,
\item
  \label{dfn:fishy:throw-bind}
  \VERB|\NormalTok{throw }\FunctionTok{>=>}\NormalTok{ f }\FunctionTok{==}\NormalTok{ throw}|.
\end{enumerate}

\end{definition}

On other words, definitions~\ref{dfn:proper} and~\ref{dfn:fishy} define
a structure that is a \VERB|\DataTypeTok{Monad}|
(\VERB|\DataTypeTok{MonadFish}|) twice and for which
\VERB|\FunctionTok{return}| is a left zero for
\VERB|\FunctionTok{catch}| (\VERB|\NormalTok{handle}|) and
\VERB|\NormalTok{throw}| is a left zero for
\VERB|\NormalTok{(}\FunctionTok{>>=}\NormalTok{)}|
(\VERB|\NormalTok{(}\FunctionTok{>=>}\NormalTok{)}|).

\hypertarget{instances-either}{%
\section{Instances: Either}\label{instances-either}}

\label{sec:instances:either}

Pragmatic programmer finally loses last bits of concentration realizing
that \VERB|\DataTypeTok{Either}| type seems to match requirements of
\cref{dfn:proper} and goes into sources to check whenever Haskell's
standard library already has such a \VERB|\FunctionTok{catch}|.
Unfortunately, \VERB|\DataTypeTok{Data.Either}| module does not define
such an operator. However, \VERB|\NormalTok{catchE}| and
\VERB|\NormalTok{throwE}| of \VERB|\DataTypeTok{ExceptT}|
(\cref{sec:either-monadtrans}) match. Of course, if we substitute
\VERB|\DataTypeTok{Identity}| for \VERB|\NormalTok{m}|,
\VERB|\DataTypeTok{ExceptT}| turns into \VERB|\DataTypeTok{Either}| and
those operators can be simplified to

\begin{Shaded}
\begin{Highlighting}[]
\OtherTok{throwE' ::}\NormalTok{ e }\OtherTok{->} \DataTypeTok{Either}\NormalTok{ e a}
\NormalTok{throwE' }\FunctionTok{=} \DataTypeTok{Left}

\OtherTok{catchE' ::} \DataTypeTok{Either}\NormalTok{ e a}
        \OtherTok{->}\NormalTok{ (e }\OtherTok{->} \DataTypeTok{Either}\NormalTok{ f a)}
        \OtherTok{->} \DataTypeTok{Either}\NormalTok{ f a}
\NormalTok{catchE' (}\DataTypeTok{Left}\NormalTok{ e)  h }\FunctionTok{=}\NormalTok{ h e}
\NormalTok{catchE' (}\DataTypeTok{Right}\NormalTok{ a) _ }\FunctionTok{=} \DataTypeTok{Right}\NormalTok{ a}
\end{Highlighting}
\end{Shaded}

\begin{lemma}

\label{thm:ExceptT-monad} For a given \VERB|\DataTypeTok{Monad}|
\VERB|\NormalTok{m}| and a fixed argument \VERB|\NormalTok{a}|,
\VERB|\DataTypeTok{ExceptT}| with \VERB|\NormalTok{throwE}| as
\VERB|\FunctionTok{return}| and \VERB|\NormalTok{catchE}| as
\VERB|\NormalTok{(}\FunctionTok{>>=}\NormalTok{)}| is a
\VERB|\DataTypeTok{Monad}| in argument \VERB|\NormalTok{e}|.

\end{lemma}

\begin{proof}

Any of the following

\begin{itemize}
\item
  \textbf{\textbf{By brute force:}} by case analysis, using the fact
  that \VERB|\NormalTok{m}| satisfies \VERB|\DataTypeTok{Monad}| laws.
\item
  \textbf{\textbf{Another way:}} trivial consequence of
  \cref{sec:logical}.
\end{itemize}

\end{proof}

\begin{lemma}

\label{thm:ExceptT-zeroes}

For \VERB|\DataTypeTok{ExceptT}| with the above operators the following
equations hold

\begin{enumerate}
\item
  \VERB|\FunctionTok{return}\NormalTok{ x }\OtherTok{`catchE`}\NormalTok{ f }\FunctionTok{==} \FunctionTok{return}\NormalTok{ x}|,
\item
  \VERB|\NormalTok{throwE e }\FunctionTok{>>=}\NormalTok{ f }\FunctionTok{==}\NormalTok{ throwE e}|.
\end{enumerate}

\end{lemma}

\begin{proof}

By trivial case analysis.

\end{proof}

\begin{theorem}

\label{thm:ExceptT-proper}

\VERB|\DataTypeTok{ExceptT}| and, by consequence,
\VERB|\DataTypeTok{Either}| satisfy \cref{dfn:proper}.

\end{theorem}

\begin{proof}

Consequence of \cref{thm:ExceptT-monad} and \cref{thm:ExceptT-zeroes}.

\end{proof}

\hypertarget{logical-perspective}{%
\section{Logical perspective}\label{logical-perspective}}

\label{sec:logical}

Note, that from a logical perspective most of the above is simply
trivial. \VERB|\DataTypeTok{Either}\NormalTok{ a b}| is just
\(a \lor b\) and so if \(\lambda b . a \lor b\) is a
\VERB|\DataTypeTok{Monad}| then \(\lambda a . a \lor b\) must be a
\VERB|\DataTypeTok{Monad}| too since \(\lor\) operator is symmetric.
Sections~\ref{sec:init}-\ref{sec:conjoinedly-monadic} simply generalize
this fact with interactions between \VERB|\DataTypeTok{Left}|,
\VERB|\DataTypeTok{Right}| and two
\VERB|\NormalTok{(}\FunctionTok{>>=}\NormalTok{)}| operators into
\cref{dfn:proper}.\cref{fn:its-dual}

The main point of this article is that \textbf{\textbf{there are other
instances}} of this generalization and, more importantly, that
\textbf{\textbf{this generalization is itself interesting}} --- the
facts that we shall demonstrate in the sections that follow.

\hypertarget{encodings}{%
\section{Encodings}\label{encodings}}

\label{sec:encodings}

Despite the noted triviality, these facts do not seem to be appreciated
by the wider Haskell community. In particular:

\begin{itemize}
\item
  \VERB|\DataTypeTok{ExceptT}| does not get much use in Hackage packages
  in general,
\item
  the equivalent of \VERB|\NormalTok{catchE}| for
  \VERB|\DataTypeTok{ErrorT}| has an overly-restricted type

\begin{Shaded}
\begin{Highlighting}[]
\OtherTok{catchError ::}\NormalTok{ (}\DataTypeTok{Monad}\NormalTok{ m)}
           \OtherTok{=>} \DataTypeTok{ErrorT}\NormalTok{ e m a}
           \OtherTok{->}\NormalTok{ (e }\OtherTok{->} \DataTypeTok{ErrorT}\NormalTok{ e m a)}
           \OtherTok{->} \DataTypeTok{ErrorT}\NormalTok{ e m a}
\NormalTok{m }\OtherTok{`catchError`}\NormalTok{ h }\FunctionTok{=} \DataTypeTok{ErrorT} \FunctionTok{$} \KeywordTok{do}
\NormalTok{    a }\OtherTok{<-}\NormalTok{ runErrorT m}
    \KeywordTok{case}\NormalTok{ a }\KeywordTok{of}
        \DataTypeTok{Left}\NormalTok{  l }\OtherTok{->}\NormalTok{ runErrorT (h l)}
        \DataTypeTok{Right}\NormalTok{ r }\OtherTok{->} \FunctionTok{return}\NormalTok{ (}\DataTypeTok{Right}\NormalTok{ r)}
\end{Highlighting}
\end{Shaded}
\item
  no \VERB|\DataTypeTok{Monad}|ic parsing combinator library from
  Hackage (most obvious beneficiaries of the observation) defines the
  would-be-\VERB|\DataTypeTok{Monad}| instance of
  \VERB|\NormalTok{throwE}| and \VERB|\NormalTok{catchE}|.
\end{itemize}

To our best knowledge, the only Hackage package that is explicitly aware
of the fact that \VERB|\DataTypeTok{Either}| is a
\VERB|\DataTypeTok{Monad}| twice is
\texttt{errors}~\cite{Hackage:errors230}\footnote{In~\cite{Gonzalez:2012:SEH}
  Gabriel Gonzalez, the author of the \texttt{errors} package, also
  explicitly mentions the fact that the \VERB|\DataTypeTok{Monad}|ic
  operators for the other index of \VERB|\DataTypeTok{Either}| seem to
  match the semantics for the corresponding \VERB|\NormalTok{throw}| and
  \VERB|\FunctionTok{catch}| operators. Though he gives no proofs or
  claims of general applicability, he mentions that the fact itself was
  first pointed out to him by Elliott Hird who named it the "success
  \VERB|\DataTypeTok{Monad}|". So, though the Net seems to have no
  evidence of that conversation, it is entirely possible some of the
  discussed facts were already discovered in their complete forms before
  (at least in the idealistic sense, but not, to our best knowledge, in
  the "communicated in this form before" sense). (Which is usually the
  case for almost anything anyway.) (Which is a yet another reminder
  that "intellectual property" is an oxymoron.)} and the only packages
that seem to be aware that \VERB|\NormalTok{throw}| and
\VERB|\FunctionTok{catch}| in general need more general types than those
given by \VERB|\DataTypeTok{MonadCatch}| of
\cref{sec:monadic-generalizations} are those discussed in
\cref{sec:other-monadic-generalizations} (but they miss the fact that
their \VERB|\FunctionTok{catch}| operators want to be
\VERB|\DataTypeTok{Monad}|ic \VERB|\NormalTok{bind}|s). To our best
knowledge, no Hackage package utilizes both facts.

As to the question why had not anybody notice and start exploiting these
facts yet we hypothesize that the answer is because Haskell cannot
express these properties conveniently (not to mention less expressive
mainstream languages which cannot express them at all).

The simplest possible encoding of \cref{dfn:proper} in Haskell is just

\begin{Shaded}
\begin{Highlighting}[]
\KeywordTok{class} \DataTypeTok{ConjoinedMonads}\NormalTok{ m }\KeywordTok{where}
\OtherTok{  return ::}\NormalTok{ a }\OtherTok{->}\NormalTok{ m e a}
\OtherTok{  (>>=)  ::}\NormalTok{ m e a }\OtherTok{->}\NormalTok{ (a }\OtherTok{->}\NormalTok{ m e b) }\OtherTok{->}\NormalTok{ m e b}

\OtherTok{  throw  ::}\NormalTok{ e }\OtherTok{->}\NormalTok{ m e a}
\OtherTok{  catch  ::}\NormalTok{ m e a }\OtherTok{->}\NormalTok{ (e }\OtherTok{->}\NormalTok{ m f a) }\OtherTok{->}\NormalTok{ m f a}
\end{Highlighting}
\end{Shaded}

\noindent but it does not play too well with the rest of the Haskell
ecosystem. In the ideal world, \cref{dfn:proper} would get encoded with
the following pseudo-Haskell definition

\begin{definition}

\label{dfn:proper-haskell} \textbf{\textbf{Proper pseudo-Haskell
definition.}}

\begin{Shaded}
\begin{Highlighting}[]
\KeywordTok{class}\NormalTok{ (}\KeywordTok{forall}\NormalTok{ a }\FunctionTok{.} \DataTypeTok{Monad}\NormalTok{ (\textbackslash{}e }\OtherTok{->}\NormalTok{ m e a)) }\CommentTok{-- `Monad` in `e`}
\NormalTok{     , }\KeywordTok{forall}\NormalTok{ e }\FunctionTok{.} \DataTypeTok{Monad}\NormalTok{ (m e) }\CommentTok{-- `Monad` in `a`}
    \OtherTok{=>} \DataTypeTok{ConjoinedMonads}\NormalTok{ m }\KeywordTok{where}
    \CommentTok{-- and that's it}
\end{Highlighting}
\end{Shaded}

\end{definition}

\noindent however, Haskell allows neither rank 2 types in type classes,
nor lambdas in types, which brings us to the following "theorem".

\begin{quasitheorem}

\label{thm:not-in-haskell}

Haskell cannot properly (equivalently to \cref{dfn:proper-haskell})
define \VERB|\DataTypeTok{ConjoinedMonads}|.

\end{quasitheorem}

\begin{proof}

Proper definition of \VERB|\DataTypeTok{ConjoinedMonads}| requires rank
2 types in type class declaration, which is not possible in modern
Haskell. There is no way to emulate rank 2 definition using only rank 1
constructions.

\end{proof}

We call it a "theorem" because we do not really know if its proof really
works out for Haskell as Haskell has an awful lot of language extensions
(including future ones) and there might be some nontrivial combination
of those that gives the desired effect. In particular, GHC version 8.6
released just before this article was finished introduced
\texttt{QuantifiedConstraints} extension~\cite{Bottu:2017:QCC} allowing
us to write

\begin{Shaded}
\begin{Highlighting}[]
\KeywordTok{data} \DataTypeTok{Swap}\NormalTok{ r a e }\FunctionTok{=} \DataTypeTok{Swap}\NormalTok{ \{}\OtherTok{ unSwap ::}\NormalTok{ r e a \}}

\KeywordTok{instance}\NormalTok{ (}\KeywordTok{forall}\NormalTok{ e }\FunctionTok{.} \DataTypeTok{Monad}\NormalTok{ (r e)}
\NormalTok{        , }\KeywordTok{forall}\NormalTok{ a }\FunctionTok{.} \DataTypeTok{Monad}\NormalTok{ (}\DataTypeTok{Swap}\NormalTok{ r a))}
      \OtherTok{=>} \DataTypeTok{ConjoinedMonads}\NormalTok{ r }\KeywordTok{where}
  \CommentTok{-- ...}
\end{Highlighting}
\end{Shaded}

\noindent which, arguably, can be considered good enough, though not
very convenient in practice.

The purposes of this article, however, is not to demonstrate that there
is a convenient form of \cref{dfn:proper} in Haskell but to show what
could be achieved if there were such a convenient definition. Which
means that we can and, hence, shall completely ignore the question of
the most elegant Haskell representation for \cref{dfn:proper} and just
use the very first definition of \VERB|\DataTypeTok{ConjoinedMonads}|
from above for simplicity.

As to the naming, it is, indeed, tempting to call this structure
\VERB|\DataTypeTok{BiMonad}|, but that name is already taken by another
structure from category theory. Then, since the structure consists of
two \VERB|\DataTypeTok{Monads}| that are "dual" to each other via
interaction laws it is tempting to call it
\VERB|\DataTypeTok{DualMonad}| as a double-pun, but that "duality" is
different from the usual duality of category theory. Which is why we
opted into using the name "\VERB|\DataTypeTok{ConjoinedMonads}|" (in the
sense of "conjoined twins", conjoined with left-zeroes).

\hypertarget{instances-constant-functors}{%
\section{Instances: constant
Functors}\label{instances-constant-functors}}

\label{sec:instances:constant}

In this section we discuss the relationship between
\VERB|\DataTypeTok{ConjoinedMonads}| (and \cref{dfn:proper}) and
\VERB|\DataTypeTok{MonadThrow}|, \VERB|\DataTypeTok{MonadCatch}|, and
\VERB|\DataTypeTok{MonadError}| from \cref{sec:monadic-generalizations}.

\hypertarget{monaderror-1}{%
\subsection{MonadError}\label{monaderror-1}}

\label{sec:instances:constant:monaderror}

\VERB|\DataTypeTok{MonadError}| (\cref{sec:monad-error}) relationship to
\VERB|\DataTypeTok{ConjoinedMonads}| turns out to be pretty simple.
Remember that \VERB|\DataTypeTok{MonadError}| is defined using
functional dependencies

\begin{Shaded}
\begin{Highlighting}[]
\KeywordTok{class}\NormalTok{ (}\DataTypeTok{Monad}\NormalTok{ m) }\OtherTok{=>} \DataTypeTok{MonadError}\NormalTok{ e m}
                 \FunctionTok{|}\NormalTok{ m }\OtherTok{->}\NormalTok{ e }\KeywordTok{where}
\end{Highlighting}
\end{Shaded}

This means that Haskell type system guarantees that for each
\VERB|\NormalTok{m}| there exist unique \VERB|\NormalTok{e}| if
\VERB|\DataTypeTok{MonadError}\NormalTok{ e m}| is inhabited. This, in
turn, means that substituting a constant \VERB|\DataTypeTok{Functor}|
\VERB|\NormalTok{r }\FunctionTok{=}\NormalTok{ \textbackslash{}x a }\OtherTok{->}\NormalTok{ m a}|
over \VERB|\DataTypeTok{Monad}| \VERB|\NormalTok{m}| into the definition
of \VERB|\DataTypeTok{ConjoinedMonads}| produces

\begin{Shaded}
\begin{Highlighting}[]
\KeywordTok{class} \DataTypeTok{ConjoinedMonads}\NormalTok{ (\textbackslash{}x a }\OtherTok{->}\NormalTok{ m a) }\KeywordTok{where}
\OtherTok{  return ::}\NormalTok{ a }\OtherTok{->}\NormalTok{ m a}
\OtherTok{  (>>=)  ::}\NormalTok{ m a }\OtherTok{->}\NormalTok{ (a }\OtherTok{->}\NormalTok{ m b) }\OtherTok{->}\NormalTok{ m b}

\OtherTok{  throw  ::}\NormalTok{ e }\OtherTok{->}\NormalTok{ m a}
\OtherTok{  catch  ::}\NormalTok{ m a }\OtherTok{->}\NormalTok{ (e }\OtherTok{->}\NormalTok{ m a) }\OtherTok{->}\NormalTok{ m a}
\end{Highlighting}
\end{Shaded}

The first two operators are just the definition of
\VERB|\DataTypeTok{Monad}\NormalTok{ m}|, the latter two match
\VERB|\DataTypeTok{MonadError}|'s \VERB|\NormalTok{throwError}| and
\VERB|\NormalTok{catchError}| exactly.

\begin{theorem}

\VERB|\DataTypeTok{MonadError}| is a
\VERB|\DataTypeTok{ConjoinedMonads}| that is constant in its first
index.

\end{theorem}

\begin{proof}

By the above argument.

\end{proof}

\hypertarget{monadthrow-and-monadcatch-1}{%
\subsection{MonadThrow and
MonadCatch}\label{monadthrow-and-monadcatch-1}}

\label{sec:instances:constant:monadcatch}

For \VERB|\DataTypeTok{MonadThrow}| and \VERB|\DataTypeTok{MonadCatch}|
(\cref{sec:monad-catch}) it is not the case that \VERB|\NormalTok{e}| is
unique, since \VERB|\DataTypeTok{Exception}\NormalTok{ e}| is a whole
class of types. Moreover, operator \VERB|\NormalTok{catchM}| of
\VERB|\DataTypeTok{MonadCatch}|, unlike \VERB|\NormalTok{catchError}| of
\VERB|\DataTypeTok{MonadError}|, does dynamic dispatch by
\VERB|\NormalTok{cast}|ing \VERB|\DataTypeTok{Exception}|s to the type
of its handler's argument and propagating errors when the
\VERB|\NormalTok{cast}| fails. Note that, strictly speaking, purely from
type perspective \VERB|\DataTypeTok{MonadCatch}| is not \emph{required}
but \emph{allowed} to \VERB|\NormalTok{cast}|, but all the instances do
actually \VERB|\NormalTok{cast}|. The latter fact means that we can
distill that common computational pattern by redefining those structures
using the technique used by imprecise exceptions of \cref{sec:imprecise}
as follows

\begin{Shaded}
\begin{Highlighting}[]
\KeywordTok{class} \DataTypeTok{Monad}\NormalTok{ m }\OtherTok{=>} \DataTypeTok{MonadThrowS}\NormalTok{ m }\KeywordTok{where}
\OtherTok{  throwS ::} \DataTypeTok{SomeException} \OtherTok{->}\NormalTok{ m a}

\KeywordTok{class} \DataTypeTok{MonadThrow}\NormalTok{ m }\OtherTok{=>} \DataTypeTok{MonadCatchS}\NormalTok{ m }\KeywordTok{where}
\OtherTok{  catchS ::}\NormalTok{ m a}
         \OtherTok{->}\NormalTok{ (}\DataTypeTok{SomeException} \OtherTok{->}\NormalTok{ m a) }\OtherTok{->}\NormalTok{ m a}

\OtherTok{throwM' ::}\NormalTok{ (}\DataTypeTok{MonadThrowS}\NormalTok{ m, }\DataTypeTok{Exception}\NormalTok{ e)}
        \OtherTok{=>}\NormalTok{ e }\OtherTok{->}\NormalTok{ m a}
\NormalTok{throwM' }\FunctionTok{=}\NormalTok{ throwS }\FunctionTok{.}\NormalTok{ toException}

\NormalTok{handleOrAgain h e }\FunctionTok{=} \KeywordTok{case}\NormalTok{ fromException e }\KeywordTok{of}
  \DataTypeTok{Just}\NormalTok{ f }\OtherTok{->}\NormalTok{ h f}
  \DataTypeTok{Nothing} \OtherTok{->}\NormalTok{ throwM e}

\OtherTok{catchM' ::}\NormalTok{ (}\DataTypeTok{MonadCatchS}\NormalTok{ m, }\DataTypeTok{Exception}\NormalTok{ e)}
        \OtherTok{=>}\NormalTok{ m a }\OtherTok{->}\NormalTok{ (e }\OtherTok{->}\NormalTok{ m a) }\OtherTok{->}\NormalTok{ m a}
\NormalTok{catchM' ma }\FunctionTok{=}\NormalTok{ catchS ma }\FunctionTok{.}\NormalTok{ handleOrAgain}
\end{Highlighting}
\end{Shaded}

Note that \VERB|\DataTypeTok{MonadCatchS}| is, again, a constant
\VERB|\DataTypeTok{ConjoinedMonads}| with error index fixed to
\VERB|\DataTypeTok{SomeException}|. Also note that
\VERB|\NormalTok{throwM'}| above is the only way to get an equivalent
for \VERB|\NormalTok{throwM}| because \VERB|\NormalTok{toException}| is
the only way to cast an arbitrary type to
\VERB|\DataTypeTok{SomeException}|. On the other hand,
\VERB|\NormalTok{catchM}| from \VERB|\DataTypeTok{MonadCatch}|, unlike
\VERB|\NormalTok{catchM'}| above, allows for instances that can cheat.
For example, \VERB|\NormalTok{catchM}| can give a constant
\VERB|\DataTypeTok{SomeException}| to the handler every time instead of
\VERB|\NormalTok{cast}|ing anything. We feel that this implies that
\VERB|\DataTypeTok{MonadCatch}| is not a proper formal structure for
error handling.

\begin{definition}

\label{dfn:proper-monad-catch}

\textbf{\textbf{Proper \VERB|\DataTypeTok{MonadCatch}| instance.}} We
shall call an instance of \VERB|\DataTypeTok{MonadCatch}| \emph{proper}
when its \VERB|\NormalTok{catchM}| can be decomposed into
\VERB|\NormalTok{catchS}| and \VERB|\NormalTok{handleOrAgain}|.

\end{definition}

\begin{theorem}

\label{thm:instances:constant:monadcatch} Every proper instance of
\VERB|\DataTypeTok{MonadCatch}| is a composition of
\VERB|\DataTypeTok{ConjoinedMonads}| that is constant in its error index
with \VERB|\NormalTok{toException}| in \VERB|\NormalTok{throwD}| and
\VERB|\NormalTok{handleOrAgain}| in \VERB|\NormalTok{catchD}|. In
particular, \VERB|\DataTypeTok{MonadThrow}| is a composition of
\VERB|\DataTypeTok{Pointed}| in the error index with
\VERB|\NormalTok{toException}|.

\end{theorem}

\begin{proof}

By the above reasoning.

\end{proof}

\hypertarget{instances-parser-combinators}{%
\section{Instances: parser
combinators}\label{instances-parser-combinators}}

\label{sec:instances:parser-combinators}

In this section we discuss the application of
\VERB|\DataTypeTok{ConjoinedMonads}| and \cref{dfn:proper} to
\VERB|\DataTypeTok{Monad}|ic parser combinators discussed in
\cref{sec:parser-combinators}.

\hypertarget{inevitable-definitions}{%
\subsection{Inevitable definitions}\label{inevitable-definitions}}

To start off, let us continue using the definition of
\VERB|\DataTypeTok{Parser}| type from
\cref{sec:parser-combinators:with-access}. The
\VERB|\DataTypeTok{Monad}| instance in index \VERB|\NormalTok{e}| for
this type is similarly easy to implement (by just trying all free
functions of appropriate types) and it, too, has two possible
implementations

\begin{Shaded}
\begin{Highlighting}[]
\OtherTok{throwP ::}\NormalTok{ e }\OtherTok{->} \DataTypeTok{Parser}\NormalTok{ s e a}
\NormalTok{throwP e }\FunctionTok{=} \DataTypeTok{Parser} \FunctionTok{$}\NormalTok{ \textbackslash{}s }\OtherTok{->} \DataTypeTok{Left}\NormalTok{ (e, s)}

\OtherTok{catchP ::} \DataTypeTok{Parser}\NormalTok{ s e a }\OtherTok{->}\NormalTok{ (e }\OtherTok{->} \DataTypeTok{Parser}\NormalTok{ s f a) }\OtherTok{->} \DataTypeTok{Parser}\NormalTok{ s f a}
\NormalTok{catchP p f }\FunctionTok{=} \DataTypeTok{Parser} \FunctionTok{$}\NormalTok{ \textbackslash{}s }\OtherTok{->}
  \KeywordTok{case}\NormalTok{ runParser p s }\KeywordTok{of}
    \DataTypeTok{Right}\NormalTok{ x }\OtherTok{->} \DataTypeTok{Right}\NormalTok{ x}
    \DataTypeTok{Left}\NormalTok{ (e, _) }\OtherTok{->}\NormalTok{ runParser (f e) s}

\OtherTok{catchP' ::} \DataTypeTok{Parser}\NormalTok{ s e a }\OtherTok{->}\NormalTok{ (e }\OtherTok{->} \DataTypeTok{Parser}\NormalTok{ s f a) }\OtherTok{->} \DataTypeTok{Parser}\NormalTok{ s f a}
\NormalTok{catchP' p f }\FunctionTok{=} \DataTypeTok{Parser} \FunctionTok{$}\NormalTok{ \textbackslash{}s }\OtherTok{->}
  \KeywordTok{case}\NormalTok{ runParser p s }\KeywordTok{of}
    \DataTypeTok{Right}\NormalTok{ x }\OtherTok{->} \DataTypeTok{Right}\NormalTok{ x}
    \DataTypeTok{Left}\NormalTok{ (e, s') }\OtherTok{->}\NormalTok{ runParser (f e) s'}
\end{Highlighting}
\end{Shaded}

\noindent with \VERB|\NormalTok{catchP}| doing backtracking on failures
and \VERB|\NormalTok{catchP'}| proceeding to handling with the current
state.

\begin{theorem}

\label{thm:instances:parser-combinators} \VERB|\DataTypeTok{Parser}| is
a \VERB|\DataTypeTok{ConjoinedMonads}| for both versions of
\VERB|\NormalTok{catchP}|.

\end{theorem}

\begin{proof}

\VERB|\DataTypeTok{Monad}| laws for \VERB|\NormalTok{catchP'}| follow
from the corresponding laws for
\VERB|\NormalTok{(}\FunctionTok{>>=}\NormalTok{)}| of
\cref{sec:parser-combinators:with-access}.

The rest can be proven by trivial case analysis and/or by using the
observation from the proof of \cref{thm:with-heuristic}.

\end{proof}

A curious consequence of the above theorem is that
\VERB|\NormalTok{(}\FunctionTok{>>=}\NormalTok{)}| of
\cref{sec:parser-combinators:with-access} also has a roll-back version
which satisfies \VERB|\DataTypeTok{Monad}| laws

\begin{Shaded}
\begin{Highlighting}[]
\NormalTok{bindP p f }\FunctionTok{=} \DataTypeTok{Parser} \FunctionTok{$}\NormalTok{ \textbackslash{}s }\OtherTok{->}
  \KeywordTok{case}\NormalTok{ runParser p s }\KeywordTok{of}
    \DataTypeTok{Left}\NormalTok{ x }\OtherTok{->} \DataTypeTok{Left}\NormalTok{ x}
    \DataTypeTok{Right}\NormalTok{ (a, _) }\OtherTok{->}\NormalTok{ runParser (f a) s}
\end{Highlighting}
\end{Shaded}

Though, of course, a \VERB|\DataTypeTok{Parser}| that would use
\VERB|\NormalTok{bindP}| in place of the usual
\VERB|\NormalTok{(}\FunctionTok{>>=}\NormalTok{)}| could not be called a
"parser" anymore.

\hypertarget{the-interesting-parts}{%
\subsection{The interesting parts}\label{the-interesting-parts}}

The first interesting fact is that
\VERB|\NormalTok{(}\FunctionTok{<\VerbBar{}>}\NormalTok{)}| operator of
the \VERB|\DataTypeTok{Alternative}| (\cref{sec:alternative}) type class
is simply a type restricted version of \VERB|\NormalTok{orElseP}| which,
in turn, is just \VERB|\NormalTok{(}\FunctionTok{>>}\NormalTok{)}|
operator for the \VERB|\DataTypeTok{Monad}| in index
\VERB|\NormalTok{e}|

\begin{Shaded}
\begin{Highlighting}[]
\OtherTok{orElseP ::} \DataTypeTok{Parser}\NormalTok{ s e a }\OtherTok{->} \DataTypeTok{Parser}\NormalTok{ s f a }\OtherTok{->} \DataTypeTok{Parser}\NormalTok{ s f a}
\NormalTok{orElseP f g }\FunctionTok{=}\NormalTok{ f }\OtherTok{`catchP`} \FunctionTok{const}\NormalTok{ g}
\end{Highlighting}
\end{Shaded}

\begin{Shaded}
\begin{Highlighting}[]
\KeywordTok{instance} \DataTypeTok{Monoid}\NormalTok{ e }\OtherTok{=>} \DataTypeTok{Alternative}\NormalTok{ (}\DataTypeTok{Parser}\NormalTok{ s e) }\KeywordTok{where}
\NormalTok{  empty }\FunctionTok{=} \DataTypeTok{Parser} \FunctionTok{$}\NormalTok{ \textbackslash{}s }\OtherTok{->} \DataTypeTok{Left}\NormalTok{ (}\FunctionTok{mempty}\NormalTok{, s)}
\NormalTok{  f }\FunctionTok{<|>}\NormalTok{ g }\FunctionTok{=}\NormalTok{ f }\OtherTok{`orElseP`}\NormalTok{ g}
\end{Highlighting}
\end{Shaded}

Of even more interest is the fact that substituting
\VERB|\NormalTok{orElseP}| instead of
\VERB|\NormalTok{(}\FunctionTok{<\VerbBar{}>}\NormalTok{)}| into the
definition of \VERB|\NormalTok{many}| operator produces
\VERB|\NormalTok{many}| and \VERB|\NormalTok{some}| operators with types
that show that \VERB|\NormalTok{some}| inherits error produced by its
argument while \VERB|\NormalTok{many}| ignores them

\begin{Shaded}
\begin{Highlighting}[]
\OtherTok{someP ::} \DataTypeTok{Parser}\NormalTok{ s e a }\OtherTok{->} \DataTypeTok{Parser}\NormalTok{ s e [a]}
\NormalTok{someP v }\FunctionTok{=} \FunctionTok{fmap}\NormalTok{ (}\FunctionTok{:}\NormalTok{) v }\FunctionTok{<*>}\NormalTok{ manyP v}

\OtherTok{manyP ::} \DataTypeTok{Parser}\NormalTok{ s e a }\OtherTok{->} \DataTypeTok{Parser}\NormalTok{ s f [a]}
\NormalTok{manyP v }\FunctionTok{=}\NormalTok{ someP v }\OtherTok{`orElseP`} \FunctionTok{pure}\NormalTok{ []}
\end{Highlighting}
\end{Shaded}

This method of substituting
\VERB|\NormalTok{(}\FunctionTok{<\VerbBar{}>}\NormalTok{)}| with
\VERB|\NormalTok{orElseP}| extends to other similar combinators like
\VERB|\NormalTok{choice}|, \VERB|\NormalTok{optional}|,
\VERB|\NormalTok{notFollowedBy}| of all three aforementioned parser
combinator libraries (Parser, Attoparsec, Megaparsec) and similar
structures. The overall effect of this substitution is very useful in
practice: it produces generic parser combinators that can be used to
express parsers that are precise about errors they raise and handle. We
can not emphasize this fact enough.

All of the above results of this section trivially generalize to their
\VERB|\DataTypeTok{MonadTrans}| versions as usual.

\hypertarget{instances-conventional-throw-and-catch-via-callcc}{%
\section{Instances: conventional throw and catch via
callCC}\label{instances-conventional-throw-and-catch-via-callcc}}

\label{sec:instances:throw-catch-cc}

It is well-known fact that Emacs LISP-style \VERB|\NormalTok{throw}| and
\VERB|\FunctionTok{catch}| can be emulated with Scheme's
\VERB|\KeywordTok{call/cc}| and some mutable
variables~\cite{CSE341:2004:Scheme:Continuations,
WikiBooks:Scheme:Continuations}. As a Haskell instance, Neil Mitchel
used the same technique translated to Haskell's
\VERB|\NormalTok{IORef}|s and \VERB|\NormalTok{callCC}| in for Shake
build system~\cite{Mitchell:2014:CE,
Mitchell:GitHub:Shake} (however, at the time of writing Shake no longer
uses that code). In this section we shall demonstrate that a structure
with the same semantics can be implemented in pure Haskell without the
use of mutable variables. In all the cases, as usual, C++/Java-style
dynamic dispatch can be added on top using the same
\VERB|\NormalTok{cast}|ing technique of sections~\ref{sec:imprecise}
and~\ref{sec:instances:constant:monadcatch}. Hence without the loss of
generality in this section we shall discuss only the most-recent-handler
case.

\hypertarget{second-rank-callcc}{%
\subsection{Second-rank callCC}\label{second-rank-callcc}}

Remember the definition of \VERB|\NormalTok{callCC}| from
\cref{sec:callcc}. The underappreciated fact about that function is that
its type is not its most general type for its term. Note that variable
\(b\) in Peirce's law

\[((a \to b) \to a) \to a\]

\noindent plays the same role as \VERB|\NormalTok{r}| plays in the
definition of \VERB|\DataTypeTok{Cont}|: it is a generalization of the
bottom \(\bot\) constant. This, of course, means that we can generalize
Peirce's law to

\[((\forall b . a \to b) \to a) \to a\]

\noindent and, by repeating the derivation in \cref{sec:callcc}, give
the following second-rank type for \VERB|\NormalTok{callCC}|

\begin{Shaded}
\begin{Highlighting}[]
\OtherTok{callCCR2 ::}\NormalTok{ ((}\KeywordTok{forall}\NormalTok{ b }\FunctionTok{.}\NormalTok{ a }\OtherTok{->} \DataTypeTok{Cont}\NormalTok{ r b) }\OtherTok{->} \DataTypeTok{Cont}\NormalTok{ r a) }\OtherTok{->} \DataTypeTok{Cont}\NormalTok{ r a}
\end{Highlighting}
\end{Shaded}

\noindent while keeping exactly the same implementation.

\hypertarget{throwt-monadtransformer}{%
\subsection{ThrowT MonadTransformer}\label{throwt-monadtransformer}}

Note that, in essence, \VERB|\FunctionTok{catch}| maintains a stack of
handler addresses and \VERB|\NormalTok{throw}| simply \texttt{jmp}s to
the most recent one. Emulation of exceptions with
\VERB|\KeywordTok{call/cc}| works
similarly~\cite{CSE341:2004:Scheme:Continuations,
WikiBooks:Scheme:Continuations}. The main never explicitly stated
observation in that translation is that the type of the handler in the
type of

\begin{Shaded}
\begin{Highlighting}[]
\FunctionTok{catch}\OtherTok{ ::} \DataTypeTok{M} \OtherTok{->}\NormalTok{ (e }\OtherTok{->} \DataTypeTok{M}\NormalTok{) }\OtherTok{->} \DataTypeTok{M}
\end{Highlighting}
\end{Shaded}

\noindent matches the type of
\VERB|\OtherTok{throw ::}\NormalTok{ e }\OtherTok{->} \DataTypeTok{M}|
and the type of escape continuation when \VERB|\DataTypeTok{M}| is
\VERB|\DataTypeTok{ContT}\NormalTok{ r m b}|. In other words, we can
simply assign

\begin{Shaded}
\begin{Highlighting}[]
\KeywordTok{type} \DataTypeTok{Handler}\NormalTok{ r e m }\FunctionTok{=} \KeywordTok{forall}\NormalTok{ b }\FunctionTok{.}\NormalTok{ e }\OtherTok{->} \DataTypeTok{ContT}\NormalTok{ r m b}
\end{Highlighting}
\end{Shaded}

\noindent to be to type of our handler and since
\VERB|\NormalTok{callCC}| provides an escape continuation directly to
its argument \VERB|\FunctionTok{catch}| can simply save it and
\VERB|\NormalTok{throw}| can simply take the most recent one and escape
into it

\begin{Shaded}
\begin{Highlighting}[]
\OtherTok{throwT ::}\NormalTok{ e }\OtherTok{->} \DataTypeTok{ThrowT}\NormalTok{ r m e a}
\NormalTok{throwT e }\FunctionTok{=} \DataTypeTok{ThrowT} \FunctionTok{$}\NormalTok{ \textbackslash{}currentThrow }\OtherTok{->}\NormalTok{ currentThrow e}
\end{Highlighting}
\end{Shaded}

Also note that since the stack \VERB|\FunctionTok{catch}| maintains
stays immutable between \VERB|\FunctionTok{catch}|es and each state of
the stack is bound to the computation argument of
\VERB|\FunctionTok{catch}|, in principle, we should be able to use a
simple context (pure function, \VERB|\DataTypeTok{Reader}|) instead of a
mutable variable as follows

\begin{Shaded}
\begin{Highlighting}[]
\KeywordTok{type} \DataTypeTok{ThrowT}\NormalTok{ r m e a }\FunctionTok{=}
  \DataTypeTok{ReaderT}\NormalTok{ (}\DataTypeTok{Handler}\NormalTok{ r e m) }\CommentTok{-- for saving last handler}
\NormalTok{          (}\DataTypeTok{ContT}\NormalTok{ r m)     }\CommentTok{-- for callCC}
\NormalTok{          a}
\end{Highlighting}
\end{Shaded}

\noindent which, after inlining all the definitions except pure
\VERB|\DataTypeTok{Cont}| becomes

\begin{Shaded}
\begin{Highlighting}[]
\KeywordTok{newtype} \DataTypeTok{ThrowT}\NormalTok{ r m e a }\FunctionTok{=} \DataTypeTok{ThrowT}
\NormalTok{  \{}\OtherTok{ runThrowT ::}\NormalTok{ (}\KeywordTok{forall}\NormalTok{ b }\FunctionTok{.}\NormalTok{ e }\OtherTok{->} \DataTypeTok{Cont}\NormalTok{ (m r) b)}
              \OtherTok{->} \DataTypeTok{Cont}\NormalTok{ (m r) a \}}
\end{Highlighting}
\end{Shaded}

Finally, since the escape continuation of delimited
\VERB|\NormalTok{callCC}| escapes to the same address where the body of
\VERB|\NormalTok{callCC}| normally returns, to emulate a single
\VERB|\FunctionTok{catch}| we need to chain two
\VERB|\NormalTok{callCC}|s as follows

\begin{Shaded}
\begin{Highlighting}[]
\OtherTok{catchT ::} \DataTypeTok{ThrowT}\NormalTok{ r m e a}
       \OtherTok{->}\NormalTok{ (e }\OtherTok{->} \DataTypeTok{ThrowT}\NormalTok{ r m f a)}
       \OtherTok{->} \DataTypeTok{ThrowT}\NormalTok{ r m f a}
\NormalTok{catchT m h }\FunctionTok{=} \DataTypeTok{ThrowT} \FunctionTok{$}\NormalTok{ \textbackslash{}outerThrow }\OtherTok{->}
\NormalTok{  callCC }\FunctionTok{$}\NormalTok{ \textbackslash{}normalExit }\OtherTok{->} \KeywordTok{do}
\NormalTok{    e }\OtherTok{<-}\NormalTok{ callCCR2 }\FunctionTok{$}\NormalTok{ \textbackslash{}newThrow }\OtherTok{->}\NormalTok{ runThrowT m newThrow }\FunctionTok{>>=}\NormalTok{ normalExit}
    \CommentTok{-- newThrow escapes here}
\NormalTok{    runThrowT (h e) outerThrow}
  \CommentTok{-- normalExit escapes here}
\end{Highlighting}
\end{Shaded}

Note that this expression requires our second-rank
\VERB|\NormalTok{callCCR2}| since our \VERB|\DataTypeTok{Handler}| is
universally quantified by the variable \VERB|\NormalTok{b}|. However, if
we fix \VERB|\NormalTok{e}| to a constant type then the conventional
\VERB|\NormalTok{callCC}| will suffice.

Similarly to other uses of generalized Kolmogorov's translation we, too,
can hide \VERB|\NormalTok{r}| parameter behind
\VERB|\KeywordTok{forall}|

\newpage

\begin{Shaded}
\begin{Highlighting}[]
\KeywordTok{newtype} \DataTypeTok{ThrowT'}\NormalTok{ m e a }\FunctionTok{=} \DataTypeTok{ThrowT'}
\NormalTok{  \{}\OtherTok{ runThrowT' ::} \KeywordTok{forall}\NormalTok{ r}
                \FunctionTok{.}\NormalTok{ (}\KeywordTok{forall}\NormalTok{ b }\FunctionTok{.}\NormalTok{ e }\OtherTok{->} \DataTypeTok{Cont}\NormalTok{ (m r) b)}
               \OtherTok{->} \DataTypeTok{Cont}\NormalTok{ (m r) a \}}

\OtherTok{throwT' ::}\NormalTok{ e }\OtherTok{->} \DataTypeTok{ThrowT'}\NormalTok{ m e a}
\OtherTok{catchT' ::} \DataTypeTok{ThrowT'}\NormalTok{ m e a}
        \OtherTok{->}\NormalTok{ (e }\OtherTok{->} \DataTypeTok{ThrowT'}\NormalTok{ m f a)}
        \OtherTok{->} \DataTypeTok{ThrowT'}\NormalTok{ m f a}
\end{Highlighting}
\end{Shaded}

\noindent without any changes to the bodies of \VERB|\NormalTok{throw}|
and \VERB|\FunctionTok{catch}|.

\begin{theorem}

For \VERB|\DataTypeTok{Monad}| \VERB|\NormalTok{m}| and any
\VERB|\NormalTok{r}|, \VERB|\DataTypeTok{ThrowT}\NormalTok{ r m}| and
\VERB|\DataTypeTok{ThrowT'}\NormalTok{ m}| are
\VERB|\DataTypeTok{ConjoinedMonads}|s.

\end{theorem}

\begin{proof}

For each index.

\begin{itemize}
\item
  In index \VERB|\NormalTok{a}|: \VERB|\DataTypeTok{ThrowT}| is a
  special case of \VERB|\DataTypeTok{ReaderT}| and
  \VERB|\DataTypeTok{Cont}| and \VERB|\NormalTok{m}| are
  \VERB|\DataTypeTok{Monad}|s.
\item
  In index \VERB|\NormalTok{e}|: by substitution of the above
  definitions into the \VERB|\DataTypeTok{Monad}| laws, since the
  definitions of \VERB|\NormalTok{throwT}| and
  \VERB|\NormalTok{throwT'}| are, essentially, identity functions.
\end{itemize}

\end{proof}

\hypertarget{instances-error-explicit-io}{%
\section{Instances: error-explicit
IO}\label{instances-error-explicit-io}}

\label{sec:instances:eio}

As we saw in \cref{sec:imprecise}, \VERB|\DataTypeTok{IO}| is defined as
a \VERB|\DataTypeTok{State}| \VERB|\DataTypeTok{Monad}| with some
magical primitive operations.\footnote{Some of which actually break
  \VERB|\DataTypeTok{Monad}| laws, but as mentioned in
  \cref{rem:io-caveats} that is out of scope of this discussion.} Which
means there is nothing preventing us from extending that
\VERB|\DataTypeTok{IO}| signature with a type for errors.

\begin{Shaded}
\begin{Highlighting}[]
\KeywordTok{newtype} \DataTypeTok{EIO}\NormalTok{ e a}
\end{Highlighting}
\end{Shaded}

Similarly to parser combinators of
\cref{sec:instances:parser-combinators} there are several possible
implementations of this \VERB|\DataTypeTok{EIO}| (including, in
principle, the ones that do backtracking on errors, though, of course,
that would be inconsistent with the semantics of the
\VERB|\DataTypeTok{RealWorld}|). The simplest one matches a definition
for non-backtracking parser combinator on
\VERB|\DataTypeTok{State}\FunctionTok{#} \DataTypeTok{RealWorld}| from
\cref{sec:parser-combinators:with-access}

\begin{Shaded}
\begin{Highlighting}[]
\KeywordTok{newtype} \DataTypeTok{EIO}\NormalTok{ e a }\FunctionTok{=} \DataTypeTok{EIO}
\NormalTok{  \{}\OtherTok{ runEIO ::} \DataTypeTok{State}\FunctionTok{#} \DataTypeTok{RealWorld}
           \OtherTok{->}\NormalTok{ (}\FunctionTok{#} \DataTypeTok{Either}\NormalTok{ e a, }\DataTypeTok{State}\FunctionTok{#} \DataTypeTok{RealWorld} \FunctionTok{#}\NormalTok{) \}}

\KeywordTok{instance} \DataTypeTok{Pointed}\NormalTok{ (}\DataTypeTok{EIO}\NormalTok{ e) }\KeywordTok{where}
  \FunctionTok{pure}\NormalTok{ a }\FunctionTok{=} \DataTypeTok{EIO} \FunctionTok{$}\NormalTok{ \textbackslash{}s }\OtherTok{->}\NormalTok{ (}\FunctionTok{#} \DataTypeTok{Right}\NormalTok{ a, s }\FunctionTok{#}\NormalTok{)}

\KeywordTok{instance} \DataTypeTok{Monad}\NormalTok{ (}\DataTypeTok{EIO}\NormalTok{ e) }\KeywordTok{where}
\NormalTok{  m }\FunctionTok{>>=}\NormalTok{ f }\FunctionTok{=} \DataTypeTok{EIO} \FunctionTok{$}\NormalTok{ \textbackslash{}s }\OtherTok{->} \KeywordTok{case}\NormalTok{ runEIO m s }\KeywordTok{of}
\NormalTok{    (}\FunctionTok{#} \DataTypeTok{Left}\NormalTok{  a, s' }\FunctionTok{#}\NormalTok{) }\OtherTok{->}\NormalTok{ (}\FunctionTok{#} \DataTypeTok{Left}\NormalTok{ a, s' }\FunctionTok{#}\NormalTok{)}
\NormalTok{    (}\FunctionTok{#} \DataTypeTok{Right}\NormalTok{ a, s' }\FunctionTok{#}\NormalTok{) }\OtherTok{->}\NormalTok{ runEIO (f a) s'}

\CommentTok{-- Note how symmetric this is with Pointed and Monad instances.}
\OtherTok{throwEIO ::}\NormalTok{ e }\OtherTok{->} \DataTypeTok{EIO}\NormalTok{ e a}
\NormalTok{throwEIO e }\FunctionTok{=} \DataTypeTok{EIO} \FunctionTok{$}\NormalTok{ \textbackslash{}s }\OtherTok{->}\NormalTok{ (}\FunctionTok{#} \DataTypeTok{Left}\NormalTok{ e, s }\FunctionTok{#}\NormalTok{)}

\OtherTok{catchEIO ::} \DataTypeTok{EIO}\NormalTok{ e a }\OtherTok{->}\NormalTok{ (e }\OtherTok{->} \DataTypeTok{EIO}\NormalTok{ f a) }\OtherTok{->} \DataTypeTok{EIO}\NormalTok{ f a}
\NormalTok{catchEIO m f }\FunctionTok{=} \DataTypeTok{EIO} \FunctionTok{$}\NormalTok{ \textbackslash{}s }\OtherTok{->} \KeywordTok{case}\NormalTok{ runEIO m s }\KeywordTok{of}
\NormalTok{  (}\FunctionTok{#} \DataTypeTok{Left}\NormalTok{  a, s' }\FunctionTok{#}\NormalTok{) }\OtherTok{->}\NormalTok{ runEIO (f a) s'}
\NormalTok{  (}\FunctionTok{#} \DataTypeTok{Right}\NormalTok{ a, s' }\FunctionTok{#}\NormalTok{) }\OtherTok{->}\NormalTok{ (}\FunctionTok{#} \DataTypeTok{Right}\NormalTok{ a, s' }\FunctionTok{#}\NormalTok{)}
\end{Highlighting}
\end{Shaded}

Note that very similar structures were proposed before
in~\cite{Iborra:2010:ETE} and
\VERB|\DataTypeTok{Control.Monad.Exception.Catch}| module of
\texttt{control-monad-exception}~\cite{Hackage:control-monad-exception0112}
discussed in \cref{sec:other-monadic-generalizations}. Also note that
the definition of GHC's \VERB|\DataTypeTok{IO}| before imprecise
exceptions were introduced was similar to \VERB|\DataTypeTok{EIO}| above
(but without the parameter \VERB|\NormalTok{e}|) and one of the primary
motivations behind introduction of builtin exceptions into GHC mentioned
in \cite{PeytonJones:1999:SIE} was to make \VERB|\DataTypeTok{IO}| more
efficient by allowing its
\VERB|\NormalTok{(}\FunctionTok{>>=}\NormalTok{)}| to be implemented
without pattern-matching. But there are, of course, other ways to
eliminate pattern matching. By moving \VERB|\DataTypeTok{Either}| in the
definition of \VERB|\DataTypeTok{EIO}| out the parentheses using the
technique from \cref{sec:parser-combinators:with-access} and then
Scott-encoding the resulting type we can make the following definition

\begin{Shaded}
\begin{Highlighting}[]
\KeywordTok{newtype} \DataTypeTok{SEIO}\NormalTok{ e a }\FunctionTok{=} \DataTypeTok{SEIO}
\NormalTok{  \{}\OtherTok{ runSEIO ::} \KeywordTok{forall}\NormalTok{ r}
             \FunctionTok{.}\NormalTok{ (e }\OtherTok{->} \DataTypeTok{State}\FunctionTok{#} \DataTypeTok{RealWorld} \OtherTok{->}\NormalTok{ r)}
            \OtherTok{->}\NormalTok{ (a }\OtherTok{->} \DataTypeTok{State}\FunctionTok{#} \DataTypeTok{RealWorld} \OtherTok{->}\NormalTok{ r)}
            \OtherTok{->} \DataTypeTok{State}\FunctionTok{#} \DataTypeTok{RealWorld}
            \OtherTok{->}\NormalTok{ r \}}

\KeywordTok{instance} \DataTypeTok{Pointed}\NormalTok{ (}\DataTypeTok{SEIO}\NormalTok{ e) }\KeywordTok{where}
  \FunctionTok{pure}\NormalTok{ a }\FunctionTok{=} \DataTypeTok{SEIO} \FunctionTok{$}\NormalTok{ \textbackslash{}err ok s }\OtherTok{->}\NormalTok{ ok a s}

\KeywordTok{instance} \DataTypeTok{Monad}\NormalTok{ (}\DataTypeTok{SEIO}\NormalTok{ e) }\KeywordTok{where}
\NormalTok{  m }\FunctionTok{>>=}\NormalTok{ f }\FunctionTok{=} \DataTypeTok{SEIO} \FunctionTok{$}\NormalTok{ \textbackslash{}err ok s }\OtherTok{->}\NormalTok{ runSEIO m err (\textbackslash{}a }\OtherTok{->}\NormalTok{ runSEIO (f a) err ok) s}

\CommentTok{-- Note the same here.}
\OtherTok{throwSEIO ::}\NormalTok{ e }\OtherTok{->} \DataTypeTok{SEIO}\NormalTok{ e a}
\NormalTok{throwSEIO e }\FunctionTok{=} \DataTypeTok{SEIO} \FunctionTok{$}\NormalTok{ \textbackslash{}err ok s }\OtherTok{->}\NormalTok{ err e s}

\OtherTok{catchSEIO ::} \DataTypeTok{SEIO}\NormalTok{ e a }\OtherTok{->}\NormalTok{ (e }\OtherTok{->} \DataTypeTok{SEIO}\NormalTok{ f a) }\OtherTok{->} \DataTypeTok{SEIO}\NormalTok{ f a}
\NormalTok{catchSEIO m f }\FunctionTok{=} \DataTypeTok{SEIO} \FunctionTok{$}\NormalTok{ \textbackslash{}err ok s }\OtherTok{->}\NormalTok{ runSEIO m (\textbackslash{}e }\OtherTok{->}\NormalTok{ runSEIO (f e) err ok) ok s}
\end{Highlighting}
\end{Shaded}

\begin{theorem}

\label{thm:instances:eio} Both \VERB|\DataTypeTok{EIO}| and
\VERB|\DataTypeTok{SEIO}| with the above operations are
\VERB|\DataTypeTok{ConjoinedMonads}|s.

\end{theorem}

\begin{proof}

Consequence of \cref{thm:instances:parser-combinators} and the fact that
Scott-encoding preserves computational properties.

\end{proof}

\hypertarget{instances-conventional-io}{%
\section{Instances: conventional IO}\label{instances-conventional-io}}

\label{sec:instances:io}

\begin{theorem}

\label{thm:instances:io} \VERB|\DataTypeTok{IO}| is a composition of
\VERB|\DataTypeTok{ConjoinedMonads}| that is constant in its error index
with \VERB|\NormalTok{toException}| in
\VERB|\NormalTok{raiseIO}\FunctionTok{#}| and
\VERB|\NormalTok{handleOrAgain}| in
\VERB|\NormalTok{catch}\FunctionTok{#}|.

\end{theorem}

\begin{proof}

A consequence of of results of
theorems~\ref{thm:instances:constant:monadcatch}
and~\ref{thm:instances:eio} for
\VERB|\NormalTok{e }\FunctionTok{==} \DataTypeTok{SomeException}|.

\end{proof}

Note that, according to \cref{rem:io-caveats}, the above works out only
because
\VERB|\NormalTok{raiseIO}\FunctionTok{#}|/\VERB|\NormalTok{throwIO}|,
unlike \VERB|\NormalTok{raise}\FunctionTok{#}|/\VERB|\NormalTok{throw}|,
are deterministic (see \cref{sec:imprecise}).

Also note that in a dialect of Haskell with separate operators for
imprecise exceptions (or without imprecise exceptions altogether) we can
completely replace \VERB|\DataTypeTok{IO}| with \VERB|\DataTypeTok{EIO}|
as defined above. We can not, however, apply that construction to GHC's
Haskell dialect since it merges precise and imprecise
\VERB|\FunctionTok{catch}| (see \cref{rem:io-two-kinds-of-exceptions}).

\hypertarget{applicatives}{%
\section{Applicatives}\label{applicatives}}

\label{sec:applicatives}

Now let us once more turn our attention to the bodies of
definitions~\ref{dfn:proper}, \ref{dfn:fishy},
and~\ref{dfn:proper-haskell} (all of which define the same structure).

\begin{Shaded}
\begin{Highlighting}[]
\KeywordTok{class}\NormalTok{ (}\KeywordTok{forall}\NormalTok{ a }\FunctionTok{.} \DataTypeTok{Monad}\NormalTok{ (\textbackslash{}e }\OtherTok{->}\NormalTok{ m e a))}
\NormalTok{     , }\KeywordTok{forall}\NormalTok{ e }\FunctionTok{.} \DataTypeTok{Monad}\NormalTok{ (m e)}
    \OtherTok{=>} \DataTypeTok{ConjoinedMonads}\NormalTok{ m }\KeywordTok{where}
\end{Highlighting}
\end{Shaded}

Since \VERB|\DataTypeTok{ConjoinedMonads}| is simply a
\VERB|\DataTypeTok{Monad}|~$\times$~\VERB|\DataTypeTok{Monad}| with
interaction laws between \VERB|\FunctionTok{pure}| and
\VERB|\NormalTok{bind}| operators (\cref{dfn:proper}) it is natural to
ask what would happen if we replace one or both of those
\VERB|\DataTypeTok{Monad}|s with more general structures like
\VERB|\DataTypeTok{Applicative}| and modify the interaction laws
accordingly.

The two structures with \VERB|\DataTypeTok{Applicative}| in index
\VERB|\NormalTok{e}| seem to be unusable for the purposes of this
article since they lack conventional error handling operators. However,
the structure with \VERB|\DataTypeTok{Monad}| in index
\VERB|\NormalTok{e}| and \VERB|\DataTypeTok{Applicative}| in index
\VERB|\NormalTok{a}| looks interesting.

\begin{Shaded}
\begin{Highlighting}[]
\KeywordTok{class}\NormalTok{ (}\KeywordTok{forall}\NormalTok{ a }\FunctionTok{.} \DataTypeTok{Monad}\NormalTok{ (\textbackslash{}e }\OtherTok{->}\NormalTok{ m e a))}
\NormalTok{     , }\KeywordTok{forall}\NormalTok{ e }\FunctionTok{.} \DataTypeTok{Applicative}\NormalTok{ (m e)}
    \OtherTok{=>} \DataTypeTok{MonadXApplicative}\NormalTok{ m }\KeywordTok{where}
\end{Highlighting}
\end{Shaded}

In this structure the \VERB|\DataTypeTok{Monad}|ic index gives
conventional \VERB|\NormalTok{throw}| and \VERB|\FunctionTok{catch}|
operators, and the \VERB|\DataTypeTok{Applicative}| index can be treated
as expressing generalized function application (see
\cref{sec:applicative-functor}) for structure \VERB|\NormalTok{m}|.

In other words, when the above structure preserves errors and pure
values similarly to \cref{dfn:proper}

\begin{Shaded}
\begin{Highlighting}[]
\NormalTok{throw e }\FunctionTok{<*>}\NormalTok{ a }\FunctionTok{==}\NormalTok{ throw e}
\FunctionTok{pure}\NormalTok{ a }\OtherTok{`catch`}\NormalTok{ f }\FunctionTok{==} \FunctionTok{pure}\NormalTok{ a}
\end{Highlighting}
\end{Shaded}

\noindent (and obeys the laws of \VERB|\DataTypeTok{Applicative}| and
\VERB|\DataTypeTok{Monad}| for corresponding operators) then it can be
used to express \(\lambda\)-calculus with exceptions by simply injecting
all \VERB|\FunctionTok{pure}| values and \VERB|\NormalTok{lift}|ing all
pure functions into it.

In particular, since \VERB|\DataTypeTok{ConjoinedMonads}| is a special
case of \VERB|\DataTypeTok{MonadXApplicative}|, all
\VERB|\DataTypeTok{ConjoinedMonads}| instances from the previous
sections can be used as a basis for such a formalism.

While it is not immediately clear how to make imprecise exceptions into
an instance of \VERB|\DataTypeTok{MonadXApplicative}| (since they are
non-deterministic, hence disobeying the above laws, and
\VERB|\NormalTok{throw}| having a wrong type to be the identity element
for \VERB|\FunctionTok{catch}|, see
\cref{rem:io-two-kinds-of-exceptions}), there are some interesting
instances of \VERB|\DataTypeTok{MonadXApplicative}| that are not
\VERB|\DataTypeTok{ConjoinedMonads}|.

For instance, a folklore example of an \VERB|\DataTypeTok{Applicative}|
that is not a \VERB|\DataTypeTok{Monad}| is "computations collecting
failures in a \VERB|\DataTypeTok{Monoid}|", which can be defined as
follows

\begin{Shaded}
\begin{Highlighting}[]
\KeywordTok{newtype} \DataTypeTok{EA}\NormalTok{ e a }\FunctionTok{=} \DataTypeTok{EA}\NormalTok{ \{}\OtherTok{ runEA ::} \DataTypeTok{Either}\NormalTok{ e a \}}

\KeywordTok{instance} \DataTypeTok{Pointed}\NormalTok{ (}\DataTypeTok{EA}\NormalTok{ e) }\KeywordTok{where}
  \FunctionTok{pure} \FunctionTok{=} \DataTypeTok{EA} \FunctionTok{.} \DataTypeTok{Right}

\KeywordTok{instance} \DataTypeTok{Monoid}\NormalTok{ e }\OtherTok{=>} \DataTypeTok{Applicative}\NormalTok{ (}\DataTypeTok{EA}\NormalTok{ e) }\KeywordTok{where}
\NormalTok{  f }\FunctionTok{<*>}\NormalTok{ a }\FunctionTok{=} \DataTypeTok{EA} \FunctionTok{$}\NormalTok{ runEA f }\FunctionTok{<**>}\NormalTok{ runEA a }\KeywordTok{where}
\NormalTok{    (}\DataTypeTok{Right}\NormalTok{ f) }\FunctionTok{<**>}\NormalTok{ (}\DataTypeTok{Right}\NormalTok{ a) }\FunctionTok{=} \DataTypeTok{Right} \FunctionTok{$}\NormalTok{ f a}
\NormalTok{    (}\DataTypeTok{Right}\NormalTok{ f) }\FunctionTok{<**>}\NormalTok{ (}\DataTypeTok{Left}\NormalTok{  e) }\FunctionTok{=} \DataTypeTok{Left}\NormalTok{ e}
\NormalTok{    (}\DataTypeTok{Left}\NormalTok{  e) }\FunctionTok{<**>}\NormalTok{ (}\DataTypeTok{Right}\NormalTok{ a) }\FunctionTok{=} \DataTypeTok{Left}\NormalTok{ e}
\NormalTok{    (}\DataTypeTok{Left}\NormalTok{ e1) }\FunctionTok{<**>}\NormalTok{ (}\DataTypeTok{Left}\NormalTok{ e2) }\FunctionTok{=} \DataTypeTok{Left} \FunctionTok{$}\NormalTok{ e1 }\OtherTok{`mappend`}\NormalTok{ e2}
\end{Highlighting}
\end{Shaded}

Note, however, that this structure is a \VERB|\DataTypeTok{Monad}| in
\VERB|\NormalTok{e}|

\begin{Shaded}
\begin{Highlighting}[]
\OtherTok{throwEA ::}\NormalTok{ e }\OtherTok{->} \DataTypeTok{EA}\NormalTok{ e a}
\NormalTok{throwEA }\FunctionTok{=} \DataTypeTok{EA} \FunctionTok{.} \DataTypeTok{Left}

\OtherTok{catchEA ::} \DataTypeTok{EA}\NormalTok{ e a }\OtherTok{->}\NormalTok{ (e }\OtherTok{->} \DataTypeTok{EA}\NormalTok{ f a) }\OtherTok{->} \DataTypeTok{EA}\NormalTok{ f a}
\NormalTok{(}\DataTypeTok{EA}\NormalTok{ a) }\OtherTok{`catchEA`}\NormalTok{ f }\FunctionTok{=} \KeywordTok{case}\NormalTok{ a }\KeywordTok{of}
  \DataTypeTok{Right}\NormalTok{ a }\OtherTok{->} \FunctionTok{pure}\NormalTok{ a}
  \DataTypeTok{Left}\NormalTok{  e }\OtherTok{->}\NormalTok{ f e}
\end{Highlighting}
\end{Shaded}

\noindent which means it is also an instance of
\VERB|\DataTypeTok{MonadXApplicative}|. If we now remember that

\begin{itemize}
\item
  graded monads~\cite{Katsumata:2014:PEM} also require
  \VERB|\NormalTok{e}| to be a \VERB|\DataTypeTok{Monoid}| and
\item
  imprecise exceptions, too, can be though as producing a
  \VERB|\DataTypeTok{Monoid}| of possible errors with
  \VERB|\FunctionTok{catch}| (including the implicit
  \VERB|\FunctionTok{catch}| over \VERB|\NormalTok{main}|) "observing"
  one of its elements,
\end{itemize}

\noindent we come to a conclusion that in a calculus with
\VERB|\DataTypeTok{IO}|-effects separated from non-determinism-effects,
imprecise exceptions over non-deterministic
\VERB|\DataTypeTok{Applicative}| computations, indeed, form a
\VERB|\DataTypeTok{Monad}| (with equivalence defined up to raising the
same set of exceptions, similarly to section 4
of~\cite{PeytonJones:1999:SIE}) over the \VERB|\DataTypeTok{Monoid}| of
imprecise exceptions. That is, those, too, are examples of
\VERB|\DataTypeTok{MonadXApplicative}|.

\hypertarget{conclusions-and-future-work}{%
\section{Conclusions and future
work}\label{conclusions-and-future-work}}

We hope that with this article we pointed and then at least partially
plugged an algebraic hole in the programming languages theory by showing
that conventional computational formalisms with
\VERB|\NormalTok{throw}\FunctionTok{/}\NormalTok{try}\FunctionTok{/catch}|-exceptions
are "conjoined" products of pairs of \VERB|\DataTypeTok{Monad}|s (or,
less imperatively, \VERB|\DataTypeTok{Monad}|s and
\VERB|\DataTypeTok{Applicative}|s). This fact, in our opinion, makes a
lot of conventional programming "click into place" similarly to how
plain \VERB|\DataTypeTok{Monad}|s "click" imperative "semicolons".

Of particular note is the fact that everything in this paper, including
\VERB|\DataTypeTok{EIO}| of \cref{sec:instances:eio}, follows the
"marriage" framework of~\cite{wadler-thiemann-03} of confining effects
to monads, but ignores the question of any additional rules for type
indexes in question. In other words, ad-hoc exception encoding
constructions like that of error-explicit IO~\cite{Iborra:2010:ETE} or
graded monads~\cite{Katsumata:2014:PEM} are mostly orthogonal to our
"conjoined" structures and can be used simultaneously.

Besides practical applications described in the body of the paper and
observations already mentioned in \cref{sec:extabstract} (rereading said
section about now is highly recommended) we also want turn your
attention to the following observations.

\begin{enumerate}
\item
  Conventional error handling with \VERB|\NormalTok{throw}| and
  \VERB|\FunctionTok{catch}| (but without dynamic dispatch) is dual to
  conventional \VERB|\DataTypeTok{Monad}|ic sequential computation, a
  fact which, in our opinion, is interesting by itself (see
  footnote~\ref{fn:its-dual}).
\item
  Meanwhile, the "without dynamic dispatch" part above, in our opinion,
  provides an algebraic foundation for the argument against building new
  languages with builtin dynamic dispatch of exception handlers and/or
  an argument against extensively relying on that feature in the
  languages that have it, a point which is commonly discussed in the
  folklore ("exceptions are evil") and was articulated by Hoare from
  programmer comprehension standpoint already in
  1981~\cite{Hoare:1981:EOC}. Not only dynamic dispatch of exceptions
  is, citing Hoare, "dangerous", but it also prevents programs from
  directly accessing the inherent \VERB|\DataTypeTok{Monad}|ic
  structures discussed in this article.
\item
  We feel that the usual arguments against using
  \VERB|\DataTypeTok{Monad}|s for error handling are moot.

  \begin{itemize}
  \item
    The problem of syntactic non-uniformness between pure computations,
    \VERB|\DataTypeTok{Applicative}|s and \VERB|\DataTypeTok{Monad}|s,
    in our view, is almost trivial to solve: common primitives like
    \VERB|\FunctionTok{map}|/\VERB|\FunctionTok{mapM}| should be
    expressed in terms of \VERB|\DataTypeTok{Applicative}|s (of which
    pure functions are trivial instance) instead of
    \VERB|\DataTypeTok{Monad}|s. For instance, \VERB|\FunctionTok{mapM}|
    for list\footnote{And, similarly, for
      \VERB|\DataTypeTok{Traversable}| which we shall continue to ignore
      for the purposes of this article.} can be rewritten as

\begin{Shaded}
\begin{Highlighting}[]
\OtherTok{mapAp ::} \DataTypeTok{Applicative}\NormalTok{ f }\OtherTok{=>}\NormalTok{ (a }\OtherTok{->}\NormalTok{ f b) }\OtherTok{->}\NormalTok{ [a] }\OtherTok{->}\NormalTok{ f [b]}
\NormalTok{mapAp f     [] }\FunctionTok{=} \FunctionTok{pure}\NormalTok{ []}
\NormalTok{mapAp f (a}\FunctionTok{:}\NormalTok{as) }\FunctionTok{=} \FunctionTok{fmap}\NormalTok{ (}\FunctionTok{:}\NormalTok{) (f a) }\FunctionTok{<*>}\NormalTok{ mapAp f as}
\end{Highlighting}
\end{Shaded}

    Meanwhile, the uniform syntax for pure functions and
    \VERB|\DataTypeTok{Applicative}|s can be made by adding some more
    missing instances of the LISP macros into the compiler in
    question.\footnote{From a cynical LISP-evangelist point of view, all
      of "the progress" of the programming languages in the last 50
      years can be summarized as "adopting more and more elements
      (lately, meta-programming) from LISP while trying very hard not to
      adopt the syntax of LISP". From a less cynical perspective, "the
      progress", at least in typed languages, consists of well-typing
      said elements.} For instance,
    quasiquotation~\cite{Mainland:2007:WNQ} is one conventional way do
    such a translation, Conal Elliot's "Compiling to
    Categories"~\cite{Elliott:2017:CTC} provides another categorically
    cute way to achieve similar results.
  \item
    We feel that the problem of modularity as stated by
    Brady~\cite{Brady:2013:PRA}

    \begin{quote}
    Unfortunately, useful as monads are, they do not compose very well.
    Monad transformers can quickly become unwieldy when there are lots
    of effects to manage, leading to a temptation in larger programs to
    combine everything into one coarse-grained state and exception
    monad.
    \end{quote}

    can be solved by applying graded monads to the
    \VERB|\DataTypeTok{Monad}| part of
    \VERB|\DataTypeTok{MonadXApplicative}|.
  \end{itemize}
\end{enumerate}

In other words, we think that a programming language that

\begin{itemize}
\tightlist
\item
  provides a primitive \VERB|\FunctionTok{catch}| operator that does no
  dynamic dispatch,
\item
  provides quasi-quoting/compiling to categories for
  \VERB|\DataTypeTok{Applicative}|s,
\item
  distinguishes between \VERB|\DataTypeTok{IO}|-effects and
  non-determinism, and
\item
  uses a graded \VERB|\DataTypeTok{MonadXApplicative}| for a base type
  of computations
\end{itemize}

\noindent would provide all the efficiency of imprecise exceptions,
simplicity of \VERB|\DataTypeTok{Monad}|s (doubled, in some sense, since
error handling would stop being special), while having none of the usual
arguments against said mechanisms applying to it.

We feel that the following future work directions on the topic would be
of particular value:

\begin{itemize}
\item
  implementation of a practical "good-enough" (\cref{sec:encodings})
  library for GHC Haskell, and, eventually, an implementation of a
  dialect of Haskell with a graded
  \VERB|\DataTypeTok{MonadXApplicative}| as a base type of computations,
\item
  research into syntax and semantics of "marriages" between precise and
  imprecise exceptions in a single language, including, but not limited
  to, research into simpler semantic models for \(\lambda\)-calculus
  with Monads~\cite{wadler-thiemann-03,
   Filinski:1994:RM},
\item
  research into the question of whether multiplying more than two
  \VERB|\DataTypeTok{Monad}|s and \VERB|\DataTypeTok{Applicatives}| with
  non-trivial interaction laws produces interesting
  structures.\footnote{It is clear that one can have more than one index
    \VERB|\NormalTok{e}| conjoined to a single \VERB|\NormalTok{a}|, but
    such a construction doesn't seem to make much sense in presence of
    graded \VERB|\DataTypeTok{Monad}|s. However, that fact by itself
    does not exclude a possibility of existence of an interesting
    structure for which there are non-trivial interactions between
    different indexes \VERB|\NormalTok{e}|.}
\end{itemize}

All the practical results of this article except for
\VERB|\NormalTok{catchT}| combinator of
\cref{sec:instances:throw-catch-cc} were born in 2014 in a course of a
single week from observing the structure of a parser combinator
\VERB|\DataTypeTok{Monad}| indexed by errors and values (and other
things beyond the scope of this article, the original structure is also
an indexed \VERB|\DataTypeTok{State}| \VERB|\DataTypeTok{Monad}| to
allow parsing of arbitrary data types, not just streams) a very
simplified version of which was presented in
\cref{sec:parser-combinators,sec:instances:parser-combinators}. The
article itself was started in 2016 but then was rewritten from scratch
four times before finally settling to the current presentation. The
\VERB|\NormalTok{catchT}| combinator was discovered while writing
\cref{sec:continuations}.

This article would have been impossible without the patience of Sergei
Soloviev who read and meticulously commented numerous drafts of the
paper, numerous people who encouraged me to write this after I described
the general idea to them, and all contributors to Emacs and org-mode
without whom neither the planning nor the writing of the actual text
would have been manageable. The author is also grateful to
\fbox{Sergey Baranov} for helpful discussions on related topics which
steered the first half of this paper into its current form.

\printbibliography

\end{document}